\documentclass[12pt]{article}
\usepackage{amsmath}
\usepackage{graphicx}
\usepackage{enumerate}
\usepackage{natbib}
\usepackage{url} % not crucial - just used below for the URL 
\usepackage{authblk}

\newcommand{\blind}{1}

\usepackage{geometry}
\usepackage{amsmath}
\allowdisplaybreaks[4]
\usepackage{float}
\usepackage{mathrsfs}
\usepackage{amsfonts}   
\usepackage{xcolor}
\usepackage{xr}
\usepackage{times}
\usepackage{bm}
\usepackage{colonequals}
\usepackage[plain,noend]{algorithm2e}
\usepackage{graphicx}
\usepackage{caption}
\usepackage{enumerate}
\usepackage{cmll}
\usepackage{geometry}
\usepackage{amsthm}
\usepackage{subcaption}
\usepackage{threeparttable}
\usepackage{booktabs}
\usepackage{multirow}
\usepackage{makecell}
\usepackage{tikz}
\usetikzlibrary{positioning,arrows.meta,quotes}
\usetikzlibrary{shapes,snakes}
\usetikzlibrary{bayesnet}
\tikzset{>=latex}
\tikzstyle{plate caption} = [caption, node distance=0, inner sep=0pt,
below left=3pt and 0pt of #1.south]
\usepackage{hyperref}
\usepackage{enumitem}
\usepackage{titlesec}
\usepackage[nodisplayskipstretch]{setspace}

\makeatletter
\def\UrlAlphabet{%
	\do\a\do\b\do\c\do\d\do\e\do\f\do\g\do\h\do\i\do\j%
	\do\k\do\l\do\m\do\n\do\o\do\p\do\q\do\r\do\s\do\t%
	\do\u\do\v\do\w\do\x\do\y\do\z\do\A\do\B\do\C\do\D%
	\do\E\do\F\do\G\do\H\do\I\do\J\do\K\do\L\do\M\do\N%
	\do\O\do\P\do\Q\do\R\do\S\do\T\do\U\do\V\do\W\do\X%
	\do\Y\do\Z}
\def\UrlDigits{\do\1\do\2\do\3\do\4\do\5\do\6\do\7\do\8\do\9\do\0}
\g@addto@macro{\UrlBreaks}{\UrlOrds}
\g@addto@macro{\UrlBreaks}{\UrlAlphabet}
\g@addto@macro{\UrlBreaks}{\UrlDigits}

\titleformat{\section}
{\normalfont\fontsize{12.5}{12.5}\bfseries}{\thesection}{1em}{}
\titleformat{\subsection}
{\normalfont\fontsize{12}{12}\bfseries}{\thesubsection}{1em}{}

\titlespacing{\section}{0pt}{3pt}{3pt}
\titlespacing{\subsection}{0pt}{1pt}{1pt}
\titlespacing{\subsubsection}{0pt}{1pt}{1pt}

\newcommand{\nodisplayskips}{
	\setlength{\abovedisplayskip}{3pt}
	\setlength{\belowdisplayskip}{3pt}
	\setlength{\abovedisplayshortskip}{3pt}
	\setlength{\belowdisplayshortskip}{3pt}}

\hypersetup{hidelinks}
\graphicspath{{./figure/}}
\makeatletter

\newtheorem{theorem}{\bf{Theorem}}
\newtheorem{remark}{\bf{Remark}}
\newtheorem{lemma}{\bf{Lemma}}

\newtheorem{proposition}{\bf{Proposition}}
\newtheorem{condition}{\bf{Condition}}

\newtheorem{assumption}{\bf{Assumption}}

\newcommand{\nn}{\nonumber}

\newcommand{\scB}{\mathscr{B}}
\newcommand{\var}{{\rm var}}
\newcommand{\tr}{{{\rm tr}}}
\newcommand{\cov}{{{\rm cov}}}
\newcommand{\diag}{{\rm diag}}

\newcommand{\cA}{\mathcal{A}}

\newcommand{\bg}{\boldsymbol{g}}
\newcommand{\bG}{\boldsymbol{G}}
\newcommand{\bU}{\boldsymbol{U}}
\newcommand{\betax}{\beta_{X}}
\newcommand{\taup}{\tau_{\rm p}}
\newcommand{\bgamma}{\boldsymbol{\gamma}}
\newcommand{\bGamma}{\boldsymbol{\Gamma}}
\newcommand{\bbeta}{\boldsymbol{\beta}}
\newcommand{\balpha}{\boldsymbol{\alpha}}
\newcommand{\bxi}{\boldsymbol{\xi}}
\newcommand{\bzero}{\boldsymbol{0}}
\newcommand{\bw}{\boldsymbol{w}}
\newcommand{\bl}{\boldsymbol{l}}
\newcommand{\bbL}{\boldsymbol{L}}
\newcommand{\bbY}{\mathbf{Y}}
\newcommand{\bbX}{\mathbf{X}}
\newcommand{\bbG}{\mathbf{G}}

\newcommand{\bbPsi}{\mathbf{\Psi}}
\newcommand{\bbPhi}{\mathbf{\Phi}}
\newcommand{\bbS}{\mathbf{S}}
\newcommand{\bbV}{\mathbf{V}}
\newcommand{\bbD}{\mathbf{D}}
\newcommand{\bbQ}{\mathbf{Q}}
\newcommand{\bbH}{\mathbf{H}}
\newcommand{\bbR}{\mathbf{R}}
\newcommand{\bbI}{\mathbf{I}}
\newcommand{\bbA}{\mathbf{A}}

\newcommand{\bbM}{\mathbf{M}}
\newcommand{\bbOmega}{\mathbf{\Omega}}

\newcommand{\bbSig}{\mathbf{\Sigma}}
\newcommand{\bbP}{\mathbf{P}}

\newcommand{\hbeta}{\hat{\beta}}
\newcommand{\enbeta}{\hat{\beta}}
\newcommand{\hgamma}{\widehat{\boldsymbol{\gamma}}}
\newcommand{\tgamma}{\widetilde{\boldsymbol{\gamma}}}
\newcommand{\hGamma}{\widehat{\boldsymbol{\Gamma}}}
\newcommand{\htau}{\hat{\tau}}

\newcommand{\hD}{\widehat{\mathbf{D}}}
\newcommand{\hQ}{\widehat{\mathbf{Q}}}
\newcommand{\hR}{\widehat{\mathbf{R}}}
\newcommand{\hSig}{\widehat{\mathbf{\Sigma}}}
\newcommand{\hPsi}{\widehat{\mathbf{\Psi}}}
\newcommand{\hS}{\widehat{\mathbf{S}}}
\newcommand{\hs}{\widehat{s}}
\newcommand{\cD}{\mathcal{D}}

\def\T{{\mathrm{\scriptscriptstyle T} }}

\begin{document}
	\pagenumbering{gobble}

	\def\spacingset#1{\renewcommand{\baselinestretch}%
		{#1}\small\normalsize} \spacingset{1}

	%%%%%%%%%%%%%%%%%%%%%%%%%%%%%%%%%%%%%%%%%%%%%%%%%%%%%%%%%%%%%%%%%%%%%%%%%%%%%%
	
	\if1\blind
	{
		\title{\bf \fontsize{13}{13}\selectfont Debiased Estimating Equation Method for Robust and Efficient Mendelian Randomization Using a Large Number of Correlated Weak and Invalid Instruments
		}
		%DEEM: A versatile and efficient Mendelian randomization method that incorporates many correlated SNPs with weak effects}
	\author{\normalsize Ruoyu Wang, Haoyu Zhang, Xihong Lin
		%\author[a]{Ruoyu Wang}
		%\author[b]{Haoyu Zhang}
		%\author[a,c]{Xihong Lin}
		%\thanks{Corresponding author. Email: xlin@hsph.harvard.edu}}%\thanks{
		%The authors gratefully acknowledge \textit{please remember to list all %relevant funding sources in the unblinded version}}\hspace{.2cm}\\
	%\affil[a]{Department of Biostatistics, Harvard T.H. Chan School of Public Health}
	%\affil[b]{Division of Cancer Epidemiology and Genetics, National Cancer Institute}
	%\affil[c]{Department of Statistics, Harvard University}
	\thanks{Ruoyu Wang conducted this work while he was a postdoctoral fellow in the Department of Biostatistics at Harvard T.H. Chan School of Public Health. He is currently an assistant professor in the Department of Statistics and Data Science at Tsinghua University. Haoyu Zhang is an Earl Stadtman tenure-track investigator in the Division of Cancer Epidemiology and Genetics at the National Cancer Institute. Xihong Lin is a professor of biostatistics at Harvard T.H. Chan School of Public Health and a professor of statistics at Harvard University. Correspondence should be addressed to Xihong Lin ({\em xihong\_lin@hsph.harvard.edu}). The authors thank the reviewers for their very helpful comments that have significantly improved the paper. This work was supported by the National Institutes of Health grants R35-CA197449, R01-HL163560, U01-HG012064, U19-CA203654,  and P30 ES000002.}}
\date{}
\maketitle
} \fi

\if0\blind
{
\bigskip
\bigskip
\bigskip
\begin{center}
	{\bf \fontsize{13}{13}\selectfont Debiased Estimating Equation Method for Robust and Efficient Mendelian Randomization Using a Large Number of Correlated Weak and Invalid Instruments}
	%with Weak Effects}
%   DEEM: A versatile and efficient Mendelian randomization method that incorporates many correlated SNPs with weak effects}
\end{center}
\medskip
} \fi

\bigskip
\begin{abstract}
Mendelian randomization (MR) is a widely used tool for causal inference in the presence of unmeasured confounders, which uses single nucleotide polymorphisms (SNPs) as instrumental variables to estimate causal effects. However,  SNPs often have weak effects on complex traits, leading to bias in existing MR analyses when weak instruments are included. In addition, existing MR methods often restrict analysis to independent SNPs via linkage disequilibrium clumping and result in a loss of efficiency in estimating the causal effect due to discarding correlated SNPs. To address these issues, we propose the Debiased Estimating Equation Method (DEEM), a summary statistics-based MR approach that can incorporate a large number of correlated, weak-effect, and invalid SNPs. DEEM effectively eliminates the weak instrument bias and improves the statistical efficiency of the causal effect estimation by leveraging information from a large number of correlated SNPs. DEEM also allows for pleiotropic effects, adjusts for the winner's curse, and applies to both two-sample and one-sample MR analyses. Asymptotic analyses of the DEEM estimator demonstrate its attractive theoretical properties. Through extensive simulations and two real data examples, we demonstrate that DEEM significantly improves the efficiency and robustness of MR analysis compared with existing methods.
\end{abstract}

\noindent%
{\it Keywords:}  Causal inference, Genome-wide association studies, Linkage disequilibrium, Statistical efficiency, Weak instrumental variables.
\vfill

\newpage
\pagenumbering{arabic}
\spacingset{1.7} % DON'T change the spacing!
\section{Introduction}\label{sec: intro}
Mendelian randomization (MR) is a widely-used statistical method in epidemiological and medical research for deducing causal relationships between exposures and outcomes in the presence of unmeasured confounders \citep{sanderson2022mendelian}. It leverages %genetic variants, predominantly 
Single Nucleotide Polymorphisms (SNPs) as instrumental variables (IVs) to assess causal effects, circumventing the biases associated with unobserved confounders in observational studies.
%	The underlying principle of MR, Mendelian inheritance, assumes that genetic variants are randomly assigned at conception and thus are not prone to confounding and reverse causation biases that often affect traditional observational studies \citep{burgess2013mendelian}. 
In recent years, the availability of the data of large-scale genome-wide association studies (GWAS), especially publicly available GWAS summary statistics, has increased the 
popularity of MR as a powerful tool for causal inference accounting for unmeasured confounders \citep{burgess2019guidelines,zhao2020statistical}.
%,sanderson2022mendelian}. 
Standard MR methods typically use independent and/or strong-effect SNPs as IVs, which can be subject to a loss of efficiency and vulnerable to weak and invalid IVs.
We propose in this paper a GWAS summary statistics-based debiased estimating equation method (DEEM) for robust and statistically efficient MR analysis using a large number of correlated weak and invalid instruments.

Despite significant recent attention, the origin of MR methods can be traced back to nearly century-old IV methods. IV methods have been extensively studied and have yielded fruitful results in economics and statistics \citep{anderson1950asymptotic, newey2009generalized}. Classical IV methods typically require individual-level data and a small number of IVs. However, MR analysis encounters several %distinct 
challenges. First, privacy concerns often prevent the public availability of individual-level GWAS data, hindering the use of traditional IV methods. 
Second, a vast majority of GWAS SNPs have small or no effects and some of them may exhibit pleiotropic effects, thereby making them weak and invalid IVs \citep{sanderson2022mendelian}.
Third, GWAS SNPs  can be  correlated due to linkage disequilibrium (LD), which poses additional statistical hurdles.  
%limited data accessibility, 
%the presence of commonly invalid and weak IVs in GWAS, and the linkage disequilibrium (LD) among SNPs, which 

Many MR methods have been developed in the last ten years  
%spurs many recent methodological advancements 
\citep{bowden2015mendelian, burgess2016combining, zhao2020statistical, ye2021debiased}.
%In genetic research, privacy concerns limit the acquisition of individual-level data, which impedes the use of traditional IV methods that need individual-level data \citep{anderson1950asymptotic, newey2009generalized}. This situation has contributed to 
Although individual-level GWAS data are often not publicly available, GWAS summary statistics are widely available. This has led to the development of a variety of GWAS summary statistics-based MR methods.
 Among them, the most popular MR method is the IVW method \citep{burgess2013mendelian}. %has emerged as the most popular one \citep{ye2021debiased}.
A main limitation of the IVW method is that it is vulnerable to weak and invalid IVs. Specifically, due to the weak IV problem \citep{bound1995problems}, the IVW estimator can be biased when SNPs with small effects are included in the analysis \citep{ye2021debiased}. Moreover, a SNP is an invalid IV if it has no effect on the exposure or, in the case of horizontal pleiotropy, affects the outcome through other biological pathways besides the exposure. Invalid IVs can result in the IVW estimator being biased \citep{bowden2016consistent}. 

The weak IV problem can be avoided by only including SNPs with strong effects. However, this selection strategy overlooks valuable information from many SNPs with weak or moderate effects. Genome-wide significant SNPs typically account for only a small portion of the variance in complex traits, challenging the statistical efficiency of MR methods that rely exclusively on such variants \citep{brion2013calculating}.
Furthermore, a SNP with a disproportionately large effect on the exposure may be a pleiotropic outlier that violates the IV assumptions \citep{zhu2018causal}. Such SNPs can introduce substantial bias to MR estimators that rely solely on SNPs with strong effects.

Several summary statistics-based methods have emerged to tackle challenges associated with invalid or weak IVs, primarily under the common two-sample design.   Specifically, MR-Egger \citep{bowden2015mendelian} adapts Egger regression from the meta-analysis literature \citep{egger1997bias} to adjust for the bias caused by invalid IVs. The weighted median method \citep{bowden2016consistent} and the weighted mode method \citep{hartwig2017robust} mitigate the influence of invalid IVs by utilizing the median and mode of causal effect estimates obtained from different SNPs, respectively. CIIV \citep{windmeijer2021confidence} and
cML \citep{xue2021constrained} either explicitly or implicitly compare causal effect estimates obtained from different SNPs and exclude SNPs that yield causal effect estimates deviating from the consensus among the plurality of SNPs. These methods are robust to invalid IVs under specific assumptions. However, they can be biased in the presence of weak IVs. 

RAPS \citep{zhao2020statistical}, dIVW \citep{ye2021debiased}, and pIVW \citep{xu2022novel} can accommodate invalid IVs while effectively eliminating the weak IV bias using corrections motivated by asymptotic analysis. However, these three methods employ a normal assumption on the size of pleiotropic effects to account for invalid IVs, assuming all variants are causal. This assumption is restrictive in practice, as different traits have different genetic architectures, with some being dense and some being sparse across the genome.
Moreover, while the above methods advance MR analysis, they generally assume IVs are uncorrelated, an assumption often violated due to LD among SNPs. Although LD clumping can be used to select a subset of uncorrelated SNPs, it can lead to a substantial loss in robustness and efficiency, as we will demonstrate. For instance, when focusing on variants in a single gene region with a known biological connection to the exposure  \citep{xue2023causal}, LD clumping could limit the selection to only a single SNP, which makes the analysis results highly sensitive to the validity of the selected SNP and undermines the effectiveness of methods that require multiple independent IVs such as CIIV, RAPS, dIVW, and pIVW. 

Several existing Bayesian methods can potentially incorporate correlated SNPs \citep{shapland2019bayesian,morrison2020mendelian,cheng2022mendelian}.  However, these methods rely on parametric prior distributions of the effect sizes, which can be biased if the prior distributions are misspecified, and their theoretical underpinnings have not been thoroughly validated. Meanwhile, approaches like
weighted generalized linear regression (WLR) \citep{burgess2016combining}, PCA-IVW \citep{burgess2017mendelian}, and 2ScML \citep{xue2023causal}, can handle correlated SNPs without relying on prior distributions, but face challenges with weak IVs.   GENIUS-MAWII \citep{ye2024genius} can incorporate correlated weak and invalid IVs. However, it requires access to individual-level data and  the conditional distribution of the exposure given the genotype to be heteroscedastic to identify the causal effect. The existing literature underscores a pressing need to develop robust and efficient summary data-based methods that can effectively handle correlated, weak-effect, and invalid SNPs. 
%To our knowledge, no MR method is provably effective in the presence of highly correlated invalid and weak IVs.

In this paper, to address limitations of  existing methods, we propose DEEM, which incorporates correlated, weak-effect and potentially invalid SNPs  into  MR analysis. Our work contributes to advancing the field of MR in several key ways. First, DEEM offers more  robust and efficient estimators than existing   methods, which typically rely on independent and/or strong-effect SNPs, by incorporating a large number of correlated, weak-effect and potentially invalid SNPs as IVs.  

Second, the proposed approach formulates MR analysis within a new general framework of summary statistics-based estimating equations (EEs). Unlike traditional EEs, which typically rely on individual-level data \citep{mccullagh1989generalized}, our summary statistics-based EEs enable flexible construction of MR estimators using summary statistics. 
The proposed general summary statistics-based EE framework facilitates a rigorous study of the theoretical properties of the DEEM estimator and several existing estimators.
%allows for a transparent statistical investigation of the estimator's properties.}
%Our statistical investigation reveals 
We show that several existing MR methods, such as IVW, WLR, and PCA-IVW, correspond to specific forms of EEs based on summary statistics, yet suffer from biases due to non-zero expectations at the true causal effect values
%, particularly 
when SNP effects are weak. These biases stem from the fact that the traditional individual-level data based EE framework cannot be directly applied to construct EEs using summary statistics. 

Third, to address the bias in existing methods and the challenges in applying traditional EEs to summary statistics-based MR analysis, we introduce DEEM that is based on summary statistics while yielding asymptotically unbiased estimators.  
 Specifically, DEEM uses a new decorrelation technique, enabling the inclusion of a large number of correlated SNPs without yielding weak IV bias, thereby enhancing the efficiency of the resulting causal effect estimator. Notably, through a diagonalization procedure, DEEM overcomes the difficulty in estimating large covariance matrices in the decorrelation step under the high-dimensional setting, making it particularly appealing with a large number of SNPs. This technique also enables DEEM to avoid the common assumption that eigenvalues of IVs' covariance matrix must be bounded away from zero \citep{newey2009generalized, kang2016instrumental, guo2018confidence, ye2024genius}—a condition often violated in the presence of highly correlated IVs. This advance widens DEEM's applicability compared to traditional methods. To our knowledge, the proposed diagonalized decorrelation procedure is new not only in summary statistics-based MR but also in general debiased EE context. We also extend DEEM  to accommodate invalid IVs with horizontal pleiotropic effects under mild conditions.

Fourth, we propose an ensemble estimator that leverages supplemental samples to further improve the efficiency of the DEEM estimator. Specifically, 
%Leveraging insights from the recent literature 
we use a supplemental exposure sample for SNP selection to mitigate the winner's curse \citep{zhao2019powerful,zhao2020statistical}.   In addition, we construct another debiased EE that incorporates supplemental data to further improve the efficiency of the DEEM estimator.  The ensemble estimator combines the solutions to two debiased EEs via a weighted average. We establish the asymptotic normality of the ensemble estimator and derive the optimal weighting that minimizes its asymptotic variance. 
%This ensemble approach offers a new strategy to exploit information in the supplemental exposure sample in summary statistics-based MR.} 
%\textcolor{red}{HZ: This paragraph might be too long. Many technical details in one paragraph. I tried to simplify a little bit.}
%In summary, the proposed summary statistics-based MR method DEEM  presents a suite of desirable features. Foremost, DEEM effectively integrates a wide array of correlated SNPs, making full use of the data with both   correlated SNPs and weak-effect SNPs, thereby significantly improving statistical efficiency over existing methods. Second, DEEM allows for potentially invalid IVs under a weaker assumption on pleiotropic effects compared to existing methods. %\citep{zhao2020statistical, ye2021debiased, xu2022novel}. Third, DEEM leverages information from supplemental exposure data via a specialized debiased EE,  further increasing its efficiency and robustness compared to existing methods. 
We also showcase attractive flexibility of DEEM by extending its applicability from two-sample settings to one-sample settings, a domain where summary statistics-based MR methods often struggle with weak IVs. Figure \ref{fig: workflow} summarizes the settings, the required data inputs, and the estimation workflow accommodated by DEEM.
\begin{figure}
	\centering
	\includegraphics[scale=0.5]{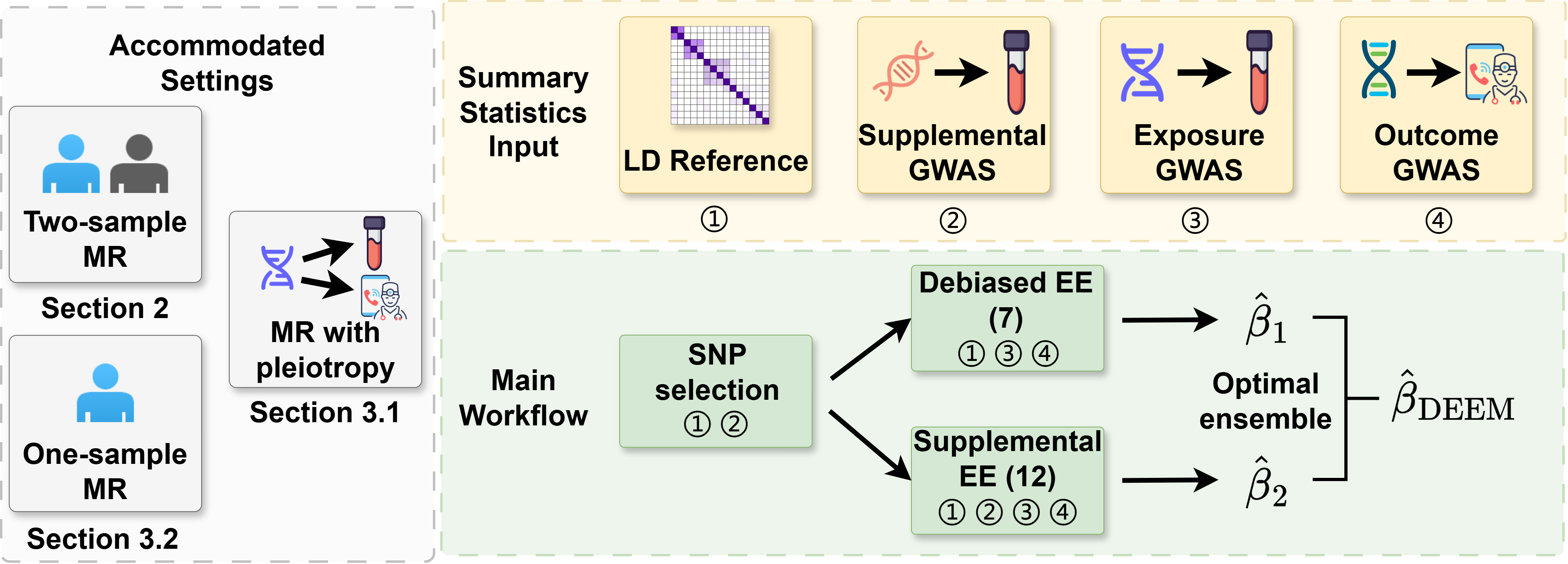}
	\caption{Overview of the  settings, inputs, and estimation workflow of DEEM. Circled numbers in the workflow identify the inputs used for SNP selection and for constructing the debiased and supplemental EEs.}\label{fig: workflow}
\end{figure}

%that invalidate the IV assumptions. 

%With the increasing availability of summary statistics from large population-based biobanks \citep{minelli2021use}, MR methods are now being applied in the one-sample setting where SNP-exposure and SNP-outcome associations are estimated from the same sample. As a versatile EE-based method, DEEM can be extended to the one-sample setting. This extension results in a weak IV robust MR method suitable for the one-sample setting. To our knowledge, few summary statistics-based MR methods are provably robust to weak IVs under the one-sample setting. 

We conduct extensive simulation studies to compare DEEM's performance with existing methods across diverse settings, which demonstrate that DEEM considerably improves efficiency and robustness over existing methods. We apply DEEM to real-world datasets to study the causal effect between low-density lipoprotein cholesterol (LDL-C) and coronary artery disease (CAD) risk, and between body mass index (BMI) and systolic blood pressure (SBP). The findings highlight DEEM's notable improvements in robustness and efficiency in these practical applications.

The rest of this paper is organized as follows. In Section \ref{sec: DEEM}, we discuss the summary statistics-based MR methods in the EE framework, and propose DEEM in the two-sample setting. In Section \ref{sec: extension}, we extend  DEEM to accommodate pleiotropic effects and the one-sample MR. We report extensive simulation results and two real data analyses in Sections \ref{sec: sim} and \ref{sec: real data}, followed by discussions. Additional simulations, real data analyses, all proofs, packages and data availability statements are provided in the Appendix.

\section{The Debiased Estimating Equation Method}\label{sec: DEEM}
\subsection{Set-up}\label{subsec: set up}

Consider the objective of estimating the causal effect of an exposure $X$ on an outcome $Y$ in the presence of unmeasured confounders $\bU$. We adopt the model:
\begin{equation}
\label{eq: outcome model}
\nodisplayskips
Y = \beta_{0} + X\betax  + f(\bU)  + \epsilon_{Y},
\end{equation}
where $\betax$ is the causal effect of interest, $\bU$ is a vector of unmeasured confounders that are correlated with $X$, $f(\cdot)$ is an unknown function, and $\epsilon_{Y}$ is a mean zero error term independent of $X$ and $\bU$. While Model (\ref{eq: outcome model}) can be extended to include measured covariates, we focus on Model (\ref{eq: outcome model}) for clarity. An extension of the proposed method to allow for measured covariates is straightforward (see Appendix Section \ref{app: adjust for covariate}). 
In this scenario, traditional least squares regression of $Y$ on $X$, disregarding $\bU$, typically fails to consistently estimate $\beta_{X}$ due to the confounding effect of $\boldsymbol{U}$ on both $X$ and $Y$. MR offers a solution to this problem by employing SNPs as IVs.

Consider $\bG$ as a $d$-dimensional vector of genotypes (SNPs). Suppose
$$X = \alpha_{0} + \bG^{\T}\balpha_{G} + g(\bU) + \epsilon_{X},$$
where $g(\cdot)$ is an unknown function, and $\epsilon_{X}$ is an error term with mean zero, assumed independent of $\bG$, $\bU$, and $\epsilon_{Y}$.
%We assume $\bG$ satisfies 
The standard IV assumptions include: (i) $\balpha_G\neq \bzero$, ensuring the SNPs are associated with the exposure; (ii) $\bG\Perp \bU$, meaning the SNPs are independent of the unmeasured confounders; and (iii) exclusion restriction, asserting that $\bG$ affects the outcome $Y$ only through $X$ and not directly \citep{burgess2013mendelian}. Potential violations of these IV assumptions particularly due to pleiotropic effects
%, are critical and 
will be discussed in Section \ref{subsec: DEEM p}.

\begin{figure}[h]
	\centering
	\begin{tikzpicture}[scale = 0.5]
		\node [circle, draw=black, fill=white, inner sep=3pt, minimum size=0.5cm] (x) at (0,0) {\large $X$};
		\node [circle,draw=black,fill=white,inner sep=3pt,minimum size=0.5cm] (z) at (-4, 0) { $G$};
		\node [circle,draw=black,fill=white,inner sep=3pt,minimum size=0.5cm] (y) at (4,0) {\large $Y$};
		\node [obs, minimum size=0.7cm] (u) at (2,2.5) {\large $U$};
		\path [draw,->] (z) edge (x);
		\path [draw,->,dashed] (z) edge node[ anchor=center, pos=0.5,font=\bfseries]{$\times$} (u);
		\path [draw,->,dashed] (z) edge [out=-30, in=-150] node[ anchor=center, pos=0.5,font=\bfseries]{$\times$} (y);
		\path [draw,->] (x) edge node[ anchor=center, above, pos=0.5,font=\bfseries] {$\betax$} (y);
		\path [draw,->] (u) edge (x);
		\path [draw,->] (u) edge (y);
	\end{tikzpicture}
	\caption{An illustration of the IV assumptions.}
\end{figure}

MR analysis infers the causal effect $\betax$ in the presence of unmeasured confounders $\bU$ by utilizing publicly available GWAS marginal regression coefficients of $X$ and $Y$ on $\bG$. Specifically, for $j = 1,\dots, d$, let $G_{j}$ be the $j$th component of $\bG$, and $\widehat{\gamma}_{j}$ and $\widehat{\Gamma}_{j}$ the regression coefficients of $G_{j}$ under the marginal working models 
%$X = \gamma_{0j}^{\dagger} + G_{j}\gamma_{j}^{\dagger} + \epsilon_{X,j}^{\dagger}$ and $Y = \Gamma_{0j}^{\dagger} + G_{j}\Gamma_{j}^{\dagger} + \epsilon_{Y,j}^{\dagger}$, respectively.
 $X = \gamma_{0j} + G_{j}\gamma_{j} + \epsilon_{X,j}$ and $Y = \Gamma_{0j} + G_{j}\Gamma_{j} + \epsilon_{Y,j}$, respectively.
According to the law of large numbers, $\widehat{\gamma}_{j}$ and $\widehat{\Gamma}_{j}$ converge in probability to
$\gamma_{j} = \var(G_{j})^{-1}\cov(G_{j}, X)$
and
$\Gamma_{j} = \var(G_{j})^{-1}\cov(G_{j}, Y)$, 
respectively. Under model \eqref{eq: outcome model} and the above standard IV assumptions, one can easily show that $\Gamma_{j} = \betax \gamma_{j}$ and hence
\begin{equation}\label{eq: population equation}
\nodisplayskips
\bGamma = \betax \bgamma,
\end{equation}
 where $\bgamma = (\gamma_{1},\dots,\gamma_{d})^{\T}$ and $\bGamma = (\Gamma_{1},\dots,\Gamma_{d})^{\T}$.
Equation \eqref{eq: population equation} connects the causal effect $\betax$ and the marginal regression coefficients, thereby facilitating the estimation of $\betax$. 

In MR analysis, we infer the causal effect $\betax$ utilizing the relationship \eqref{eq: population equation} and the estimated marginal regression coefficients $\hgamma = (\widehat{\gamma}_{1}, \dots, \widehat{\gamma}_{d})^{\T}$ and $\hGamma =  (\widehat{\Gamma}_{1}, \dots, \widehat{\Gamma}_{d})^{\T}$ from publicly available GWAS summary statistics.   Our focus in this section is on a two-sample design where these coefficients are derived from two independent cohorts, which implies that $\hgamma$ and $\hGamma$ are independent. On the other hand, due to the LD among SNPs, individual components within $\hgamma$ and $\hGamma$ are often correlated.   We assume the following normal models: 
\begin{equation}\label{eq: normal model}
\nodisplayskips
\hgamma - \bgamma \sim N(\bzero, \bbSig_{\gamma})\ \ \text{and} \ \ \hGamma - \bGamma \sim N(\bzero, \bbSig_{\Gamma}),
\end{equation}
where $\bbSig_{\gamma}$ and $\bbSig_{\Gamma}$ represent the covariance matrices of $\hgamma$ and $\hGamma$. The normality assumption can be justified by the large sample size typical of modern GWAS and the central limit theory. 

When the covariance matrices $\bbSig_{\gamma}$ and  $\bbSig_{\Gamma}$ are diagonal, the jointly normal assumption \eqref{eq: normal model} reduces to the widely adopted independent bivariate normal model \citep{burgess2013mendelian,zhao2019powerful,zhao2020statistical,ye2021debiased}. The independent bivariate normal model essentially assumes the independence of SNPs used in the MR analysis \citep{zhao2020statistical}. In contrast, by allowing for non-diagonal $\bbSig_{\gamma}$ and  $\bbSig_{\Gamma}$, model \eqref{eq: normal model} can accommodate correlated SNPs, which can increase the statistical power of MR analyses by making full use of the available genetic information.

The relationship \eqref{eq: population equation} implies 
$\hGamma - \betax \hgamma$ has zero mean. In scenarios with a single IV, $\beta_X$ could be estimated by solving the equation $\widehat{\Gamma} - \beta \widehat{\gamma} = 0$, which  yields the estimator $\widehat{\Gamma} / \widehat{\gamma}$. However, MR analysis often involves numerous SNPs as candidate IVs, and hence the equation represents an overdetermined system with $d$ equations but only one unknown parameter $\beta_X$, which usually has no solution. Thus, we can combine the components of the estimating function $\hGamma - \beta \hgamma$ into a single EE $\boldsymbol{a}^{\T}(\hGamma - \beta \hgamma) = 0$ via a linear combination, where $\boldsymbol{a}$ is a $d$-dimensional vector of combination coefficients. Analogous to optimal combination coefficients in equation (1) of \cite{qu2000improving} and the conditional score in \cite{zhao2019powerful}, the combination coefficient vector $\boldsymbol{a}=\bbSig^{-1}\bgamma$ is optimal in the sense that it minimizes the asymptotic variance of the resulting estimator $\hat{\beta}$ over different choices of $\boldsymbol{a}$, where $\bbSig = \cov(\hGamma - \betax \hgamma) = \bbSig_{\Gamma} + \betax^{2}\bbSig_{\gamma}$. This yields the optimal EE 
\begin{equation}\label{eq: EE opt}
    \bgamma^{\T}\bbSig^{-1}(\hGamma - \beta \hgamma) = 0. 
 \end{equation}   
    Under regularity conditions, the solution to the optimal EE is both consistent and asymptotically normal with an asymptotic variance of $(\bgamma^{\T}\bbSig^{-1}\bgamma)^{-1}$. 

The optimal EE is useful in theoretical analysis but infeasible because $\bgamma$  and $\bbSig$ are unknown. One can estimate $\bbSig$ and $\bgamma$ to obtain a practically feasible EE. The matrix $\bbSig$ is high-dimensional and may be nearly singular \citep{zhu2017bayesian}, which makes the estimation of $\bbSig^{-1}$ problematic. Note that when $\bbSig$ is misspecified, the estimator of $\beta_X$ that solves the EE $\bgamma^{\T}\bbSig^{-1}(\hGamma - \beta \hgamma) = 0$ remains consistent \citep{mccullagh1989generalized}. Thus, we approximate $\bbSig$ by a possibly inconsistent data-dependent positive definite working covariance matrix $\bbV$. 
   \begin{remark}\label{remark: impact V}
  		The choice of $\bbV$ affects the efficiency of the resulting estimator. Under regularity conditions, the solution to $\bgamma^{\T}\bbV^{-1}(\hGamma - \beta \hgamma) = 0$ has the asymptotic variance $\bgamma^{\T}\bbV^{-1}\bbSig\bbV^{-1}\bgamma/(\bgamma^{\T}\bbV^{-1}\bgamma)^{2}$ and the relative efficiency $\cos^{2}(\bbSig^{-1/2}\bgamma, \bbSig^{1/2}\bbV^{-1}\bgamma)$ with respect to the solution of the optimal EE, where $\cos(\boldsymbol{a}_{1}, \boldsymbol{a}_{2}) = \boldsymbol{a}_{1}^{\T}\boldsymbol{a}_{2}/(\|\boldsymbol{a}_{1}\|\|\boldsymbol{a}_{2}\|)$ is the cosine similarity between two vectors $\boldsymbol{a}_{1}$ and $\boldsymbol{a}_{2}$ and $\|\cdot\|$ is the Euclidean norm. Thus, the efficiency loss compared to the optimal EE is determined by the discrepancy between the directions of the vectors $\bbSig^{-1/2}\bgamma$ and $\bbSig^{1/2}\bbV^{-1}\bgamma$. As a general principle, one needs to specify $\bbV$ to make the above discrepancy as small as possible using side information or domain knowledge. See Section \ref{subsec: LD & V} for a specific construction of $\bbV$.
  	\end{remark}

To motivate the proposed method, we treat $\bbV$ as fixed  when deriving the moment conditions in Sections \ref{subsec: set up}-\ref{subsec: orDEEM} and \ref{subsec: DEEM p}-\ref{subsec: DEEM os} for simplicity. The working matrix $\bbV$ is allowed to be data-dependent under suitable convergence conditions (Condition \ref{cond: matrices origin} in the Appendix) in our subsequent theoretical results. Replacing $\bgamma^{\T}\bbSig^{-1}$ in \eqref{eq: EE opt} by $\hgamma^{\T}\bbV^{-1}$, we obtain a practically feasible EE 
\begin{equation}\label{eq: plug-in EE}
\nodisplayskips
\hgamma^{\T}\bbV^{-1}(\hGamma - \beta\hgamma) = 0.
\end{equation}

Despite distinct original motivations, many existing MR estimators, such as the two-stage least-squares (TSLS) estimator \citep{hayashi2011econometrics}, the IVW estimator \citep{burgess2013mendelian}, the WLR estimator \citep{burgess2016combining}, and the PCA-IVW estimator \citep{burgess2017mendelian}, can be unified as the solutions of Equation \eqref{eq: plug-in EE} with some specific $\mathbf{V}$. The connection between these estimators and Equation \eqref{eq: plug-in EE} is further elaborated in Appendix Section \ref{app: proof of sec setup}. The solution $\hbeta_{\rm PlugIn} = \hgamma^{\T}\bbV^{-1}\hGamma / \hgamma^{\T}\bbV^{-1}\hgamma$ of equation \eqref{eq: plug-in EE} is asymptotically normal and exhibits the same asymptotic variance as the solution of the oracle optimal EE when $\|\bgamma\|$ is large and $\bbV^{-1}$ is a consistent estimator for $\bbSig^{-1}$. Please refer to Appendix Section \ref{app: weak IV} for more details. 
However, $\hbeta_{\rm PlugIn}$ is susceptible to the weak IV bias when SNP effects are small. Note that
	\begin{equation}\label{eq: weak IV bias}
		\nodisplayskips
		\begin{aligned}
			E\{\hgamma^{\T}\bbV^{-1}(\hGamma - \betax \hgamma)\} 
			& = \tr\{\bbV^{-1}\cov(\hgamma, \hGamma - \betax \hgamma)\} 
			= -\betax\tr\{\bbV^{-1}\bbSig_{\gamma}\} \neq 0, 
		\end{aligned}
	\end{equation}
when $\betax\neq 0$, which introduces a non-negligible bias to $\hbeta_{\rm PlugIn}$ if $\|\bgamma\|^{2}/\tr\{\bbSig_{\gamma}\}\not\to \infty$ as shown in Appendix Section \ref{app: weak IV}. 

\subsection{The Debiased Estimating Equation}\label{subsec: orDEEM}

To mitigate the above issue, instead of the EE \eqref{eq: plug-in EE}, we propose the following debiased EE,
\begin{equation}\label{eq: orDEEM}
	\nodisplayskips
	\{\hgamma - \hQ(\beta)(\hGamma - \beta \hgamma)\}^{\T} \bbV^{-1}(\hGamma - \beta \hgamma) = 0,
\end{equation}
where 
$\hQ(\beta) = - \beta \hD_{\gamma}(\hD_{\Gamma} + \beta^{2} \hD_{\gamma})^{-1}$, 
$\hD_{\gamma} = 
\diag\left\{[\bbV^{-1}\hSig_{\gamma}]_{11}/[\bbV^{-1}]_{11}, \dots, [\bbV^{-1}\hSig_{\gamma}]_{dd}/[\bbV^{-1}]_{dd}\right\}$, $\hD_{\Gamma} = 
\diag\left\{[\bbV^{-1}\hSig_{\Gamma}]_{11}/[\bbV^{-1}]_{11}, \dots, [\bbV^{-1}\hSig_{\Gamma}]_{dd}/[\bbV^{-1}]_{dd}\right\}$, $\hSig_{\gamma}$ and $\hSig_{\Gamma}$ are the summary statistics-based estimators for $\bbSig_{\gamma}$ and $\bbSig_{\Gamma}$. 
We delve below into the rationale behind equation \eqref{eq: orDEEM}.
%in Section \ref{subsec: orDEEM}.  
%{\red It would be easier for the reader if $ \hQ(\beta)$ can be defined here by moving the definition in the last paragraph on page 11 to here using $\widehat\Sigma$'s and $\widehat D$'s.} {\re I moved the definition of $\hQ$ here.} {\red The following is a bit hard to follow. I tried to simplify the presentation. Please see whether you can simplify the presentation more.} {\re I tried to simplify the following a bit.}
%We designate an EE as \emph{unbiased} if it has mean zero at the true value $\betax$ of the causal effect.

To motivate \eqref{eq: orDEEM}, we start examining equation \eqref{eq: weak IV bias}. A close look at equation \eqref{eq: weak IV bias} reveals that the weak IV problem is rooted in the correlation between $\hgamma$ and $\hGamma - \betax \hgamma$. This motivates us to replace $\hgamma$ with an estimator for $\bgamma$ that is uncorrelated to $\hGamma - \betax \hgamma$. 
For any square matrix $\bbQ^{\star}$, the estimator $\hgamma - \bbQ^{\star}(\hGamma - \betax \hgamma)$ is unbiased for $\bgamma$ and satisfies 
\[
\cov\left\{\hgamma - \bbQ^{\star}(\hGamma - \betax \hgamma), \hGamma - \betax \hgamma\right\} = -\betax \bbSig_{\gamma} - \bbQ^{\star}(\bbSig_{\Gamma} + \betax^{2}\bbSig_{\gamma}),
\]
which equals zero when $\bbQ^{\star} = - \betax \bbSig_{\gamma}(\bbSig_{\Gamma} + \betax^{2}\bbSig_{\gamma})^{-1}$. Define the decorrelation matrix $\bbQ^{\star}(\beta) =  - \beta \bbSig_{\gamma}(\bbSig_{\Gamma} + \beta^{2}\bbSig_{\gamma})^{-1}$ for any $\beta$.  Then, $\hgamma - \bbQ^{\star}(\betax)(\hGamma - \betax \hgamma)$ is uncorrelated to $\hGamma -\betax \hgamma$
and 
\begin{equation}\label{eq: full DEEM}
\nodisplayskips
\{\hgamma - \bbQ^{\star}(\beta)(\hGamma - \beta \hgamma)\}^{\T}\bbV^{-1}(\hGamma - \beta \hgamma) = 0
\end{equation}
is an unbiased EE for $\beta$. 
Although motivated differently, the profile score in \cite{zhao2020statistical}, which assumes SNPs are independent, can be interpreted as a special case of the EE \eqref{eq: full DEEM} when $\bbSig_{\gamma}$ and $\bbSig_{\Gamma}$ are diagonal and $\bbV = \bbSig_{\Gamma} + \beta^{2}\bbSig_{\gamma}$.
The matrix $\bbQ^{\star}(\beta)$ contains unknown matrices $\bbSig_{\gamma}$ and $\bbSig_{\Gamma}$ which are estimable based on summary statistics \citep{zhu2017bayesian}. The estimation error for 
$\bbSig_{\gamma}$ and $\bbSig_{\Gamma}$	is overlooked in \cite{zhao2020statistical}, a reasonable simplification considering that these matrices are diagonal and the estimation errors for diagonal matrices are typically small \citep{cai2016estimating}.  However, in this paper, we deal with non-diagonal $\bbSig_{\gamma}$ and $\bbSig_{\Gamma}$, where the corresponding estimation error becomes significant, particularly if $d$ is large. Moreover, $\bbSig_{\gamma}$ and $\bbSig_{\Gamma}$ may be singular if the LD matrix is singular %vilhjalmsson2015modeling
\citep{zhu2017bayesian}, which makes $(\bbSig_{\Gamma} + \beta^{2} \bbSig_{\gamma})^{-1}$ and hence $\bbQ^{\star}(\beta)$ hard to estimate and even undefined. 

We modify $\bbQ^{\star}(\beta)$ to address this issue. Note that estimators for diagonal matrices typically converge faster than those of  matrices without structural assumptions \citep{cai2016estimating}. Replacing $\bbSig_{\gamma}$ and $\bbSig_{\Gamma}$ in $\bbQ^{\star}(\beta)$ with diagonal matrices $\bbD_{\gamma}$ and $\bbD_{\Gamma}$, we obtain
\begin{equation}\label{eq: DEEM pop}
\nodisplayskips
\{\hgamma - \bbQ(\beta)(\hGamma - \beta \hgamma)\}^{\T}\bbV^{-1}(\hGamma - \beta \hgamma) = 0,
\end{equation}	
where $\bbQ(\beta) = -\beta \bbD_{\gamma}(\bbD_{\Gamma} + \beta^{2}\bbD_{\gamma})^{-1}$.  A naive approach is to use the diagonal elements of $\bbSig_{\gamma}$ and $\bbSig_{\Gamma}$ to form $\bbD_{\gamma}$ and $\bbD_{\Gamma}$.  However, this would result in $\hgamma - \bbQ(\betax)(\hGamma - \betax \hgamma)$ and $\hGamma -\betax \hgamma$ being correlated, leading to bias in the EE \eqref{eq: DEEM pop}.  See Appendix Section \ref{app: proof of sec orDEEM} for details.

To address this, a careful design of the diagonal matrices $\bbD_{\gamma}$ and $\bbD_{\Gamma}$ is necessary.  We begin by identifying the targets 
$\bbD_{\gamma}$ and $\bbD_{\Gamma}$	at the population level. The details of their estimation will be addressed later. Define the inner product of two square matrices $\bbA_{1}$ and $\bbA_{2}$ as $\langle\bbA_{1}, \bbA_{2}\rangle_{\bbV} =\tr\{\bbA_{1}^{\T}\bbV^{-1}\bbA_{2}\}$. To ensure the unbiasedness of the EE \eqref{eq: DEEM pop}, we propose to take
$\bbD_{\gamma}$ and $\bbD_{\Gamma}$ as the projections of $\bbSig_{\gamma}$ and $\bbSig_{\Gamma}$ onto the set of all diagonal matrices $\cD$, i.e.,
\begin{equation}
\nodisplayskips
\bbD_{\gamma} = \mathop{\arg\min}_{\bbD \in \cD} \langle\bbSig_{\gamma} - \bbD, \bbSig_{\gamma} - \bbD\rangle_{\bbV}\enspace \text{and} \enspace \bbD_{\Gamma} = \mathop{\arg\min}_{\bbD \in \cD} \langle\bbSig_{\Gamma} - \bbD, \bbSig_{\Gamma} - \bbD\rangle_{\bbV}.
\label{eq:dd}
\end{equation}
With this choice of $\bbD_{\gamma}$ and $\bbD_{\Gamma}$, Proposition 1 shows that \eqref{eq: DEEM pop} is unbiased, and provides the explicit expression of $\bbD_{\gamma}$ and $\bbD_{\Gamma}$.
%and the unbiasedness of \eqref{eq: DEEM pop}. 
For any matrix $\bbA$, let $[\bbA]_{ij}$ be the $(i,j)$th element of $\bbA$.

\begin{proposition} \label{prop: D and unbiasedness}
For $j=1, \dots, d$, the $j$th diagonal elements of $\bbD_{\gamma}$ and $\bbD_{\Gamma}$ that solve (\ref{eq:dd}) equal
$[\bbV^{-1}\bbSig_{\gamma}]_{jj}/[\bbV^{-1}]_{jj}$ and $[\bbV^{-1}\bbSig_{\Gamma}]_{jj}/[\bbV^{-1}]_{jj}$, respectively.
If $\bbD_{\Gamma} + \betax^{2}\bbD_{\gamma}$ is nonsingular, then
\[
\nodisplayskips
E\left[\{\hgamma - \bbQ(\betax)(\hGamma - \betax \hgamma)\}^{\T} \bbV^{-1}(\hGamma - \betax \hgamma)\right] = 0.
\] 
\end{proposition}
The proof of Proposition \ref{prop: D and unbiasedness} is in Appendix Section \ref{app: proof of sec orDEEM}.
Here, we provide some explanations on the unbiasedness of \eqref{eq: DEEM pop}. Let $\bbI$ be the identity matrix of the proper order. Then,
\begin{equation}\label{eq: inner product representation}
\nodisplayskips
\begin{aligned}
&E[\{\hgamma - \bbQ(\betax)(\hGamma - \betax \hgamma)\}^{\T}\bbV^{-1}(\hGamma - \betax \hgamma)]\\
&= 
- \langle \bbQ(\betax), \bbSig_{\Gamma} - \bbD_{\Gamma}\rangle_{\rm \bbV} - \langle \betax^{2}\bbQ(\betax) + \betax \bbI, \bbSig_{\gamma} - \bbD_{\gamma}\rangle_{\rm \bbV}.
\end{aligned}
\end{equation}
The matrices $\bbQ(\betax)$ and $\betax^{2}\bbQ(\betax) + \betax \bbI$ are both diagonal. According to the property of the projections, the residuals $\bbSig_{\gamma} - \bbD_{\gamma}$ and $\bbSig_{\Gamma} - \bbD_{\Gamma}$ are orthogonal to every diagonal matrix and hence the right-hand-side of Equation \eqref{eq: inner product representation} equals zero. Thus, the diagonal matrix $\bbQ(\betax)$ partially decorrelates $\hgamma - \bbQ(\betax)(\hGamma - \betax \hgamma)$ and $\hGamma -\betax \hgamma$ such that the estimating equation \eqref{eq: DEEM pop} is unbiased.

%	The following proposition shows that the debiased EE is unbiased in the presence of weak IVs.

Compared to \eqref{eq: plug-in EE}, the EE \eqref{eq: DEEM pop} includes an additional term $-\bbQ(\beta)(\hGamma - \beta \hgamma)$. The additional term is mean zero when $\beta = \betax$, and is of order $o_{P}(\|\hgamma\|)$ if the effects of SNPs are large in the sense that $\|\bgamma\|^{2} \gg \tr\{\var\{\bbQ(\betax)(\hGamma - \betax \hgamma)\}\}$. In this case, the impact of the additional term is negligible. Then, the solution of \eqref{eq: DEEM pop} has the same asymptotic properties as $\hbeta_{\rm PlugIn}$. When $\|\bgamma\|$ is small, the EE \eqref{eq: plug-in EE} has a non-negligible bias, and the term $-\bbQ(\beta)(\hGamma - \beta \hgamma)$ helps to eliminate the bias.

To estimate $\bbD_{\gamma}$ and $\bbD_{\Gamma}$, we first construct estimators $\hSig_{\gamma}$ and $\hSig_{\Gamma}$ for $\bbSig_{\gamma}$ and $\bbSig_{\Gamma}$ using $\hgamma$, $\hGamma$, the variance estimators of their components, and the LD matrix from a reference panel such as the 1000 Genomes data. Careful attention is required to ensure that the LD structure of the reference sample aligns with that of the exposure and outcome samples. To achieve this, it is crucial to select a reference sample that shares a similar ancestry as the exposure and outcome samples.  See Appendix Section \ref{app: est Sig} for the concrete forms of $\hSig_{\gamma}$ and $\hSig_{\Gamma}$. Given $\hSig_{\gamma}$ and $\hSig_{\Gamma}$,
Proposition \ref{prop: D and unbiasedness} leads to the estimators $\hD_{\gamma}$ and $\hD_{\Gamma}$ for $\bbD_{\gamma}$ and $\bbD_{\Gamma}$, and the debiased EE in \eqref{eq: orDEEM}.

We numerically compare the estimators from the plug-in EE \eqref{eq: plug-in EE}, the decorrelated (Decorr) EE \eqref{eq: full DEEM} --- where unknown $\bbSig_{\Gamma}$, $\bbSig_{\gamma}$ in $\bbQ^{\star}(\beta)$ are replaced by $\hSig_{\Gamma}$ and $\hSig_{\gamma}$, respectively --- and the diagonalized-decorrelated (Diag-decorr) EE \eqref{eq: orDEEM}. Simulations were conducted under the two-sample setting described in Section \ref{sec: sim} with $\pi_{c} = 0.05$, p-value threshold $0.1$ and no LD clumping. The parameter $h_{c}^{2}$  represents the overall effect size of SNPs on $X$, with a higher value indicating a larger effect. %{\red how is $h_c$ defined?} {\re It reflects the heritability in the random effect sense. It is directly heritability because our $G$ and $Y$ are not unit variance.} 
Figure \ref{fig: diag} shows the plug-in EE leads to biased estimates. Compared to the estimate from the decorrelated EE, the proposed diagonalized-decorrelated EE yields  an estimate with a smaller bias and a smaller standard error especially when SNPs have small effect sizes.
\begin{figure}
	\centering
	\includegraphics[scale = 0.3]{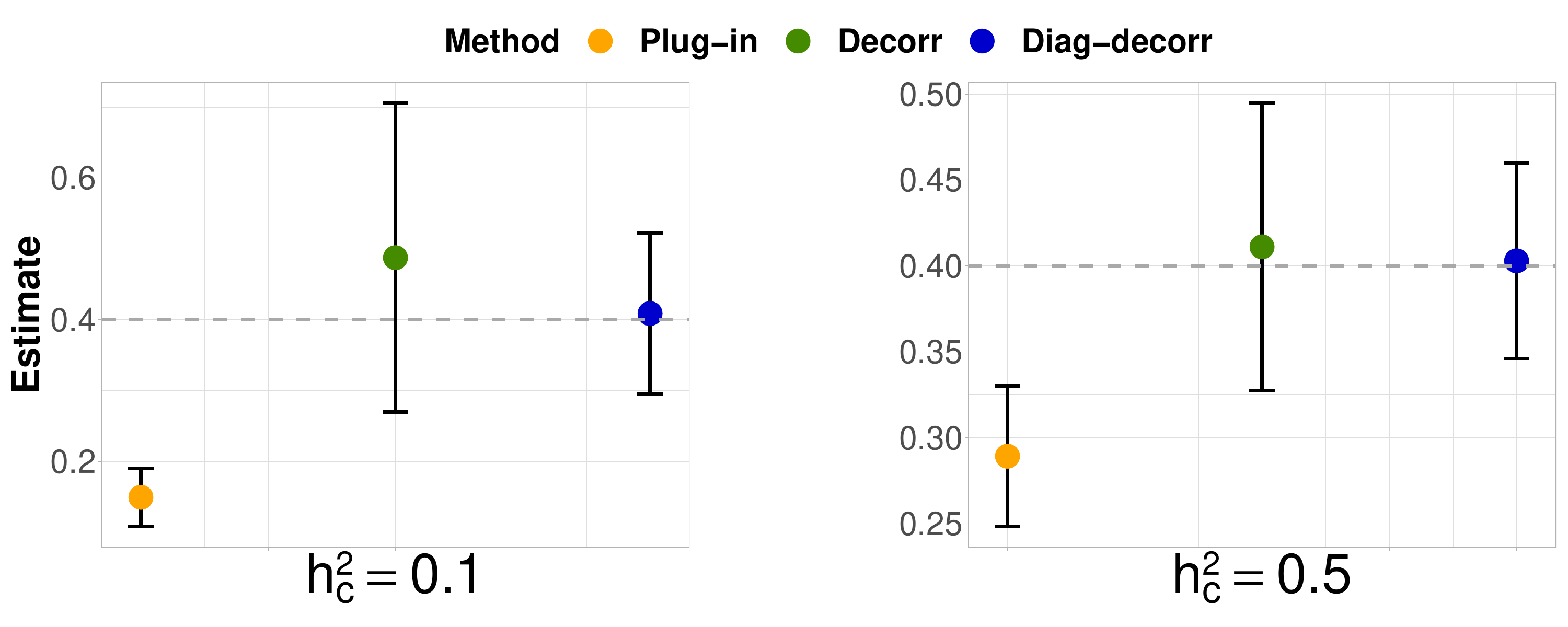}
	\caption{Boxplots of the estimators using the plug-in EE \eqref{eq: plug-in EE}, the naive decorrelated \eqref{eq: full DEEM} (with estimated $\bbQ^{\star}(\beta)$), and diagonalized-decorrelated EE \eqref{eq: orDEEM}   based on simulations. $\betax = 0.4$. $h_{c}^{2}$ represents the overall effect size of SNPs on $X$.
    }\label{fig: diag}
\end{figure} 
\begin{remark}\label{remark: diag rate}

The projection inequality in Appendix Section \ref{app: est Sig} gives
\[
\|\hD_\gamma-\bbD_\gamma\|
=O_P\left(\max_j[\bbSig_\gamma]_{jj}\sqrt{\log d/n_{\rm ref}}\right)
\]
under the stated concentration, projection, and LD approximation conditions. Here $n_{\rm ref}$ is the reference sample size. This diagonal bound depends logarithmically on $d$; it does not by itself establish operator-norm consistency of the full covariance estimator, which is assumed separately for variance estimation. Analogous conclusions apply to $\hD_\Gamma$.

\end{remark}

The EE \eqref{eq: orDEEM} effectively incorporates highly correlated SNPs with weak effects. By solving \eqref{eq: orDEEM}, we can obtain an asymptotically normal estimator for $\betax$ in the absence of horizontal pleiotropic effects.  The formal theoretical result is deferred to Section \ref{subsec: AN and combine}. We extend the EE \eqref{eq: orDEEM} to accommodate pleiotropic effects in Section \ref{subsec: DEEM p}.

\subsection{ The Ensemble Estimator Using  a Supplemental Sample}\label{subsec: AN and combine}
In the context of MR analyses with millions of SNPs in GWAS, selectively including SNPs based on their relevance to the exposure $X$ is critical.  Not all SNPs contribute to $X$, and indiscriminate inclusion introduces noise, potentially overwhelming the signal. If there is a gene region with a known biological link to $X$, the debiased EE \eqref{eq: orDEEM} can be effectively applied using SNPs from that specific region.  In cases where such a gene region is not available, or when leveraging genome-wide data is preferable for enhanced efficiency, the SNP selection strategy becomes crucial for the downstream analysis \citep{zhao2020statistical}. One potential approach is to include the
$j$th SNP in the analysis only if $|\widehat{\gamma}_{j}|$ exceeds a specific threshold. However, this criterion can lead to bias, as the expectation of $\widehat{\gamma}_{j}$ may deviate from the true value $\gamma_{j}$ conditional on the selection of the $j$th SNP because the selection procedure depends on $\widehat{\gamma}_{j}$. The phenomenon is known as the ``winner's curse" \citep{zhao2019powerful,ma2023breaking}. To avoid the winner's curse, we select the SNPs based on the estimate $\tgamma$ of $\bgamma$ from an independent supplemental exposure sample, e.g., the summary statistics of a GWAS in the GWAS Catalog that is independent of the GWASs used to calculate $\hgamma$ and $\hGamma$.  See Appendix Section \ref{app: selection} for details of the selection procedure. The same strategy has been adopted in many existing MR methods \citep{zhao2019powerful,ye2021debiased}. We refer to the supplemental exposure sample as the supplemental sample for short.
For simplicity, we keep the notations in the above sections and use them to denote the corresponding quantities of the selected SNPs. For instance, $\hGamma$ now denotes the estimated marginal effects of SNPs in the selected set on the outcome $Y$ obtained from the outcome sample, and $\tgamma$ is the estimated marginal effects of SNPs in the selected set on the exposure $X$ obtained from the supplemental sample.

Using SNPs selected based on the supplemental sample, an asymptotically unbiased estimator of $\betax$ can be obtained using the EE \eqref{eq: orDEEM}. 
%On the other hand, 
The estimate $\tgamma$ from the supplemental sample contains additional information about $\bgamma$ and can be used to construct an EE 
\begin{equation}\label{eq: disDEEM}
	\nodisplayskips
	\tgamma^{\T} \bbV^{-1}(\hGamma - \beta \hgamma) = 0.
\end{equation}
Although $\tgamma$ itself suffers from selection bias, the EE \eqref{eq: disDEEM} remains unbiased due to the independence between $\tgamma$ and $(\hgamma, \hGamma)$. 
Thus, the solution of \eqref{eq: disDEEM}, which has the explicit form
$\tgamma^{\T} \bbV^{-1}\hGamma/\tgamma^{\T} \bbV^{-1}\hgamma$, is an asymptotically unbiased estimator for $\betax$. 
Leveraging the EE \eqref{eq: disDEEM}, we can further exploit the supplemental sample to improve efficiency.

By solving \eqref{eq: orDEEM} and \eqref{eq: disDEEM}, two asymptotically unbiased estimators for the causal effect $\betax$ can be obtained. However, neither of them is guaranteed to be fully efficient because they are obtained from EEs that are different from the infeasible optimal EE \eqref{eq: EE opt}. We propose an ensemble estimator by  combining the solutions of EEs \eqref{eq: orDEEM} and \eqref{eq: disDEEM} through a weighted average to further improve efficiency. Specifically, denote the solutions of \eqref{eq: orDEEM} and \eqref{eq: disDEEM} by $\hbeta_{1}$ and $\hbeta_{2}$, respectively. For weights $w_{1}$ and $w_{2}$ that sum to one, the DEEM estimator for $\betax$ is defined as the ensemble estimator
\[
\nodisplayskips
\enbeta_{\rm DEEM} = w_{1} \hbeta_{1} + w_{2}\hbeta_{2}.
\]

Next, we establish the asymptotic normality of the ensemble estimator $\enbeta_{\rm DEEM}$ for a generic $\bw=(w_1,w_2)^T$. Let $n_{e}$, $n_{o}$, and $n_{s}$ be the (effective) sample sizes of the exposure sample (for $\hgamma$), the outcome sample (for $\hGamma$), and the supplemental sample (for $\tgamma$), respectively.  Assume $n_{e}/n_{o}$ is bounded away from zero and infinity. Let $n = \min\{n_{e}, n_{o}\}$. Our theoretical framework allows the dimensions of $\hgamma$ and $\hGamma$ to diverge as $n \to \infty$, permits the weights $w_{1}$ and $w_{2}$ to vary with $n$ and accommodates a data-driven working covariance matrix $\bbV$.
\begin{theorem}\label{thm: AN origin}
Under suitable regularity conditions (see Appendix Section \ref{app: notations and conds}),
Equation \eqref{eq: orDEEM} has a unique solution $\hbeta_{1}$ in a neighborhood of $\betax$ with probability approaching one, and the ensemble estimator $\enbeta_{\rm DEEM} = w_{1}\hbeta_{1} + w_{2}\hbeta_{2}$ satisfies
\[
\nodisplayskips
(\bw^{\T}\bbPsi\bw)^{-\frac{1}{2}}(\enbeta_{\rm DEEM} - \betax)		
\to N(0, 1)\]
in distribution as $n\to \infty$, where $\bw = (w_{1}, w_{2})^{\T}$ is the weight vector and $\bbPsi$ is defined in Appendix Section \ref{app: proof of sec AN and combine}. 
\end{theorem}
The proof of Theorem \ref{thm: AN origin} can be found in Appendix Section \ref{app: proof of sec AN and combine}. The regularity conditions required by Theorem \ref{thm: AN origin} are deferred to Appendix Section \ref{app: notations and conds} due to space limit. We briefly describe the asymptotic regime and provide a high-level interpretation of the key regularity conditions in the following remark. 
\begin{remark}\label{remark: cond thm1}
    We allow the working covariance matrix $\bbV$ to be data-dependent and correlated with other summary statistics, and consider the randomness of $\tgamma$ in Theorem \ref{thm: AN origin} and Theorems \ref{thm: AN pleiotropy}, \ref{thm: AN os} in Section \ref{sec: extension}. We assume $\bbV$ converges to some matrix $\bbV_{*}$ which may not be $\bbSig$, and
    is well-conditioned and nonrandom conditional on the supplemental sample. Theorem \ref{thm: AN origin} focuses on the ``many weak IV" regime where the number of selected SNPs $q$ goes to infinity and $\|\bgamma\|$ can go to zero because it fits better in the MR context with many weak-effect SNPs. Specifically, we assume that the convergence rate of the estimated nuisance matrices, such as $\bbV$, is faster than the rate $n^{-\kappa}$, where $\kappa$ is a positive constant. For the number of the selected SNPs, we assume that there are positive sequences $\{q_{1n}\}_{n=1}^{\infty}$ and $\{q_{2n}\}_{n=1}^{\infty}$ such that $P(q_{1n} \leq q \leq q_{2n}) \to 1$, $q_{1n} \to \infty$ and $n^{-\kappa}\sqrt{q_{2n}}\to 0$. For the association strength $\|\bgamma\|$, we assume $\|\bgamma\|^{2} \geq cq_{2n}/n$ for some constant $c > 0$. 
\end{remark}

The asymptotic normality of $\hbeta_{1} - \betax$ and $\hbeta_{2} - \betax$ is implied by Theorem \ref{thm: AN origin} by taking $\bw = (1, 0)^{\T}$ and $(0, 1)^{\T}$ respectively.
A simple calculation can show that the optimal weight vector minimizing the asymptotic variance of $\enbeta_{\rm DEEM}$ is 
$(\boldsymbol{1}^{\T}\bbPsi^{-1}\boldsymbol{1})^{-1}\bbPsi^{-1}\boldsymbol{1}$ and the corresponding asymptotic variance is $(\boldsymbol{1}^{\T}\bbPsi^{-1}\boldsymbol{1})^{-1}$,
where $\boldsymbol{1} = (1, 1)^{\T}$. In Appendix Section \ref{app: proof of sec AN and combine}, we introduce a consistent estimator $\hPsi$ for $\bbPsi$. Then, the DEEM estimator with the estimated optimal weight $(\widehat{w}_{\rm  1},\widehat{w}_{\rm 2}) = (\boldsymbol{1}^{\T}\hPsi^{-1}\boldsymbol{1})^{-1}\boldsymbol{1}^{\T}\hPsi^{-1}$ is
%can be estimated by
$\enbeta_{\rm DEEM}^{\rm opt} = \widehat{w}_{\rm 1}\hbeta_{1} + \widehat{w}_{\rm 2}\hbeta_{2}$, 
where $\widehat{\bbPsi}$ is the estimator for $\bbPsi$.
%The method we propose, which encompasses the construction and combination of $\hbeta_{1}$ and $\hbeta_{2}$, is termed the DEEM (Debiased Estimating Equation Method). 
It can be shown that
$(\boldsymbol{1}^{\T}\hPsi^{-1}\boldsymbol{1})^{\frac{1}{2}}(\enbeta_{\rm DEEM}^{\rm opt} - \betax) \to N(0, 1)$
in distribution according to Theorem \ref{thm: AN origin} and Slutsky's theorem, provided that $\hbeta_{1} - \betax$ and $\hbeta_{2} - \betax$ share the same convergence rate. We always use the estimated optimal weight in implementation.

\subsection{The LD Estimation and the Working Covariance Matrix Construction}
\label{subsec: LD & V}

The constructions of $\hSig_{\gamma}$, $\hSig_{\Gamma}$, and $\bbV$ require an estimate of the LD among SNPs. We use a blockwise diagonal matrix based on the sample LD matrix from a reference panel (e.g., the 1000 Genomes data) to estimate the LD. Specifically, suppose the SNPs selected based on the supplemental sample can be divided into $K$ LD blocks specified by the user or existing literature. Let $\hR_{{\rm ref}, k}$ be the sample LD matrix for the $k$-th LD block from a reference panel for $k = 1,\dots, K$. Then, the LD matrix for the selected SNPs can be estimated by $\hR_{\rm ref} = \diag\{\hR_{{\rm ref}, 1}, \dots, \hR_{{\rm ref}, K}\}$.

Users have the flexibility in choosing the working covariance matrix  $\bbV$ for their analysis.   As an example, we provide a specific construction of $\bbV$ adopted in our numerical implementations constructed based on $\hR_{\rm ref}$. As publicly available GWAS summary statistics often provide the SEs of $\widehat{\bGamma}$ and $\widehat{\bgamma}$, we use these SEs to construct a diagonal working variance matrix  $\bbV_d$
    %$\hS_{\Gamma}$ 
     to approximate the variances of components of  $\hGamma - \betax\hgamma$ (see Appendix Section \ref{app: construction of V} for its specific form). Notice that the correlation structure of $\hGamma - \betax\hgamma$ can be approximated by the LD matrix \citep{zhu2017bayesian}. We set 
    $\bbV = \bbV_d^{1/2}\hR\bbV_d^{1/2}$ where $\hR = \hR_{\rm ref} + c\bbI_{q}$ 
    is a block diagonal working correlation matrix of $\hGamma - \betax\hgamma$, $c$ some positive constant, and $\bbI_{q}$ the identity matrix of size $q$.
    The term $c\bbI_{q}$ is included to ensure that $\hR$ is well-conditioned as some SNPs in an LD block may be highly correlated and $\hR_{\rm ref}^{-1}$ may be unstable.
    The above construction of the working covariance matrix $\bbV$ performs well in our numerical experiments. Please refer to Appendix Section \ref{app: construction of V} for more details of the rationale behind the above construction and Section \ref{app: choice V} for a numerical illustration for the effectiveness of incorporating the correlation among SNPs for improving the efficiency of the resulting estimator.

    As noted in Remark \ref{remark: cond thm1}, our asymptotic normality result relies on the ``probability limit" $\bbV_{*}$ of $\bbV$ being well-conditioned, which suggests that one should take a well-conditioned $\bbV$. On the other hand, the (infeasible) optimal estimating equation (\ref{eq: EE opt}) suggests that the possibly singular true covariance matrix $\bbSig$ may benefit the efficiency of the resulting estimator. This reveals an inherent trade-off between efficiency and the accuracy of the normal approximation.
    Intuitively, the trade-off is governed by the tuning parameter $c$ in the construction of $\bbV$ in the last paragraph, with a
    larger $c$ leading to a better coverage rate at the cost of efficiency loss. This intuition is justified by our numerical experiments in Appendix Section \ref{app: choice c}. In our numerical experiment, the choice $c = 1$ provides a reasonable compromise, yielding good efficiency and satisfactory coverage of the normal-based confidence interval. Nevertheless, determining how to optimally balance efficiency and approximation accuracy remains an open and important problem.

\section{Extensions}\label{sec: extension}
\subsection{Accounting for Direct Effects}\label{subsec: DEEM p}
In previous sections, we adopted the classical IV  assumption that  SNPs don't directly affect the outcome variable $Y$, which may be %implausible 
violated due to pleiotropy \citep{bowden2015mendelian}, where SNPs affect both $X$ and $Y$ directly (Refer to Figure~\ref{fig: direct effect} in the Appendix for an illustration). Thus, we modify the outcome model \eqref{eq: outcome model} to allow for a direct effect of $\bG$ on $Y$. Suppose 
\begin{equation}\label{eq: pleiotropy outcome model}
\nodisplayskips
Y = \beta_{0} + X\betax + \bG^{\T}\bbeta_{G} + f(\bU)  + \epsilon_{Y},
\end{equation}
where $\bG \Perp \bU$ and $\bbeta_{G}$ quantifies the direct effects of 
$\bG$ on $Y$.
Invoking the law of large numbers, the marginal regression coefficients of $X$ and $Y$ on $G_{j}$ converge in probability to
$ \gamma_{j} = \var(G_{j})^{-1}\cov(G_{j}, X)$
and
$\Gamma_{j} = \betax \var(G_{j})^{-1}\cov(G_{j},X) + \var(G_{j})^{-1}\cov(G_{j}, \bG)\bbeta_{G}$, 
respectively, for $j=1,\dots,d$. Let $\bgamma = (\gamma_{1},\dots,\gamma_{d})^{\T}$ and $\bGamma = (\Gamma_{1},\dots,\Gamma_{d})^{\T}$. Then 
\[
\nodisplayskips
\bGamma = \betax \bgamma + \bbP \bbeta_{G},
\]
where $\bbP = \diag\{\var(G_{1})^{-1}, \dots, \var(G_{d})^{-1}\}\cov(\bG)$. The displayed relation is written for all $d$ SNPs. With a slight abuse of notation, in the presence of SNP selection, $\bbP$ denotes the $q\times d$ matrix of selected rows of the full-SNP matrix above, whereas $\bbeta_G$ retains the direct effects of all $d$ SNPs. Thus $\bbSig_{\rm p}=\bbP\bbP^{\T}$ retains contributions from unselected SNPs correlated with the selected SNPs. Here, we retain the normal assumption \eqref{eq: normal model}. The traditional IV assumption is violated in the presence of the direct effects. To identify the causal effect $\beta_X$ in the presence of the direct effects, we impose the following assumption on $\bbeta_{G}$. 
\begin{assumption}\label{ass: direct effect}
The components of $\bbeta_{G}$ are independent and identically distributed (i.i.d.) with mean zero and variance $\taup$.
\end{assumption}
Assumption \ref{ass: direct effect} allows for invalid instruments to have direct effects on $Y$. It only makes moment assumptions on the direct effect $\bbeta_G$ without distributional restrictions such as normality.
% It is less restrictive than the normal assumption on the components of $\bbeta_{G}$, which have often been assumed in the literature 
% \citep{zhao2019powerful, ye2021debiased}. Under the normal assumption, every component of $\bbeta_{G}$ is nonzero with probability one; hence, every SNP directly affects the outcome $Y$, which is a strong assumption and is generally not practically plausible. In contrast, Assumption \ref{ass: direct effect} allows the distribution of the direct effect to have a probability mass at zero, making it more practically plausible than the normal assumption.

Under Assumption  \ref{ass: direct effect}, it can be easily shown that
$E(\hGamma - \betax\hgamma) = E(\bbP\bbeta_{G}) = 0$.  This moment condition ensures that the EE \eqref{eq: disDEEM} remains unbiased even when direct effects are present. The EE \eqref{eq: orDEEM} however requires adjustments because the variance-covariance matrix of $\hGamma$ differs with and without direct effects, which can invalidate the debiasing procedure in Section \ref{subsec: orDEEM}. Specifically, $\cov(\hGamma) = \bbSig_{\Gamma} + \taup \bbSig_{\rm p}$ contains an additional component $\taup \bbSig_{\rm p}$ when direct effects are present, where $\bbSig_{\rm p} = \bbP\bbP^{\T}$. To account for this, we modify $\bbQ(\beta)$ and introduce 
$\bbQ(\beta; \tau) = - \beta \bbD_{\gamma}(\bbD_{\Gamma} + \tau \bbD_{\rm p} + \beta^{2} \bbD_{\gamma})^{-1}$
for any $\beta$ and $\tau$, where $\bbD_{\rm p} = \mathop{\arg\min}_{\bbD \in \cD} \langle \bbSig_{\rm p} - \bbD, \bbSig_{\rm p} - \bbD\rangle_{\bbV}$. According to the proof of Proposition \ref{prop: D and unbiasedness} in Appendix Section \ref{app: proof of sec orDEEM}, it is straightforward to show that the EE
$\{\hgamma - \bbQ(\beta; \taup)(\hGamma - \beta \hgamma)\}^{\T} \bbV^{-1}(\hGamma - \beta \hgamma) = 0$ is unbiased in the presence of direct effects.

The matrix $\bbSig_{\rm p}$ can be estimated based on summary statistics and an estimate $\hD_{\rm p}$ for $\bbD_{\rm p}$ can be obtained similarly to $\hD_{\gamma}$ and $\hD_{\Gamma}$. See Appendix Section \ref{app: est Sig} for more details. The remaining issue is to estimate $\taup$. Note that the definitions of $\bbD_{\gamma}$, $\bbD_{\Gamma}$, and $\bbD_{\rm p}$ imply the moment condition
$E\{(\hGamma - \betax \hgamma)^{\T}\bbV^{-1}(\hGamma - \betax \hgamma)\} = \tr\{\bbV^{-1}(\bbD_{\Gamma} + \betax^{2}\bbD_{\gamma})\} + \taup \tr\{\bbV^{-1}\bbD_{\rm p}\}$,
which motivates us to estimate $\taup$ utilizing
$\htau_{\rm p} = \tr\{\bbV^{-1}\hD_{\rm p}\}^{-1}(\hGamma - \hbeta_{2} \hgamma)^{\T}\bbV^{-1}(\hGamma - \hbeta_{2} \hgamma) - \tr\{\bbV^{-1}\hD_{\rm p}\}^{-1}\tr\{\bbV^{-1}(\hD_{\Gamma} + \hbeta_{2}^{2}\hD_{\gamma})\}$,
where $\hbeta_{2}$ is the solution to the Equation \eqref{eq: disDEEM}.
Let
$\hQ(\beta; \tau) = - \beta \hD_{\gamma}(\hD_{\Gamma} + \tau \hD_{\rm p} + \beta^{2} \hD_{\gamma})^{-1}$.
We obtain the debiased EE
\begin{equation}\label{eq: orDEEM direct effect}
\nodisplayskips
\{\hgamma - \hQ(\beta; \htau_{\rm p})(\hGamma - \beta \hgamma)\}^{\T}\bbV^{-1}(\hGamma - \beta \hgamma) = 0.
\end{equation}

Let $\hbeta_{1, \rm p}$ be the solution of the EE \eqref{eq: orDEEM direct effect}. The causal effect $\betax$ can be estimated through the weighted average $\enbeta_{\rm DEEM, p} = w_{1} \hbeta_{1, \rm p} + w_{2}\hbeta_{2}$ in line with the approach outlined in Section \ref{subsec: AN and combine}. The weights that minimize the resulting asymptotic variance can be obtained similarly as in Section \ref{subsec: AN and combine} based on the following theorem. We term $\enbeta_{\rm DEEM, p}$ as the DEEM estimator with pleiotropy.  The asymptotic normality of $\enbeta_{\rm DEEM, p} - \betax$ can be established in the presence of pleiotropic effects.

\begin{theorem}\label{thm: AN pleiotropy}
Under Assumption \ref{ass: direct effect} and suitable regularity conditions (see Appendix Section \ref{app: notations and conds}),
Equation \eqref{eq: orDEEM direct effect} has a unique solution $\hbeta_{1, \rm p}$ in a neighborhood of $\betax$ with probability approaching one, and the ensemble estimator $\enbeta_{\rm DEEM, p} = w_{1}\hbeta_{1, \rm p} + w_{2}\hbeta_{2}$ satisfies
\[
\nodisplayskips
(\bw^{\T}\bbPsi_{\rm p}\bw)^{-\frac{1}{2}}(\enbeta_{\rm DEEM, p} - \betax)		
\to N(0, 1)\]
in distribution as $n\to \infty$, where $\bw = (w_{1}, w_{2})^{\T}$ is the weight vector and $\bbPsi_{\rm p}$ is defined in Appendix Section \ref{app: proof of sec p}. 
\end{theorem}
Theorem \ref{thm: AN pleiotropy} establishes the asymptotic normality of $\enbeta_{\rm DEEM, p} - \betax$, under Assumption \ref{ass: direct effect} and suitable regularity conditions, affirming the robustness of our approach in the face of pleiotropy. The proof of this theorem, alongside the derivation of optimal weights, is detailed in Appendix Section \ref{app: proof of sec p}, further underscoring the methodological advancements offered by our analysis in addressing the complexities introduced by direct SNP effects on outcomes. In Appendix Section \ref{app: proof of sec p}, we introduce a consistent estimator $\hPsi_{\rm p}$ for $\bbPsi_{\rm p}$, which enables the inference for the causal effect $\betax$ based on Theorem \ref{thm: AN pleiotropy}.

\subsection{One-sample MR}\label{subsec: DEEM os}
Recently, large population-based biobanks, such as the UK biobank, have
made GWAS summary statistics available for various exposures and outcomes, which allow powerful one-sample MR studies. This section considers Model \eqref{eq: pleiotropy outcome model} in the one-sample setting where  $\hgamma$ and $\hGamma$ are obtained from the same sample. In the one-sample setting, the expectation $E(\hGamma - \betax\hgamma) = 0$ is preserved, ensuring that the EE \eqref{eq: disDEEM} maintains its unbiased nature. However, modifications are necessary to extend EE \eqref{eq: orDEEM direct effect} to this setting.

Specifically, additional terms arise in the expressions for $\cov(\hgamma, \hGamma - \betax \hgamma)$ and $\var(\hGamma - \betax \hgamma)$ due to the covariance between $\hgamma$ and $\hGamma$. 
We assume the summary-statistic covariance relation $\cov(\hgamma,\hGamma) = \rho_{U}\bbSig_{\gamma}$, where $\rho_{U}$ is a scalar parameter. A fixed-design calculation motivating this assumption is provided in Appendix Section \ref{app: proof of sec os}.
 Simple algebra shows that $\cov(\hgamma, \hGamma - \betax \hgamma) = (\rho_{U} - \betax)\bbSig_{\gamma}$ and $\var(\hGamma - \betax \hgamma) = \bbSig_{\Gamma} + \taup \bbSig_{\rm p} + (\betax^{2} - 2\rho_{U}\betax)\bbSig_{\gamma}$. We introduce $\bbQ(\beta; \tau, \rho)$, which extends $\bbQ(\beta; \tau)$, defined as $(\rho - \beta) \bbD_{\gamma}(\bbD_{\Gamma} + \tau \bbD_{\rm p} + (\beta^{2} - 2\rho\beta) \bbD_{\gamma})^{-1}$.
Similarly to Proposition \ref{prop: D and unbiasedness}, we can show that the EE $\{\hgamma - \bbQ(\beta; \taup, \rho_{U})(\hGamma - \beta \hgamma)\}^{\T} \bbV^{-1}(\hGamma - \beta \hgamma) = 0$ is unbiased. Please refer to  Appendix Section \ref{app: proof of sec os} for the proof.

The remaining task is to estimate $\rho_{U}$ and $\taup$ based on summary statistics. 
These parameters fulfill the moment conditions $E\{\hgamma^{\T} \bbV^{-1}(\hGamma - \betax \hgamma)\}
= (\rho_{U} - \betax)\tr\{\bbV^{-1}\bbD_{\gamma}\}$
and
$E\{(\hGamma - \betax \hgamma)^{\T}\bbV^{-1}(\hGamma - \betax \hgamma)\}
= \tr\{\bbV^{-1}(\bbD_{\Gamma} + (\betax^{2} - 2\rho_{U}\betax)\bbD_{\gamma})\} + \taup \tr\{\bbV^{-1}\bbD_{\rm p}\}$, respectively.
Thus, $\rho_{U}$ and $\taup$ can be estimated by
$\hat{\rho}  = \tr\{\bbV^{-1}\hD_{\gamma}\}^{-1}\hgamma^{\T}\bbV^{-1}(\hGamma - \hbeta_{2} \hgamma) + \hbeta_{2}$ and
$\htau_{\rm os} = \tr\{\bbV^{-1}\hD_{\rm p}\}^{-1}(\hGamma - \hbeta_{2} \hgamma)^{\T}\bbV^{-1}(\hGamma - \hbeta_{2} \hgamma) - \tr\{\bbV^{-1}\hD_{\rm p}\}^{-1}\tr\{\bbV^{-1}(\hD_{\Gamma} + (\hbeta_{2}^{2} - 2 \hat{\rho}\hbeta_{2})\hD_{\gamma})\}$, respectively,
where $\hbeta_{2}$ is the solution of Equation \eqref{eq: disDEEM}.
Letting
$\hQ(\beta; \tau, \rho) = (\rho - \beta) \hD_{\gamma}(\hD_{\Gamma} + \tau \hD_{\rm p} + (\beta^{2} - 2\rho\beta) \hD_{\gamma})^{-1}$, we obtain the debiased EE
\begin{equation}\label{eq: orDEEM os}
\nodisplayskips
\{\hgamma - \hQ(\beta; \htau_{\rm os}, \hat{\rho})(\hGamma - \beta \hgamma)\}^{\T}\bbV^{-1}(\hGamma - \beta \hgamma) = 0.
\end{equation} 
Let $\hbeta_{1, \rm os}$ denote the solution to Equation \eqref{eq: orDEEM os}, where the subscript ${\rm os}$ stands for one-sample, and define $\enbeta_{\rm DEEM, os} = w_{1} \hbeta_{1, \rm os} + w_{2} \hbeta_{2}$ as the DEEM estimator in the one sample setting. The estimator $\enbeta_{\rm DEEM, os}$ is asymptotically normal under appropriate regularity conditions. 
\begin{theorem}\label{thm: AN os}
Under Assumption \ref{ass: direct effect} and suitable regularity conditions (see Appendix Section \ref{app: notations and conds}),
Equation \eqref{eq: orDEEM os} has a unique solution $\hbeta_{1, \rm os}$ in a neighborhood of $\betax$ with probability approaching one, and the ensemble estimator $\enbeta_{\rm DEEM, os} = w_{1}\hbeta_{1, \rm os} + w_{2}\hbeta_{2}$ satisfies
\[
\nodisplayskips
(\bw^{\T}\bbPsi_{\rm os}\bw)^{-\frac{1}{2}}(\enbeta_{\rm DEEM, os} - \betax)		
\to N(0, 1)\]
in distribution as $n\to \infty$, where $\bw = (w_{1}, w_{2})^{\T}$ is the weight vector and $\bbPsi_{\rm os}$ is defined in Appendix Section \ref{app: proof of sec os}.  
\end{theorem}
The proof of Theorem \ref{thm: AN os} is in Appendix Section \ref{app: proof of sec os}. The optimal weight can be similarly obtained as in Section \ref{subsec: AN and combine}. We term $\enbeta_{\rm DEEM, os}$ as the one-sample DEEM estimator. A consistent estimator $\hPsi_{\rm os}$ for $\bbPsi_{\rm os}$ is provided in Appendix Section \ref{app: proof of sec os}, which enables the inference for the causal effect $\betax$ based on Theorem \ref{thm: AN os}.  

\section{Simulation Studies}\label{sec: sim}
In this section, we evaluate the performance of the DEEM across various genetic architectures via simulation, leveraging genotype data with realistic LD patterns.  We simulated the genotype data using HapMap3 SNPs. See  Appendix Section \ref{app: gen genotype} for the details of generating the genotype data. Our focus was on $d = 24,648$ HapMap3 SNPs on chromosome 22. The exposure $X$ and outcome $Y$ variables were simulated using the following models: 
%according to our design:
\[
\nodisplayskips
\begin{aligned}
&X = \bG^{\T}\balpha_{G} + U + \epsilon_{X},\\
&Y =  X\betax + \bG^{\T}\bbeta_{G} + U +  \epsilon_{Y},
\end{aligned}\]
where $\betax = 0.4$ denotes the causal effect of $X$ on $Y$, with $\epsilon_{X}$ and $\epsilon_{Y}$ being independent, each following a normal distribution $N(0, 0.5)$, and $U \sim N(0, 1)$. For $j = 1,\dots, d$, denote by $\alpha_{G,j}$ and $\beta_{G,j}$ the $j$th component of $\balpha_{G}$ and $\bbeta_{G}$, respectively.
To introduce sparsity in the genetic effects, binary indicators $I_{\alpha_{G,j}}$ and $I_{\beta_{G},j}$ were generated independently for each SNP ($j = 1,\dots, d$) following Bernoulli distributions with probabilities $\pi_{c}$ and $\pi_{d}$, corresponding to the proportions of causal SNPs affecting $X$ and SNPs exerting direct effects on $Y$, respectively.  The genetic effect sizes $\alpha_{G,j}$ and $\beta_{G,j}$ were assigned values of zero for SNPs not affecting $X$ ($I_{\alpha_{G,j}} = 0$)  or not directly affecting $Y$ ($I_{\beta_{G},j} = 0$), respectively. For SNPs affecting $X$ ($I_{\alpha_{G,j}} = 1$) and those having direct effects on $Y$ ($I_{\beta_{G},j} = 1$), effect sizes were drawn from $\alpha_{G,j} \sim N(0, h_{c}^{2} / \sum_{i=1}^{d}I_{\alpha_{G,i}})$ and $\beta_{G,j} \sim N(0, h_{d}^{2} / \sum_{i=1}^{d}I_{\beta_{G},i})$, where $h_{c}^{2}$ and $h_{d}^{2}$ denote heritability coefficients controlling the magnitude of influence on $X$ and the direct effect on $Y$, respectively. Parameters $\pi_{d} = 0.1$ and $\pi_{c} = 0.005, 0.05$, or $0.1$, along with heritability levels $h_{d}^{2} = h_{c}^{2} = 0.1$ or $0.5$, were varied to examine the DEEM's robustness under different levels of sparsity and effect sizes. The genetic effects on $X$ ($\balpha_{G}$) remained constant across simulations while the direct genetic effects on $Y$ ($\bbeta_{G}$) were generated independently in each replication. Set the sizes for the exposure, outcome, and supplemental samples at $n_{e} = 50,000$, $n_{o} = 50,000$, and $n_{s} = 20,000$, respectively.

LD clumping was performed to exclude colinear SNPs based on an  LD matrix estimated from an independent sample comprising $3000$ observations. This process utilized a block-diagonal approximation of the LD matrix, adhering to the LD blocks delineated by  \cite{berisa2016approximately}. Within the framework of the DEEM implementation, we selected SNPs exhibiting estimated absolute correlation coefficients below 0.9 through this LD clumping procedure. The working covariance matrix $\bbV$ is constructed as per Section \ref{subsec: LD & V} with $c = 1$. Our analysis applied DEEM in both two-sample and one-sample contexts. In the two-sample setting, the estimator $\hbeta_{\rm DEEM, p}$ was utilized, whereas $\hbeta_{\rm DEEM, os}$ was employed for the one-sample scenario. The DEEM estimator was implemented with the estimated optimal weights in both settings. 

For comparative analysis, we included two methods capable of addressing correlated SNPs: WLR \citep{burgess2016combining} and PCA-IVW \citep{burgess2017mendelian}. 
We further compared DEEM with four other popular methods that require the included SNPs to be independent: IVW, dIVW \citep{ye2021debiased}, pIVW \citep{xu2022novel}, and RAPS \citep{zhao2020statistical}.
For IVW, dIVW, pIVW, and RAPS, SNPs were chosen to be approximately independent, with absolute correlation coefficients under 0.01, as determined by the clumping process. We applied the selection procedure described in Appendix Section \ref{app: selection}. All methods used the selected SNPs from the supplemental sample.  Notably, dIVW, pIVW, RAPS, and DEEM—designed to be robust against weak instruments—were applied with a p-value inclusion threshold of $10^{-1}$ to incorporate potentially informative weak IVs. Conversely, for IVW, WLR, and PCA-IVW, we evaluated two p-value thresholds: $10^{-1}$ and $10^{-4}$. An “s” is appended to the name of a method (e.g., IVW(s)) to signify the employment of the more stringent threshold of $10^{-4}$. 
The average number of SNPs included per method across 200 replications is documented in Table \ref{table: nSNP main} in the Appendix.

We first report the results in the case where $\hgamma$ and $\hGamma$ come from two independent samples.  
%To characterize the average strength of IVs in different settings, we define the average absolute effect size of the causal SNPs as ${\rm EFF_{a}} = \sum_{j = 1}^{d}|\alpha_{G,j}|/ \sum_{j=1}^{d}I_{\alpha_{G,j}}$. $\rm EFF_{a}$ equals to $2.00\times 10^{-2}$, $6.30\times 10^{-3}$, $4.43\times 10^{-3}$, $4.47\times 10^{-2}$, $1.41\times 10^{-2}$ and $9.91\times 10^{-3}$ when $(p_{c}, h_{c}^{2}) = (0.005, 0.1)$, $(0.05, 0.1)$, $(0.1, 0.1)$, $(0.005, 0.5)$, $(0.05, 0.5)$, and $(0.1, 0.5)$, respectively. It can be seen that $\rm EFF_{a}$ increases as $h_{c}^{2}$ increases and decrease as $p_{c}$ increases.  
Figure~\ref{fig: sim ts selected} shows the biases, standard errors (SEs), and root mean square errors (RMSEs) of different methods across $200$ replications. The performance of WLR(s) and WLR is similar to that of PCA-IVW(s) and PCA-IVW, respectively. Additionally, dIVW and pIVW exhibit similar results to RAPS across different scenarios. We exclude WLR, dIVW and pIVW methods in Figure~\ref{fig: sim ts selected} for clarity of demonstration and defer their results to Figure~\ref{fig: sim ts} in the Appendix.
\begin{figure}[h]
	\centering
	\includegraphics[scale = 0.3]{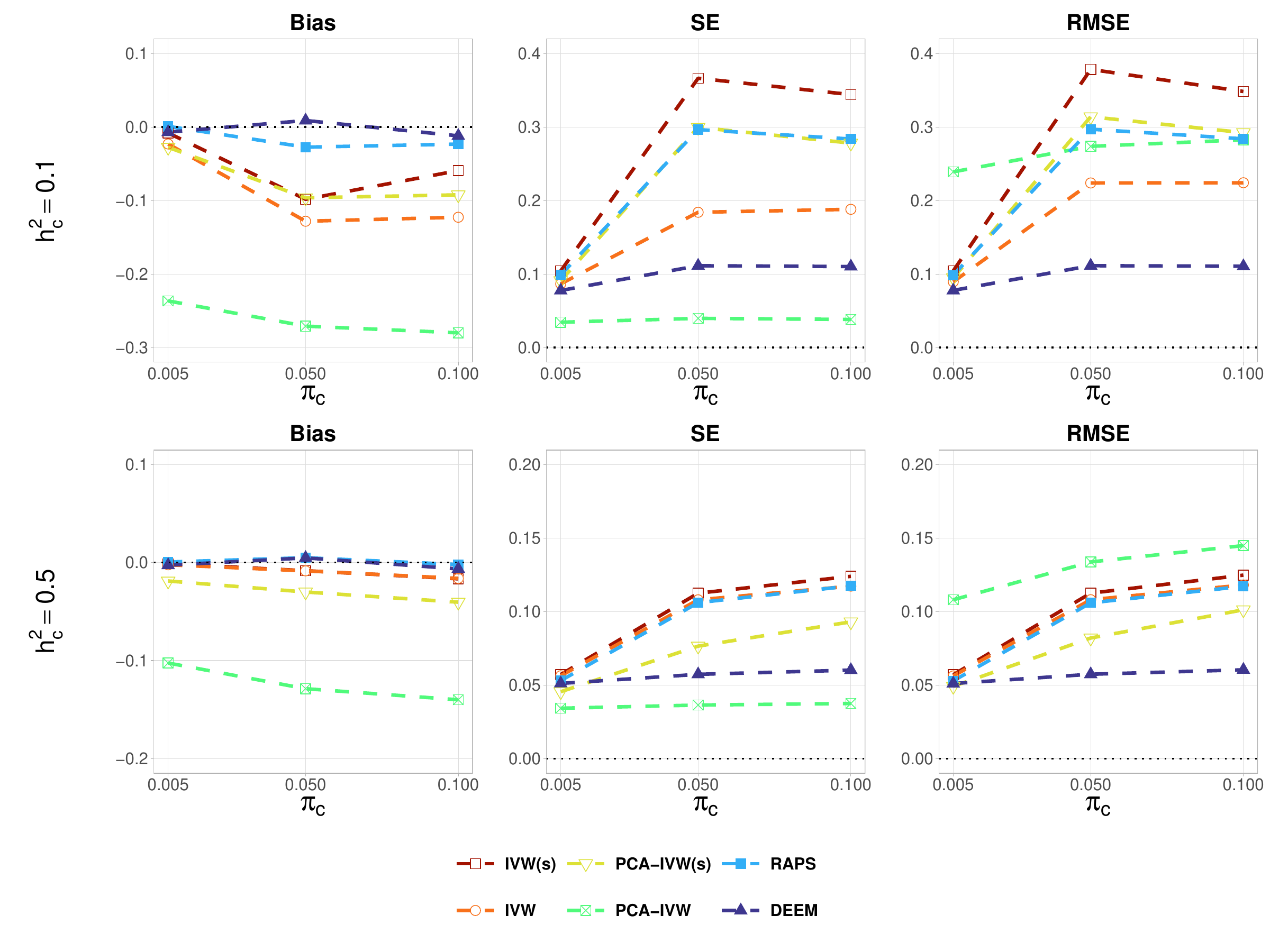}
	\caption{Simulation results on the biases, SEs, and RMSEs of different methods (IVW, PCA-IVW, RAPS and DEEM) in the two-sample setting with different combinations of $(\pi_{c}, h_{c}^{2})$. The methods PCA-IVW, RAPS and DEEM use the p-value threshold $0.1$. The ``s" in the parentheses (IVW(s) and PCA-IVW(s)) means the stringent p-value threshold $10^{-4}$ is adopted. $\pi_c$ denotes the proportion of SNPs that affect $X$.}\label{fig: sim ts selected}
\end{figure}

PCA-IVW exhibits significant biases under the p-value threshold $10^{-1}$, suggesting that it is biased due to the weak IV problem. This undermines the validity of its small standard errors. IVW also suffers from weak IV bias because a loose threshold of $10^{-1}$ is adopted and weak IVs are included. By using a stricter SNP selection threshold $10^{-4}$, IVW(s) and PCA-IVW(s) mitigate the weak IV bias at the cost of some loss of efficiency. Nevertheless, bias persists in these methods when SNP effects are weak and dense.	
%On the other hand, 
RAPS performs well in handling weak IV bias, but it tends to have large standard errors when $(\pi_{c}, h_{c}^{2}) = (0.05, 0.1)$ or $(0.1, 0.1)$. DEEM is robust against the weak IV bias and achieves a significant efficiency gain compared to RAPS, and the relative efficiency can be up to seven.
%Remarkably, 
In addition, DEEM consistently achieves the lowest RMSE compared to all other considered methods with the exception of the case where $(\pi_{c}, h_{c}^{2}) = (0.005, 0.5)$, which corresponds to the case with extremely sparse SNP effects on $X$ and very strong effects. The RMSE reductions compared to other methods can be up to $71.7\%$. The advantage of DEEM in RMSE is particularly pronounced when the effects are weak and dense.   WLR(s) and WLR suffer from the weak IV bias especially when SNP effects are weak and dense; dIVW and pIVW are robust to weak IV bias but have large standard errors under certain scenarios. See Appendix Section \ref{app: detail sim} for the performance of these methods and some other robust MR methods.

Next, we evaluate the performance of different methods in the one-sample setting where $\hgamma$ and $\hGamma$ are from the same sample, maintaining parameter settings consistent with the two-sample setting. Results are presented in Figure~\ref{fig: sim os selected}. The results of WLR(s), WLR, dIVW, and pIVW are not included in Figure~\ref{fig: sim os selected} for the sake of presentation clarity.
%same reasons as in the two-sample setting. 
See Figure~\ref{fig: sim os} in the Appendix for the results of WLR(s), WLR, dIVW, and pIVW compared with DEEM.

\begin{figure}[h]
	\centering
	\includegraphics[scale = 0.3]{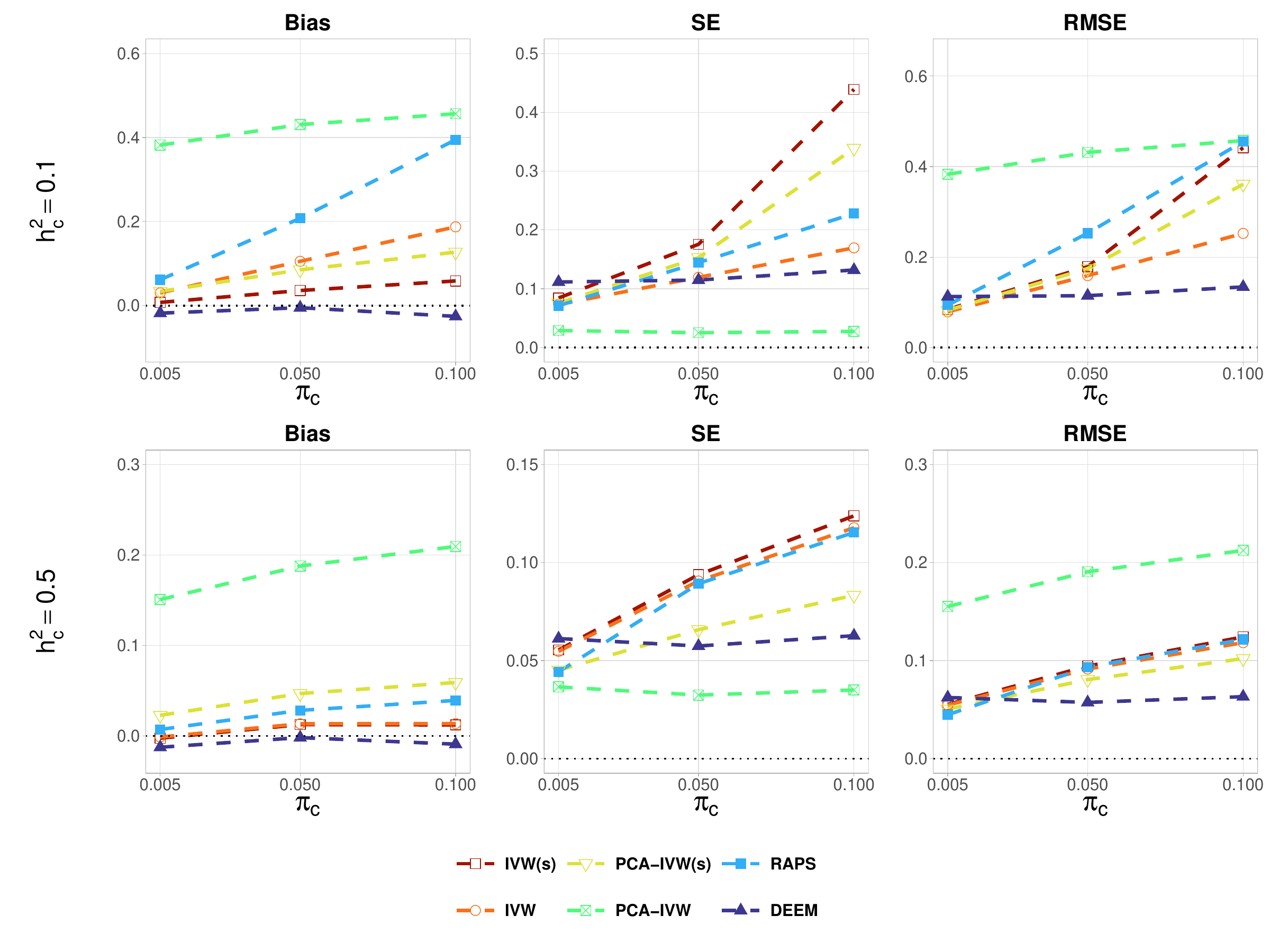}
	\caption{Simulation results of biases, SEs, and RMSEs of different methods (IVW, PCA-IVW, RAPS and DEEM) in the one-sample setting with different combinations of $(\pi_{c}, h_{c}^{2})$.
 The methods  PCA-IVW, RAPS and DEEM use the p-value threshold $0.1$. The ``s" in the parentheses (IVW(s) and PCA-IVW(s)) means the stringent p-value threshold $10^{-4}$ is adopted. $\pi_c$ denotes the proportion of SNPs that affect $X$.}\label{fig: sim os selected}
\end{figure}

In the one-sample setting, the performances of IVW(s), IVW, PCA-IVW(s), and PCA-IVW closely mirror those observed in the two-sample setting. However, RAPS exhibits significant biases when $(\pi_{c}, h_{c}^{2}) = (0.05, 0.1)$ and $(0.1, 0.1)$, attributed to a lack of correction for the correlation between $\hgamma$ and $\hGamma$.    
Albeit biased, IVW(s) exhibits relatively small biases among DEEM's competitors,  likely due to its rigorous SNP selection process combining a p-value threshold with LD clumping, ensuring the inclusion of only the most strongly associated SNPs. Similar derivations to those in Appendix Section \ref{app: weak IV} can demonstrate that the correlation between $\hgamma$ and $\hGamma$ has a minor impact on the resulting estimator when the included SNPs all have strong effects. 

Conversely, DEEM stands out for its consistently low bias across varied scenarios within the one-sample setting, often surpassing IVW(s) in efficiency with a relative efficiency of up to $11$ in certain scenarios. This superior performance can be attributed to DEEM's ability to effectively include correlated SNPs with weak effects while properly accounting for the correlation between $\hgamma$ and $\hGamma$. When $\pi_{c} = 0.005$ (extremely sparse and strong SNP effects on $X$), DEEM exhibits a similar or marginally higher RMSE compared to methods that incorporate only strong or independent IVs. Notably, when $\pi_{c}$ = $0.05$ or $0.1$ and SNPs are weak, DEEM achieves a lower RMSE than these other methods. 
WLR(s) and WLR suffer from the weak IV bias especially when SNP effects are weak and dense; dIVW and pIVW also tend to be biased due to the inadequate adjustment of the correlation between $\hgamma$ and $\hGamma$. See Figure~\ref{fig: sim os} in the Appendix for the performance of these methods and their comparisons with DEEM. 

Figure~\ref{fig: RE-CP} shows the relative efficiency of DEEM compared to other methods that have minimal bias,  alongside the coverage rate of the DEEM-based $95\%$ confidence interval. We compare DEEM with RAPS in the two-sample setting and IVW(s) in the one-sample setting because all other competitors have large biases under certain combinations of $\pi_{c}$ and $h_{c}^{2}$. 
\begin{figure}[h]
\centering
\subcaptionbox{}{\includegraphics[scale = 0.225]{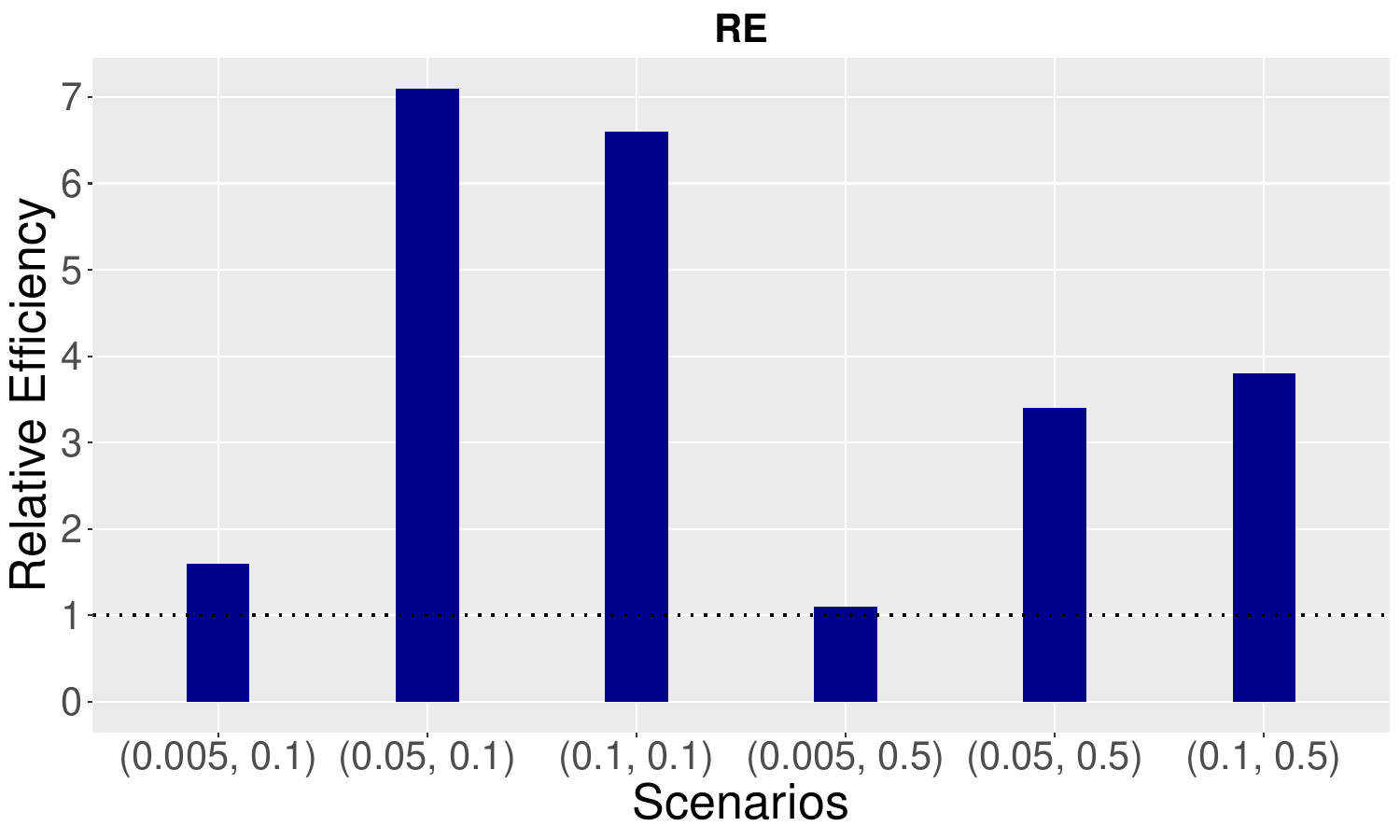}}
\subcaptionbox{}{\includegraphics[scale = 0.225]{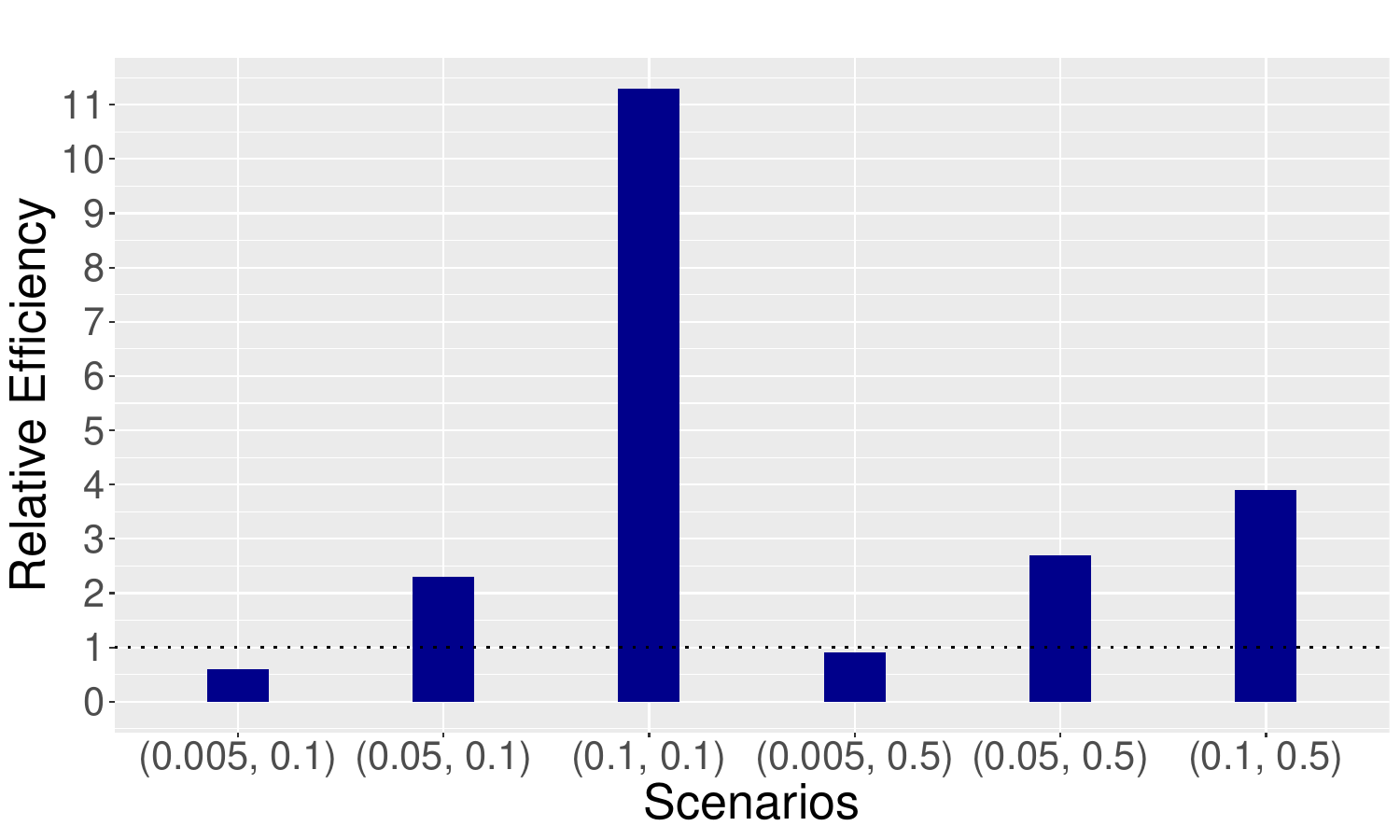}}
\subcaptionbox{}{\includegraphics[scale = 0.225]{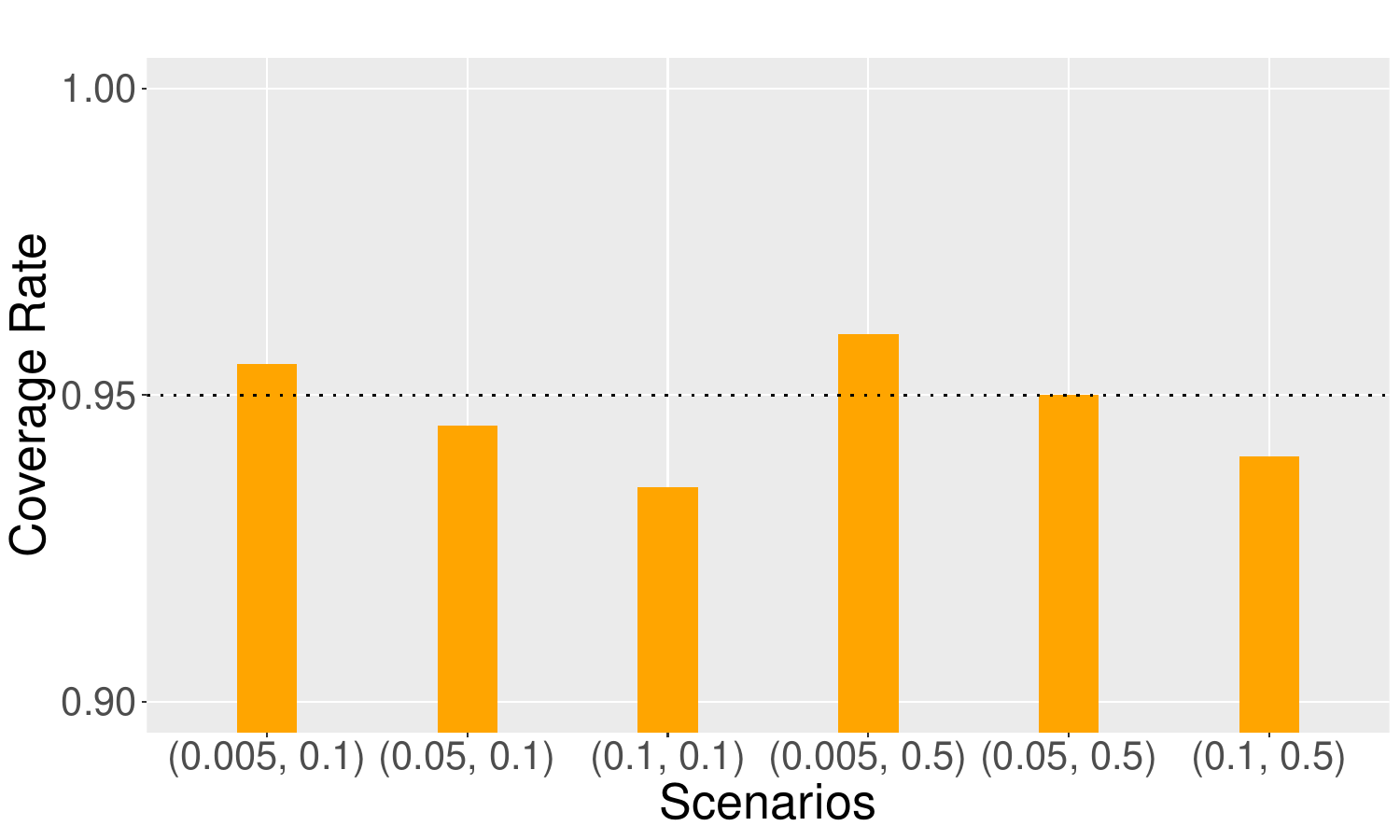}}
\subcaptionbox{}{\includegraphics[scale = 0.225]{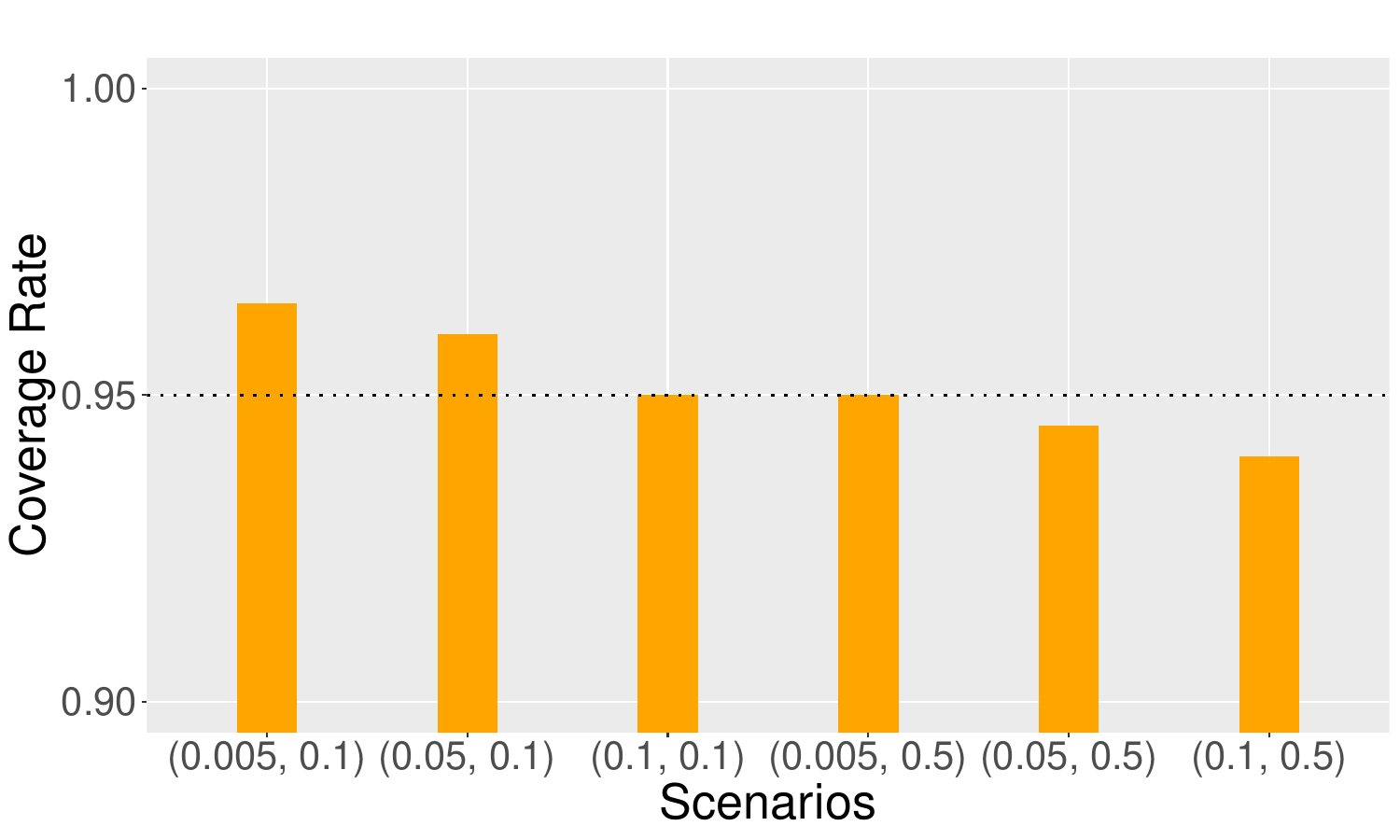}}
\caption{Simulation results of the relative efficiencies of DEEM compared with RAPS and IVW(s) and the 95\% coverage rates of DEEM under different combinations of $(\pi_{c}, h_{c}^{2})$. (a) Relative efficiencies of DEEM with respect to RAPS in the two-sample setting; (b) Relative efficiencies of DEEM with respect to IVW(s) in the one-sample setting; (c) 95\% coverage rates of DEEM in the two-sample setting; (d) 95\% coverage rates of DEEM in the one-sample setting.}\label{fig: RE-CP}
\end{figure} 
Figure~\ref{fig: RE-CP}(a) shows that DEEM consistently outperforms RAPS in terms of efficiency in the two-sample setting. Notably, the relative efficiency of DEEM can exceed six when $5\%$ or $10\%$ of SNPs have effects on $X$ and their effects are weak. Figure~\ref{fig: RE-CP}(b) shows that, compared to IVW(s),  DEEM often has a much higher efficiency in the one-sample setting, with the exception of the case where $\pi_{c} = 0.005$. When SNP effects on $X$ are denser and weak, DEEM’s relative efficiency compared to IVW(s) is notably much higher, reaching $11$ times when $(\pi_{c}, h_{c}^{2}) = (0.1,0.1)$.
%when the effect of SNPs is dense.   
Figure~\ref{fig: RE-CP}(c) and (d) demonstrate that the confidence intervals derived from DEEM have coverage rates close to the nominal level.  
%, and larger than $90\%$ across various scenarios.

In Appendix Section \ref{app: en}, we conduct a comparative analysis of the ensemble estimators $\hbeta_{\rm DEEM, p}$ and $\hbeta_{\rm DEEM, os}$ against their base estimators in the two-sample and one-sample settings, respectively.  The results reveal that the ensemble estimator consistently performs at least
as well as both of its base estimators in terms of bias and SE. The ensemble estimator outperforms
both of its base counterparts when $h_{c}^{2}=
0.1$ in the two sample setting.

Appendix Section \ref{app: sim laplace} contains further simulation results with $\balpha_{G}$ and $\bbeta_{G}$ generated from the Laplace distribution. The findings in Section \ref{app: sim laplace} are consistent with those presented here.
In Appendix Section \ref{app: sim CHP}, we evaluate the robustness of various methods in the presence of pleiotropic SNPs that influence unmeasured confounders, thus violating the assumptions of both Models \eqref{eq: outcome model} and \eqref{eq: pleiotropy outcome model}. Remarkably, DEEM demonstrates strong robustness in these scenarios, outperforming many existing robust MR methods specifically designed to address such violations. This resilience can largely be attributed to DEEM's capacity to incorporate a large number of valid IVs, thereby mitigating the influence of pleiotropic SNPs.

\section{Real Data Analysis}\label{sec: real data}
\subsection{The Effect of LDL-C on CAD}\label{subsec: LDL-CAD}
In this section, we conduct a two-sample MR analysis to study the causal effect of LDL-C on CAD risk.   The supplemental exposure ($X)$ dataset is the GWAS for LDL-C by the Global Lipids Genetics Consortium; the exposure $(X$) dataset is the GWAS for LDL-C from the UK biobank; the outcome ($Y$) dataset is the GWAS for CAD risk by the CARDIoGRAMplusC4D Consortium. Details about the employed GWASs are in Appendix Section \ref{app: package and data}.
We applied all the methods included in the simulation study and used $1.08$ million Hapmap3 SNPs appeared in all three datasets as the candidate IVs. The SNP selection process, including clumping thresholds, aligned with the procedures outlined in Section \ref{sec: sim}. The working covariance matrix $\bbV$ is constructed as per Section \ref{subsec: LD & V} with $c = 1$.
Findings are summarized in Figure~\ref{fig: LDL-CAD}.

\begin{figure}
	\centering
	\subcaptionbox*{}{\includegraphics[scale = 0.5]{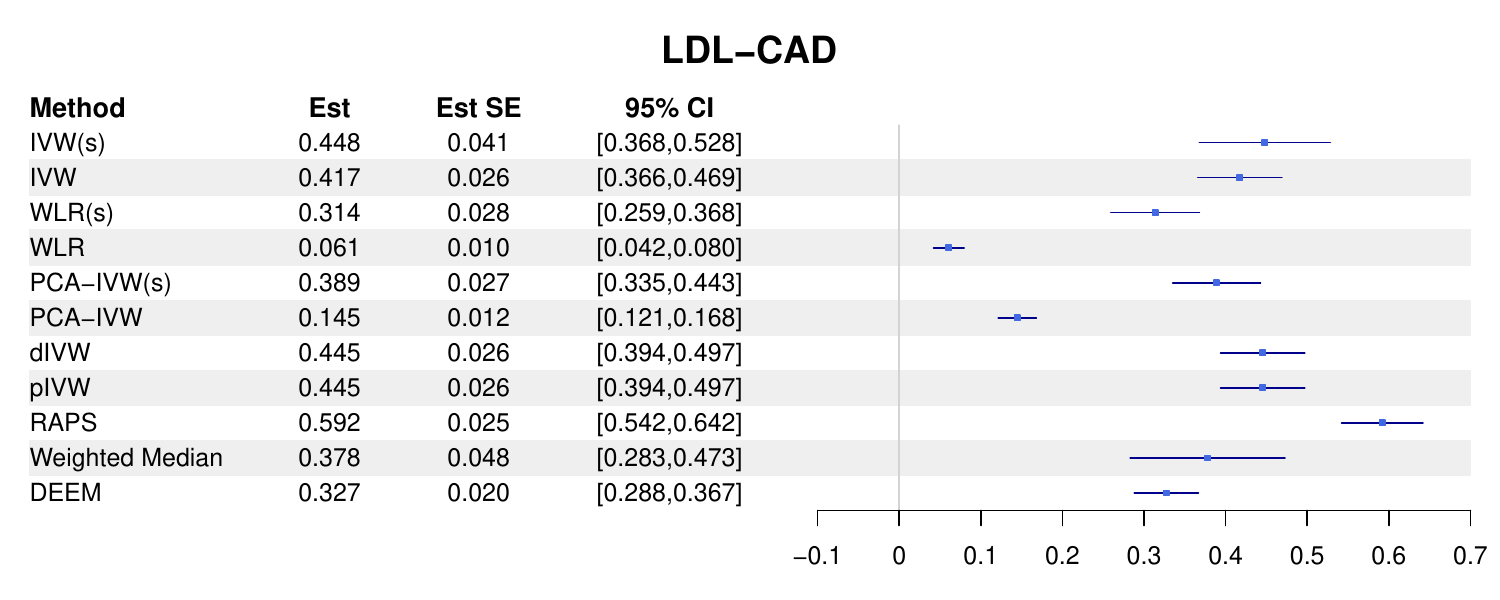}}
	\caption{MR analysis of the effect of LDL-C on the CAD risk. Est: estimated causal effect; Est SE: estimated standard error; CI: confidence interval. The ``s" in the parentheses indicates the method is implemented with the stringent p-value threshold $10^{-4}$.}\label{fig: LDL-CAD}
\end{figure}

All methods indicate a positive causal effect of LDL-C on CAD risk. Notably, WLR and PCA-IVW yield significantly smaller estimates than others, potentially due to their vulnerability to weak IV bias. DEEM, on the other hand, showcases the smallest estimated standard error among all methods except for the two potentially biased methods mentioned above, highlighting its advantage in efficiency through the inclusion of correlated weak IVs. Nonetheless, DEEM's estimate is marginally lower compared to IVW(s), IVW, dIVW, and RAPS. This could stem from a varying degree of sensitivity of different methods to the pleiotropic effects of SNPs on unmeasured confounders, known as "correlated pleiotropy" \citep{cheng2022mendelian}. Specifically, MR analysis assumes that the instruments are independent of unmeasured confounders; violation of this assumption can introduce bias of varying magnitudes across different MR methods.  

%The SNPs with strong and weak effects suggest different causal effects in the presence of correlated pleiotropy as illustrated in Appendix Section \ref{app: sim CHP}.
%We leverage this phenomenon to investigate the presence of correlated pleiotropy. First, we select the SNPs with p-values smaller than $0.1$ based on the supplemental sample. Then, the selected SNPs are divided into ``strong group" and ``weak group", respectively. Specifically, the strong group consists of ten percent of the selected SNPs with the smallest p-values, while the weak group consists of the remaining selected SNPs. Causal effect estimates using RAPS based on the strong and weak group reveal a disagreement, with the estimate based on the strong group being larger ($0.521$) with a confidence interval of $[0.424, 0.618]$ and the estimate based on the weak group being smaller ($0.161$) with a confidence interval of $[-0.024, 0.347]$. This discrepancy could be considered as evidence of correlated pleiotropy, as illustrated in Section \ref{app: sim CHP}. 

To examine this issue, we consider another popular MR method, the weighted median method \citep{bowden2016consistent}, specifically designed to account for possible correlated pleiotropy under specific assumptions. 
Simulations in Appendix Section \ref{app: sim CHP} of suggest that the weighted median method and DEEM are relatively robust in the presence of correlated pleiotropy. These two methods produce similar point estimates for the causal effect of LDL-C on CAD risk. Potential uncontrolled confounders in the univariate MR analysis of LDL-C include high-density lipoprotein cholesterol and triglycerides.  Multivariable MR analysis, as detailed in \cite{rees2017extending}, adjusted for these confounders, estimated the causal effect of LDL-C as $0.375$ with 95$\%$CI $[0.292, 0.457]$. The estimate aligns closely with those obtained by the weighted median method and DEEM in Figure~\ref{fig: LDL-CAD}, underscoring their robustness. Among these two methods, DEEM exhibits a significantly smaller estimated standard error, highlighting its efficiency.   DEEM's robustness in this context can be attributed to its ability to effectively incorporate a large number of valid IVs, thereby mitigating the influence of pleiotropic SNPs that affect uncontrolled confounders.

%	In summary, despite the potential presence of correlated pleiotropy, the consistent findings from various MR methods suggest a significant positive causal effect of LDL-C on CAD risk. These results can serve as evidence supporting the causal relationship between LDL-C and CAD risk, although caution should be exercised in interpreting the results and considering potential biases due to correlated pleiotropy.

\subsection{The Effect of BMI on SBP}\label{subsec: BMI-SBP}
High blood pressure may increase the risk of several common diseases,   such as heart disease and stroke. BMI is positively associated with blood pressure in observational studies  \citep{mokdad2003prevalence}.
Here, we perform an MR analysis to investigate the causal effect of BMI on SBP.   The supplemental exposure ($X$) dataset is the GWAS for BMI conducted by the GIANT Consortium; the exposure $(X$) dataset and the outcome $(Y$) dataset are the GWASs for the corresponding traits from the UK biobank. Details about these GWASs are deferred to Appendix Section \ref{app: package and data}. The analysis is a one-sample MR because both the exposure and outcome data come from the UK biobank.
The analysis is conducted the same way as in Section \ref{subsec: LDL-CAD}, except that here, we adopt the one-sample DEEM.
Figure~\ref{fig: BMI-SBP} shows the results. 

\begin{figure}[h]
	\centering
	\subcaptionbox*{}{\includegraphics[scale = 0.5]{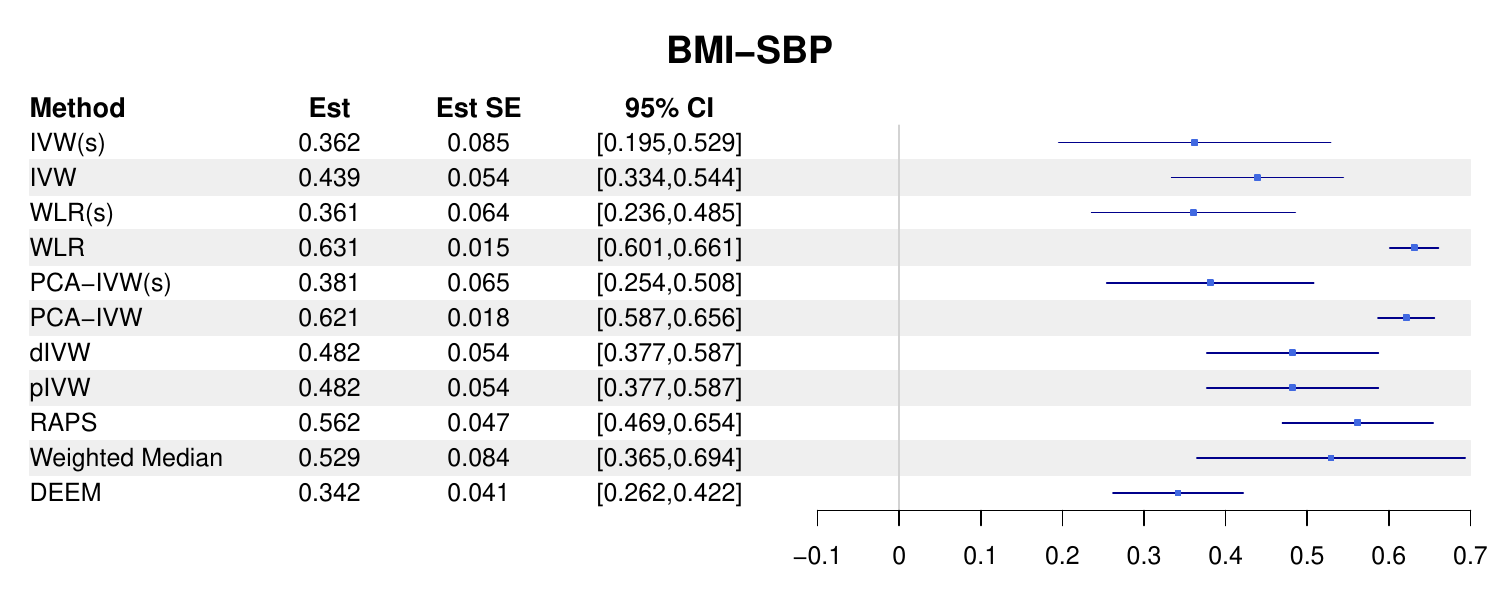}}
	\caption{MR analysis of the effect of BMI on SBP. Est: estimated causal effect; Est SE: estimated standard error; CI: confidence interval. The ``s" in the parentheses indicates the method is implemented with the stringent p-value threshold $10^{-4}$.}\label{fig: BMI-SBP}
\end{figure}
All methods suggest a positive causal effect of BMI on SBP, albeit with varying effect size estimates. IVW(s) and DEEM yield similar point estimates. These two methods exhibit less bias compared to others in the simulation under the one-sample setting in Section \ref{sec: sim}. This indicates a higher reliability of these methods in the real data example compared to other methods. Compared to IVW(s), DEEM incorporates correlated SNPs with weak effects and achieves a higher efficiency.

\section{Discussions}\label{sec: discuss}
This paper proposes DEEM, a new summary statistics-based efficient and robust MR method designed to accommodate a large number of correlated, weak-effect, and invalid SNPs. Traditional MR analysis typically uses only SNPs with strong effects \citep{burgess2013mendelian,bowden2015mendelian}.  Recent advancements allow for the inclusion of SNPs with weak effects but require the included SNPs to be independent \citep{zhao2020statistical,ye2021debiased, xu2022novel}. DEEM overcomes these limitations using a general debiased EE framework that employs a carefully designed decorrelation technique, which can effectively mitigate weak IV bias in the presence of many highly correlated SNPs. DEEM can achieve a significant efficiency gain compared to existing MR methods by harnessing information from a large number of correlated, weak and invalid IVs. In addition, the DEEM ensemble estimator leverages a supplemental sample to avoid the winner's curse 
%by selecting SNPs 
and enhance efficiency by utilizing an additional EE.
%that exploits information from this supplemental sample. 
%DEEM was first developed for two-sample MR analysis, adhering to standard IV assumptions. However, its flexibile EE framework allows for  extension to one-sample settings as well as  accommodation of pleiotropic effects.
%These favorable characteristics make DEEM a valuable tool with potential applications across various fields, including epidemiology, genetics, and medical research, where MR is frequently employed for causal inference. We establish the asymptotic normality of the proposed estimator under various settings. Numerical simulations underscore DEEM's efficiency gain and robustness compared with the existing methods.
The general summary statistics-based debiased EE framework that underpins the development of DEEM provides a clear lens to examine the weak IV problem and offers a unified basis for motivating and analyzing various existing methods, including TSLS, IVW, WLR, PCA-IVW, and RAPS. It also potentially inspires new methodologies in areas such as multivariate MR and high-dimensional regression.

DEEM addresses the winner's curse by selecting SNPs using an independent supplemental sample. An alternative approach without using a supplemental sample is to employ randomized IV selection combined with Rao–Blackwellization \citep{ma2023breaking}. However, the effectiveness of this method hinges on accurate covariance estimation of $\hgamma$, which becomes particularly challenging when dealing with a large number of correlated SNPs. Additionally, the Rao–Blackwellization step in \cite{ma2023breaking} requires an analytical characterization of the selection event that is highly complex when multiple SNPs are selected from a single LD block. It is of future research interest to  develop alternative approaches to address the winner's curse without a supplemental sample.
%remain critical directions for future research.}

DEEM assumes a linear relationship between exposures and outcomes.
Future research could enhance DEEM by accommodating nonlinear exposure effects. DEEM assumes that the components of the direct effect vector $\bbeta_{G}$ are i.i.d. Extension of the results to settings where these components are non-i.i.d. is left for future research. While DEEM accounts for the direct effects of SNPs on the outcome, its estimates may be biased by correlated pleiotropy, where some SNPs are correlated with unmeasured confounders. Potential approaches for mitigating this issue include multivariable MR \citep{rees2017extending} and robustification of the estimation procedure through loss functions such as the Huber or Tukey loss. We propose this as a direction for future research. 	

DEEM focuses on estimating causal effect sizes assuming a known causal graph. In practice, causal relationships among variables may be unknown. Recently, \citet{WangQZ23} address the open issue of deriving causal effect bounds with unknown causal graphs and latent variables \citep{Maathuis2009estimating} and propose an algorithm that avoids enumerating causal graphs. However, statistical inference for such scenarios remains understudied. Exploring DEEM's potential in statistical inference with an unknown causal graph represents an intriguing future direction.

\noindent\emph{Conflict of interest:} None declared.

\clearpage

\appendix
{\noindent\Large \bf Appendix}
\vspace{20pt}

\renewcommand{\thesection}{A\arabic{section}}
\renewcommand{\thecondition}{A\arabic{condition}}
\renewcommand{\thetheorem}{A\arabic{theorem}}
\renewcommand{\thetable}{A\arabic{table}}
\renewcommand{\thefigure}{A\arabic{figure}}
\renewcommand{\theexample}{A\arabic{example}}
\renewcommand{\theproposition}{A\arabic{proposition}}
\renewcommand{\thelemma}{A\arabic{lemma}}

\section{The Weak IV Bias}\label{app: weak IV}
	Compared to the infeasible EE $\bgamma^{\T}\bbSig^{-1}(\hGamma - \beta \hgamma) = 0$, Equation \eqref{eq: plug-in EE} replaces the unknown $\bgamma$ by its estimator $\hgamma$.  We find that the estimation error of $\hgamma$ can deteriorate the performance of the resulting estimator and lead the resulting estimator to be biased in the presence of weak IVs. We delve into this issue in this section.
	
	For any matrix $\bbA$, let $\tr\{\bbA\}$ be its trace, while $\lambda_{\rm min}(\bbA)$ and $\lambda_{\rm max}(\bbA)$ represent its minimum and maximum singular values, respectively. Let $\|\cdot\|$ be the Euclidean/spectral norm of a vector/matrix. 
	For simplicity, we first assume $\bbV$ is known. Theoretical results involving a data-dependent working matrix $\bbV$ are postponed to Section \ref{subsec: AN and combine}.  
	Throughout the analysis, we assume $\lambda_{\rm max}(\bbV) / \lambda_{\rm min}(\bbV)$ is bounded.
	Upon examining the EE \eqref{eq: plug-in EE}, we observe that the estimation error in $\hgamma$ disrupts this unbiasedness, particularly when $\betax$ is nonzero. To elaborate, consider:
	\begin{equation}\label{eq: weak IV bias S}
		\nodisplayskips
		\begin{aligned}
			E\{\hgamma^{\T}\bbV^{-1}(\hGamma - \betax \hgamma)\} 
			& = \tr\{\bbV^{-1}\cov(\hgamma, \hGamma - \betax \hgamma)\} 
			= -\betax\tr\{\bbV^{-1}\bbSig_{\gamma}\} \neq 0, 
		\end{aligned}
	\end{equation}
	where the last equality holds because $\hgamma$ and $\hGamma$ are independent in the two-sample setting. Under regularity conditions, Equation \eqref{eq: weak IV bias S} implies the solution $\hbeta_{\rm PlugIn} = \hgamma^{\T}\bbV^{-1}\hGamma / \hgamma^{\T}\bbV^{-1}\hgamma$ of Equation \eqref{eq: plug-in EE} has an asymptotic bias
	\begin{equation}\label{eq: b-plug}
		\nodisplayskips
		- \frac{\betax\tr\{\bbV^{-1}\bbSig_{\gamma}\}}{\bgamma^{\T}\bbV^{-1}\bgamma + \tr\{\bbV^{-1}\bbSig_{\gamma}\}} = - \frac{\betax}{\Theta(\|\bgamma\|^{2} / \tr\{\bbSig_{\gamma}\}) + 1},
	\end{equation}
	where $\Theta(\|\bgamma\|^{2} / \tr\{\bbSig_{\gamma}\})$ means the term belongs to $\left[c^{-1}\|\bgamma\|^{2} / \tr\{\bbSig_{\gamma}\}, c\|\bgamma\|^{2} / \tr\{\bbSig_{\gamma}\}\right]$ for some universal constant $c > 1$. See Section \ref{app: proof of sec weak IV} for the proof of \eqref{eq: b-plug}.
	The bias term in \eqref{eq: b-plug} is negligible when SNP effects are significantly larger than the variance of $\hgamma$ because it converges to zero provided $\|\bgamma\|^{2} / \tr\{\bbSig_{\gamma}\} \to \infty$.
	In such cases, both the TSLS estimator \citep{hayashi2011econometrics} and the IVW estimator \citep{burgess2013mendelian}, as special instances of $\hbeta_{\rm PlugIn}$, are consistent and asymptotically normal \citep{hayashi2011econometrics, ye2021debiased}. Furthermore, if $\bbV^{-1}$ is a consistent estimator for $\bbSig^{-1}$, $\hbeta_{\rm PlugIn}$ exhibits the same asymptotic variance as the solution of the oracle optimal EE and hence is efficient. 
	
	Nevertheless, the effects of SNPs are often small in practice, especially when the trait is highly polygenic. %\citep{visscher201710}
	When the magnitude of $\bgamma$ is similar to the variance $\bbSig_{\gamma}$, the bias term \eqref{eq: b-plug} does not converge to zero, and hence $\hbeta_{\rm PlugIn}$ is inconsistent, which leads to the weak IV problem \citep{bound1995problems}.    
	Practitioners usually select strong IVs based on criteria such as F-statistics to ensure the condition $\|\bgamma\|^{2} / \tr\{\bbSig_{\gamma}\} \to \infty$. However, this approach, as noted in Section \ref{sec: intro}, can lead to a significant efficiency loss. These observations motivate us to modify the estimation \eqref{eq: plug-in EE} to effectively address the weak IV bias.

	\section{Selection of SNPs based on the Supplemental Sample}\label{app: selection}
	Suppose $\tgamma = (\tilde{\gamma}_{1}, \dots, \tilde{\gamma}_{d})^{\T}$ is the estimate of $\bgamma$ derived from the supplemental sample, with $\tilde{\sigma}_{\gamma j}^{2}$ denoting the estimated variance of $\tilde{\gamma}_{j}$ for $j=1, \dots, d$. Then we can calculate the p-value $\tilde{p}_{j} = 2\{1 - \Phi(|\tilde{\gamma}_{j}/\tilde{\sigma}_{\gamma j}|)\}$ of the association test for the $j$th SNP, where $\Phi$ is the distribution function of the standard normal distribution. SNPs are then filtered based on a significance threshold $c_{\rm sel}$, with those having p-values less than $c_{\rm sel}$ being retained. Specifically, let $\widetilde{\cA} = \{j: \tilde{p}_{j} \leq c_{\rm sel},\ j = 1,\dots, d\}$. Only SNPs within $\widetilde{\cA}$ are included in the subsequent MR analysis. The choice $c_{\rm sel} = 0.1$ performs well in our simulation studies and real data examples. The selection procedure aims to exclude the bulk of SNPs that do not contribute meaningful information to the MR analysis to improve efficiency. Notice that Equation \eqref{eq: population equation}, which is the cornerstone of our approach, remains valid even when some components of $\bgamma$ are equal to zero. Thus, the consistency of the selection procedure is not necessary for the consistency of the resulting estimator of $\betax$. On the other hand, including SNPs with zero effect does not gain additional information but introduces more estimation error, which may affect the statistical efficiency for estimating $\betax$. The selection procedure can mitigate this problem. 

    \section{Some Remarks}
    \begin{remark}
In conventional MR studies, only uncorrelated SNPs with strong effects are included in the analysis \citep{burgess2013mendelian}. Suppose $\bgamma = (\bgamma_{1}^{\T}, \bgamma_{2}^{\T})^{\T}$ and $\bGamma = (\bGamma_{1}^{\T}, \bGamma_{2}^{\T})^{\T}$, where $\bgamma_{1}$ and $\bGamma_{1}$ are the marginal regression coefficients of uncorrelated strong IVs, and $\bgamma_{2}$ and $\bGamma_{2}$ are the coefficients of the rest of the IVs. Rewrite $\bbSig$ into the corresponding block form
\[
\nodisplayskips
\bbSig = 
\begin{pmatrix}
\bbSig_{11} & \bbSig_{12}\\
\bbSig_{21} & \bbSig_{22}
\end{pmatrix}.
\]
Then the solution of the optimal EE based on only uncorrelated strong IVs has the asymptotic variance $(\bgamma_{1}^{\T}\bbSig_{11}^{-1}\bgamma_{1})^{-1}$. By some algebra, we have
\[
\nodisplayskips
\begin{aligned}
\bgamma^{\T}\bbSig^{-1}\bgamma - \bgamma_{1}^{\T}\bbSig_{11}^{-1}\bgamma_{1} 
& = (\bbSig_{21}\bbSig_{11}^{-1}\bgamma_{1}  - \bgamma_{2})^{\T}\bbSig^{22}(\bbSig_{21}\bbSig_{11}^{-1}\bgamma_{1}  - \bgamma_{2}) \geq 0,
\end{aligned}
\]
where $\bbSig^{22} = (\bbSig_{22} - \bbSig_{21}\bbSig_{11}^{-1}\bbSig_{12})^{-1}$, which implies $(\bgamma^{\T}\bbSig^{-1}\bgamma)^{-1} \leq (\bgamma_{1}^{\T}\bbSig_{11}^{-1}\bgamma_{1})^{-1}$. 
This demonstrates the potential efficiency gain of incorporating correlated weak IVs.
\end{remark}

\begin{remark}
The parameters $\gamma_{j}= \var(G_{j})^{-1}\cov(G_{j}, X)$
and
$\Gamma_{j} = \var(G_{j})^{-1}\cov(G_{j}, Y)$ 
are well-defined without assuming the linear models. Moreover, model \eqref{eq: outcome model} together with the IV exogeneity assumptions implies $\Gamma_{j}  = \var(G_{j})^{-1}\cov(G_{j}, Y)  = \betax \var(G_{j})^{-1}\cov(G_{j}, X) = \betax \gamma_{j}$ and hence the relationship \eqref{eq: population equation} holds irrespective of whether the relationship between $X$ and $\bG$ is linear. The linear $X$ -- $\bG$ model is presented here primarily for the ease of illustration. 
%If  covariates are included in Model \eqref{eq: outcome model}, the relationship \eqref{eq: population equation} is preserved if these covariates are not correlated with $\bU$ and are controlled for in both $Y$ -- $\bG$ and $X$ -- $\bG$ regressions.
\end{remark}
	
	\section{Notations and Regularity Conditions}\label{app: notations and conds}
	Recall that  $\widetilde{\cA}$ is the set of SNPs selected by the supplemental sample. Hereafter, we use $q = |\widetilde{\cA}|$ to denote the number of selected SNPs. $\tgamma$, $\hgamma$ and $\hGamma$ among other quantities implicitly depend on $\widetilde{\cA}$. Here $d$ continues to denote the total number of SNPs: $\bbeta_G\in\mathbb R^d$, $\bbP\in\mathbb R^{q\times d}$ consists of the selected rows of the full-SNP mapping, and $\bbSig_{\rm p}=\bbP\bbP^{\T}$ is $q\times q$. In this section, we allow $\bbV$ to be data dependent and assume it converges to some matrix $\bbV_{*}$ which is nonrandom conditional on $\widetilde{\cA}$. With some abuse of notation, we define $\bbD_{\gamma}$, $\bbD_{\Gamma}$, and $\bbD_{\rm p}$ based on $\bbV_{*}$ rather than $\bbV$ in the analysis of asymptotic properties, that is,
	$\bbD_{\gamma} = \mathop{\arg\min}_{\bbD \in \cD} \langle\bbSig_{\gamma} - \bbD, \bbSig_{\gamma} - \bbD\rangle_{\bbV_{*}}$, $\bbD_{\Gamma} = \mathop{\arg\min}_{\bbD \in \cD} \langle\bbSig_{\Gamma} - \bbD, \bbSig_{\Gamma} - \bbD\rangle_{\bbV_{*}}$
	and $\bbD_{\rm p} = \mathop{\arg\min}_{\bbD \in \cD} \langle \bbSig_{\rm p} - \bbD, \bbSig_{\rm p} - \bbD\rangle_{\bbV_{*}}$.
	Recall that $n_{e}$ and $n_{o}$ are the sample sizes of the exposure, and the outcome sample that produce  $\hgamma$, and $\hGamma$, respectively. Assume $n_{e}/n_{o}$ is bounded away from zero and infinity. Let $n = \min\{n_{e}, n_{o}\}$. In the following, we use $c$ and $C$ to symbolize absolute constants which are  possibly different in different places. For two sequences of positive numbers $\{a_{n}\}_{n=1}^{\infty}$ and $\{b_{n}\}_{n=1}^{\infty}$, we say $a_{n} \asymp b_{n}$ or $a_{n} = \Theta(b_{n})$ if $c^{-1}a_{n} \leq b_{n} \leq c a_{n}$ for some constant $c > 1$. For any square matrix $\bbA$, let $\bbA^{+}$ be the Moore-Penrose inverse of $\bbA$ and $[\bbA]_{ij}$ the $(i,j)$-th element of $\bbA$. As in Section \ref{subsec: set up}, we use $\bbSig$ to denote $\var(\hGamma - \betax\hgamma)$, which equals $\bbSig_{\Gamma} + \taup\bbSig_{\rm p} + (\betax^{2} - 2\rho_{U}\betax)\bbSig_{\gamma}$. We have $\taup = 0$ in the setting without pleiotropy and $\rho_{U} = 0$ in the two sample setting. Throughout the analysis, we allow all the quantities to change with $n$ except for the constants in the conditions and assume $\betax$ and $\rho_{U}$ are bounded as $n \to \infty$. Now we are ready to state the conditions required to establish the theoretical properties. For conditions that involve random quantities, we always assume the
    condition holds almost surely and omit ``almost surely" in the statement for clarity.
	%    The theoretical analysis is conducted conditional on the supplemental sample. Conditional on the supplemental sample, $\tgamma$ and $q$ are regarded as nonrandom and involved in the conditions. The conditions which involve these two quantities are assumed to hold with probability approaching one, and all the asymptotic convergence results hold in probability \citep{sweeting1989conditional}. }   

    \begin{condition}\label{cond: matrices origin}
    	For some positive constants $c$ and $C$, it holds that
    	(i) $\tr\{\bbSig_{\rm \gamma}\} \geq cq\|\bbSig_{\rm \gamma}\|$, $\tr\{\bbSig_{\rm \Gamma}\} \geq cq\|\bbSig_{\rm \Gamma}\|$ and $\|\bbSig_{\gamma}\|, \|\bbSig_{\Gamma}\| \asymp n^{-1}$; (ii) $\lambda_{\rm max}(\bbV_{*}^{-1}) / \lambda_{\rm min}(\bbV_{*}^{-1}) \leq C$ and $\|\bbV_{*}\| \asymp n^{-1}$; (iii) $\bbD_{\gamma}$ and $\bbD_{\Gamma}$ are positive definite, $\lambda_{\rm max}(\bbD_{\gamma}) / \lambda_{\rm min}(\bbD_{\gamma}) \leq C$, $\lambda_{\rm max}(\bbD_{\Gamma}) / \lambda_{\rm min}(\bbD_{\Gamma}) \leq C$, $\|\bbD_{\gamma}\| \asymp \|\bbSig_{\gamma}\|$ and $\|\bbD_{\Gamma}\| \asymp \|\bbSig_{\Gamma}\|$;
    	(iv) $\max\big\{\|\hD_{\gamma} - \bbD_{\gamma}\| / \|\bbD_{\gamma}\|$, $\|\hD_{\Gamma} - \bbD_{\Gamma}\| / \|\bbD_{\Gamma}\|, \|\bbV^{-1} - \bbV_{*}^{-1}\| / \|\bbV^{-1}\|\big\} = O_{P}(n^{-\kappa})$ for some $\kappa > 0$.
    \end{condition}
    
    \begin{condition}\label{cond: selection}
    	There are positive sequences $\{q_{1n}\}_{n=1}^{\infty}$ and $\{q_{2n}\}_{n=1}^{\infty}$ such that $P(q_{1n} \leq q \leq q_{2n}) \to 1$; $q_{1n} \to \infty$, $n^{-\kappa}\sqrt{q_{2n}}\to 0$, and $q_{2n} \leq C n\|\bgamma\|^{2}$ for some constant $C > 0$, where $\kappa$ is the constant defined in Condition \ref{cond: matrices origin}.
    \end{condition}
    
    \begin{condition}\label{cond: angle}
    	For some constant $c>0$, we have (i) $P\left(|\bgamma^{\T}\bbV_{*}^{-1}\tgamma| > c \sqrt{\bgamma^{\T}\bbV_{*}^{-1}\bgamma}\sqrt{\tgamma^{\T}\bbV_{*}^{-1}\tgamma}\right) \to 1$; (ii) $P\left(|\tgamma^{\T}\bbV_{*}^{-1}\bbSig\bbV_{*}^{-1}\tgamma| > c \|\bbSig\|\|\bbV_{*}^{-1}\tgamma\|^{2}\right) \to 1$; (iii) $\bgamma^{\T}\bbV_{*}^{-1}\bbSig\bbV_{*}^{-1}\bgamma > c\|\bbSig\|\|\bbV_{*}^{-1}\bgamma\|^{2}$.
    \end{condition}
    \begin{condition}\label{cond: balance origin}
    	$\tr\{\bbH^{2}\} \to \infty$ where $\bbH = (\bbA_{\rm H} + \bbA_{\rm H}^{\T}) / 2$ and
    	\[
    	\bbA_{\rm H} = \begin{pmatrix}
    		-\betax \bbSig_{\gamma}^{\frac{1}{2}}\{\bbI + \betax \bbQ(\betax)\}^{\T}\bbV_{*}^{-1}\bbSig_{\gamma}^{\frac{1}{2}} 
    		& \bbSig_{\gamma}^{\frac{1}{2}}\{\bbI + \betax \bbQ(\betax)\}^{\T}\bbV_{*}^{-1}\bbSig_{\Gamma}^{\frac{1}{2}}\\
    		\betax \bbSig_{\Gamma}^{\frac{1}{2}} \bbQ(\betax)^{\T}\bbV_{*}^{-1}\bbSig_{\gamma}^{\frac{1}{2}}
    		& - \bbSig_{\Gamma}^{\frac{1}{2}}\bbQ(\betax)^{\T}\bbV_{*}^{-1}\bbSig_{\Gamma}^{\frac{1}{2}}
    	\end{pmatrix}.
    	\]
    \end{condition}
    Theorem \ref{thm: AN origin} is established under Conditions \ref{cond: matrices origin}--\ref{cond: balance origin}. 
    Condition \ref{cond: matrices origin} imposes some regularity conditions on $\bbSig_{\gamma}$, $\bbSig_{\Gamma}$, $\bbD_{\gamma}$, $\bbD_{\Gamma}$, $\bbV_{*}$ and their estimators.
    Note that each component of $\hGamma$ and $\hgamma$ is $\sqrt{n}$-consistent, which justifies the rate of the covariance matrices in Condition \ref{cond: matrices origin} (i). Condition \ref{cond: selection} restricts the number of selected SNPs. Conditions \ref{cond: matrices origin} and \ref{cond: selection} are needed to control the estimation error in the asymptotic expansions of $\hbeta_{1}$ and $\hbeta_{2}$. Condition \ref{cond: angle} (i) requires $\tgamma$ not to be nearly orthogonal to $\bgamma$, which is a reasonable requirement as $\tgamma$ is an estimate for $\bgamma$. Condition \ref{cond: angle} (ii) and (iii) require $\bbV_{*}^{-1}\tgamma$ and $\bbV_{*}^{-1}\bgamma$ not to fall very close to the null space of $\bbSig$. The requirement is mild and imposes no restriction on $\tgamma$ and $\bgamma$ if $\bbSig$ is well-conditioned. Condition \ref{cond: angle} is adopted to establish the asymptotic expansion of $\hbeta_{1}$ and $\hbeta_{2}$. Condition \ref{cond: balance origin} is a mild regularity condition that requires the trace of the semipositive definite matrix $\bbH^{2}$ with diverging dimension to diverge to infinity, which is needed to establish the asymptotic normality of $\enbeta_{\rm DEEM}$.

    \begin{condition}\label{cond: matrices pleiotropy}
    	For some positive constants $c$ and $C$,
    	(i) $\tr\{\bbSig_{\rm p}\} \geq cq\|\bbSig_{\rm p}\|$, $\tr\{\bbSig\} \geq cq\|\bbSig\|$;
    	(ii) $\sqrt{\tr\{\bbSig\}} / \tr\{\bbSig_{\rm p}\} \to 0$;
    	(iii) $\|\bbD_{\rm p}\| \asymp \|\bbSig_{\rm p}\|$ and $\lambda_{\rm max}(\bbD_{\rm p}) / \lambda_{\rm min}(\bbD_{\rm p}) \leq C$; (iv) $\|\hD_{\rm p} - \bbD_{\rm p}\| / \|\bbD_{\rm p}\| = o_{P}(n^{-\kappa})$.
    \end{condition}
    
    \begin{condition}\label{cond: bound pleiotropy}
    	For some constants $C > c > 1$, (i) components of $\bbeta_{G}$ have finite fourth moment $c\taup^{2} \leq \tau_{4, \rm p} \leq C\taup^{2}$;
    	(ii) $\|\bgamma\| \leq C$, $\taup \leq C$, and $\taup\|\bbSig_{\rm p}\| / \|\bgamma\|^{2} \to 0$.
    \end{condition}
    
    Next we introduce a useful quantity in the following analysis.
    For integer $m$ and any $m \times m$ matrix $\bbA$ and $j = 1,\dots,m$, let $b_{j}(\bbA) = |a_{jj}| + \sum_{i\neq j}a_{ij}^2 / \{\sum_i|a_{ii}| + (\sum_{i}\sum_{k\neq i}a_{ik}^{2})^{1/2}\}$, where $a_{ik}$ is the $ik$th element of $\bbA$ for $i,k = 1,\dots,m$. Define $\scB(\bbA) = \max_{j}b_{j}(\bbA) / \sum_{j}b_{j}(\bbA)$, with $b_j(\mathbf 0)=\scB(\mathbf 0)=0$. {
In Conditions \ref{cond: linearity pleiotropy} and \ref{cond: linearity os}, the ratio $\|\bbH\|/\sqrt{\tr\{\bbH^2\}}$ is defined as zero when $\bbH=\mathbf0$, and the comparison $0\asymp0$ is understood as satisfied.
} Here, $b_{j}(\bbA)$ can be regarded as a numerical characteristic of the scale of the $j$th column of $\bbA$. The quantity $\scB(\bbA)$ is a numerical characteristic of the evenness among different columns of $\bbA$. $\scB(\bbA)$ is close to $1/m$ if different columns are similar in size. 
    
    Let 
    \[
    \begin{aligned}
    	&\bl_{1\rm p} = 
    	\begin{pmatrix}
    		-\betax\bbSig_{\gamma}^{\frac{1}{2}} & \bbSig_{\Gamma}^{\frac{1}{2}} & \sqrt{\taup}\bbP
    	\end{pmatrix}^{\T}\bbV_{*}^{-1}\bgamma,\\
    	&\bl_{2\rm p} = 
    	\begin{pmatrix}
    		-\betax\bbSig_{\gamma}^{\frac{1}{2}} & \bbSig_{\Gamma}^{\frac{1}{2}} & \sqrt{\taup}\bbP
    	\end{pmatrix}^{\T}\bbV_{*}^{-1}\tgamma,\\
    	&\bbL_{\rm p} = 
    	\begin{pmatrix}
    		\bl_{1\rm p}^{\T}\bl_{1\rm p} & \bl_{1\rm p}^{\T}\bl_{2\rm p}\\
    		\bl_{2\rm p}^{\T}\bl_{1\rm p} & \bl_{2\rm p}^{\T}\bl_{2\rm p}
    	\end{pmatrix},
    \end{aligned}
    \]
    and
    \[
    \begin{aligned}
    	&\bbH_{\rm p} = \bbH_{\rm p, 1} + \frac{\mu_{\tau, \rm p}}{\tr\{\bbV_{*}^{-1}\bbD_{\rm p}\}}\bbH_{\rm p, 2},
    \end{aligned}
    \]
    where $\bbH_{\rm p, 1} = (\bbA_{\rm H, p} + \bbA_{\rm H, p}^{\T}) / 2$,
    \[
    \bbA_{\rm H, p} = 
    \begin{pmatrix}
    	\bbSig_{\gamma}^{\frac{1}{2}} + \betax \bbSig_{\gamma}^{\frac{1}{2}} \bbQ(\betax; \taup)^{\T}\\
    	-\bbSig_{\Gamma}^{\frac{1}{2}}\bbQ(\betax; \taup)^{\T}\\
    	-\sqrt{\taup}\bbP^{\T}\bbQ(\betax; \taup)^{\T}
    \end{pmatrix}
    \bbV_{*}^{-1}
    \begin{pmatrix}
    	-\betax\bbSig_{\gamma}^{\frac{1}{2}}\\
    	\bbSig_{\Gamma}^{\frac{1}{2}}\\
    	\sqrt{\taup}\bbP^{\T}
    \end{pmatrix}^{\T},
    \]
    $\mu_{\tau, \rm p} = - \betax\tr\{\bbV_{*}^{-1}\bbD_{\rm p}\bbD_{\rm F, p}^{-1}\bbD_{\gamma}\}$, $\bbD_{\rm F, p} = \bbD_{\Gamma} + \betax^{2}\bbD_{\gamma} + \taup \bbD_{\rm p}$
    and
    \[
    \bbH_{\rm p, 2} = 
    \begin{pmatrix}
    	- \betax \bbSig_{\gamma}^{\frac{1}{2}}\\
    	\bbSig_{\Gamma}^{\frac{1}{2}}\\
    	\sqrt{\taup}\bbP^{\T}
    \end{pmatrix}\bbV_{*}^{-1}
    \begin{pmatrix}
    	- \betax \bbSig_{\gamma}^{\frac{1}{2}}\\
    	\bbSig_{\Gamma}^{\frac{1}{2}}\\
    	\sqrt{\taup}\bbP^{\T}
    \end{pmatrix}^{\T}.
    \]
    {
For $j=1,\ldots,2q+d$, define
$\boldsymbol v_{{\rm p},j}=(\ell_{1{\rm p},j},\ell_{2{\rm p},j})^{\T}$,
where $\ell_{k{\rm p},j}$ is the $j$th component of $\bl_{k{\rm p}}$, and let
\[
\mathbf K_{\rm p}=\bbL_{\rm p}+2\tr\{\bbH_{\rm p}^{2}\}
\begin{pmatrix}1&0\\0&0\end{pmatrix}.
\]
}
    \begin{condition}\label{cond: linearity pleiotropy}
    	{
(i) $\displaystyle\max_{2q<j\leq2q+d}\boldsymbol v_{{\rm p},j}^{\T}\mathbf K_{\rm p}^{+}\boldsymbol v_{{\rm p},j}\to0$;
} (ii) $\|\bbH_{\rm p}\| / \sqrt{\tr\{\bbH_{\rm p}^{2}\}} \to 0$, $\sqrt{\tr\{\bbH_{\rm p}^{2}\}} \asymp \sum_{j}b_{j}(\bbH_{\rm p})$ and $\scB(\bbH_{\rm p}) \to 0$; {
(iii) for the weight sequences under consideration, $\bw^{\T}\bbPsi_{\rm p}\bw>0$ eventually, where $\bbPsi_{\rm p}$ is defined in Section \ref{app: proof of sec p}.
}
    \end{condition}
    {
In addition, assume that the condition number of $\bbPsi_{\rm p}$ is bounded.
}
Theorem \ref{thm: AN pleiotropy} can be proved under Conditions \ref{cond: matrices origin}, \ref{cond: selection}, \ref{cond: angle}, \ref{cond: matrices pleiotropy}, \ref{cond: bound pleiotropy}, and \ref{cond: linearity pleiotropy}. Conditions \ref{cond: matrices origin}, \ref{cond: selection}, \ref{cond: angle}, \ref{cond: matrices pleiotropy}, and \ref{cond: bound pleiotropy} are adopted to establish the asymptotic expansions of $\hbeta_{1, \rm p}$ and $\hbeta_{2}$. Condition \ref{cond: linearity pleiotropy} {
controls the normalized non-Gaussian linear contributions and the quadratic contributions in the asymptotic expansion of $\enbeta_{\rm DEEM, p}$
}, which is essential in establishing the asymptotic normality. 
    
    Next, we introduce some conditions for Theorem \ref{thm: AN os}.
    Let $\bbD_{\rm F, os} = \bbD_{\Gamma} + (\betax^{2} - 2 \rho_{U}\betax)\bbD_{\gamma} + \taup \bbD_{\rm p}$.
    \begin{condition}\label{cond: matrices os}
{
For some positive constants $c,C$, (i) $\lambda_{\rm max}(\bbD_{\rm F,os})/\lambda_{\rm min}(\bbD_{\rm F,os})\le C$ and $\|\bbD_{\rm F,os}\|\asymp\|\bbSig\|$; (ii) $\|\bbSig\|\ge c/n$ and
\[
 R_n=\frac{n\|\bgamma\|\sqrt{\|\bbSig\|}}q\longrightarrow0;
\]
(iii) the condition number of $\bbPsi_{\rm os}$ is bounded.
}
    \end{condition}
    Let 
    \[
    \begin{aligned}
    	&\bl_{1\rm os} = 
    	\begin{pmatrix}
    		(\rho_{U} - \betax)\bbSig_{\gamma}^{\frac{1}{2}} & (\bbSig_{\Gamma} - \rho_{U}^{2}\bbSig_{\gamma})^{\frac{1}{2}} & \sqrt{\taup}\bbP
    	\end{pmatrix}^{\T}\bbV_{*}^{-1}
    	\\
    	&\phantom{\bl_{1\rm os} = }
    	\left[\left(1 + \frac{\mu_{\tau, \rm os} \mu_{\rho, \tau} + \mu_{\rho}}{\tr\{\bbV_{*}^{-1}\bbD_{\gamma}\}}\right)\bgamma + \frac{\left(\mu_{\tau, \rm os} \mu_{\rho, \tau} + \mu_{\rho}\right)\mu_{\beta}}{\tgamma^{\T}\bbV_{*}^{-1}\bgamma}\tgamma\right],\\
    	&\bl_{2\rm os} = 
    	\begin{pmatrix}
    		(\rho_{U} - \betax)\bbSig_{\gamma}^{\frac{1}{2}} & (\bbSig_{\Gamma} - \rho_{U}^{2}\bbSig_{\gamma})^{\frac{1}{2}} & \sqrt{\taup}\bbP
    	\end{pmatrix}^{\T}\bbV_{*}^{-1}\tgamma,\\
    	&\bbL_{\rm os} = 
    	\begin{pmatrix}
    		\bl_{1\rm os}^{\T}\bl_{1\rm os} & \bl_{1\rm os}^{\T}\bl_{2\rm os}\\
    		\bl_{2\rm os}^{\T}\bl_{1\rm os} & \bl_{2\rm os}^{\T}\bl_{2\rm os}
    	\end{pmatrix},
    \end{aligned}
    \]
    and
    \[
    \begin{aligned}
    	&\bbH_{\rm os} = \bbH_{\rm os, 1} + \frac{\mu_{\tau, \rm os}}{\tr\{\bbV_{*}^{-1}\bbD_{\rm p}\}}\bbH_{\rm os, 2} + \frac{\mu_{\tau, \rm os} \mu_{\rho, \tau} + \mu_{\rho}}{\tr\{\bbV_{*}^{-1}\bbD_{\gamma}\}}\bbH_{\rm os, 3},
    \end{aligned}
    \]
    where $\bbH_{\rm os, 1} = (\bbA_{\rm H,os,1} + \bbA_{\rm H,os,1}^{\T}) / 2$,
    \[
    \begin{aligned}
    	&\bbA_{\rm H,os,1} = 
    	\begin{pmatrix}
    		\bbSig_{\gamma}^{\frac{1}{2}}\left(\bbI - (\rho_{U} - \betax)\bbQ(\betax; \taup, \rho_{U})\right)\\
    		-(\bbSig_{\Gamma} - \rho_{U}^{2}\bbSig_{\gamma})^{\frac{1}{2}}\bbQ(\betax; \taup, \rho_{U})^{\T}\\
    		-\sqrt{\taup}\bbP^{\T}\bbQ(\betax; \taup, \rho_U)^{\T}
    	\end{pmatrix}
    	\bbV_{*}^{-1}
    	\begin{pmatrix}
    		(\rho_{U} - \betax)\bbSig_{\gamma}^{\frac{1}{2}}\\
    		(\bbSig_{\Gamma} - \rho_{U}^{2}\bbSig_{\gamma})^{\frac{1}{2}}\\
    		\sqrt{\taup}\bbP^{\T}
    	\end{pmatrix}^{\T},\\
    \end{aligned}
    \]
    \[
    \bbH_{\rm os, 2} = 
    \begin{pmatrix}
    	(\rho_{U} - \betax)\bbSig_{\gamma}^{\frac{1}{2}}\\
    	(\bbSig_{\Gamma} - \rho_{U}^{2}\bbSig_{\gamma})^{\frac{1}{2}}\\
    	\sqrt{\taup}\bbP^{\T}
    \end{pmatrix}
    \bbV_{*}^{-1}
    \begin{pmatrix}
    	(\rho_{U} - \betax)\bbSig_{\gamma}^{\frac{1}{2}}\\
    	(\bbSig_{\Gamma} - \rho_{U}^{2}\bbSig_{\gamma})^{\frac{1}{2}}\\
    	\sqrt{\taup}\bbP^{\T}
    \end{pmatrix}^{\T},
    \]
    $\bbH_{\rm os, 3} = (\bbA_{\rm H,os,3} + \bbA_{\rm H,os,3}^{\T}) / 2$,
    \[
    \bbA_{\rm H,os,3} = 
    \begin{pmatrix}
    	\bbSig_{\gamma}^{\frac{1}{2}}\\
    	\mathrm{0}\\
    	\mathrm{0}
    \end{pmatrix}
    \bbV_{*}^{-1}
    \begin{pmatrix}
    	(\rho_{U} - \betax)\bbSig_{\gamma}^{\frac{1}{2}}\\
    	(\bbSig_{\Gamma} - \rho_{U}^{2}\bbSig_{\gamma})^{\frac{1}{2}}\\
    	\sqrt{\taup}\bbP^{\T}
    \end{pmatrix}^{\T},
    \]
    and    
    \[
    \begin{aligned}
    	&\mu_{\beta} = - \tr\{\bbV_{*}^{-1}\bbD_{\gamma}\}^{-1}\bgamma^{\T}\bbV_{*}^{-1}\bgamma,\\
    	&\mu_{\tau, \rm os} = (\rho_{U} - \betax) \tr\{\bbV_{*}^{-1}\bbD_{\rm p}\bbD_{\rm F, os}^{-1}\bbD_{\gamma}\},\\
    	&\mu_{\rho, \tau} = 2\betax\tr\{\bbV_{*}^{-1}\bbD_{\rm p}\}^{-1}\tr\{\bbV_{*}^{-1}\bbD_{\gamma}\},\\ 
    	&\mu_{\rho} = - \tr\{\bbV_{*}^{-1}\bbD_{\gamma}\} - 2(\rho_{U} - \betax)\betax\tr\{\bbV_{*}^{-1}\bbD_{\gamma}\bbD_{\rm F, os}^{-1}\bbD_{\gamma}\}. 
    \end{aligned}
    \]
    {
For $j=1,\ldots,2q+d$, define
$\boldsymbol v_{{\rm os},j}=(\ell_{1{\rm os},j},\ell_{2{\rm os},j})^{\T}$,
where $\ell_{k{\rm os},j}$ is the $j$th component of $\bl_{k{\rm os}}$, and let
\[
\mathbf K_{\rm os}=\bbL_{\rm os}+2\tr\{\bbH_{\rm os}^{2}\}
\begin{pmatrix}1&0\\0&0\end{pmatrix}.
\]
}
    \begin{condition}\label{cond: linearity os}
    	{
(i) $\displaystyle\max_{2q<j\leq2q+d}\boldsymbol v_{{\rm os},j}^{\T}\mathbf K_{\rm os}^{+}\boldsymbol v_{{\rm os},j}\to0$;
} (ii) $\|\bbH_{\rm os}\| / \sqrt{\tr\{\bbH_{\rm os}^{2}\}} \to 0$, $\sqrt{\tr\{\bbH_{\rm os}^{2}\}} \asymp \sum_{j}b_{j}(\bbH_{\rm os})$ and $\scB(\bbH_{\rm os}) \to 0$; {
(iii) for the weight sequences under consideration, $\bw^{\T}\bbPsi_{\rm os}\bw>0$ eventually, where $\bbPsi_{\rm os}$ is defined in Section \ref{app: proof of sec os}.
}
    \end{condition}
    Conditions \ref{cond: matrices origin}, \ref{cond: selection}, \ref{cond: angle}, \ref{cond: matrices pleiotropy}, \ref{cond: bound pleiotropy}, \ref{cond: matrices os}, and \ref{cond: linearity os} are needed for the proof of Theorem \ref{thm: AN os}. The asymptotic expansions of $\hbeta_{1, \rm os}$ and $\hbeta_{2}$ can be established under Conditions \ref{cond: matrices origin}, \ref{cond: selection}, \ref{cond: angle}, \ref{cond: matrices pleiotropy}, \ref{cond: bound pleiotropy}, and \ref{cond: matrices os}. Condition \ref{cond: linearity os} {
controls the normalized non-Gaussian linear contributions and the quadratic contributions in the asymptotic expansion of $\enbeta_{\rm DEEM, os}$
}, which is needed in establishing the asymptotic normality. 
    \section{Proofs}
    \subsection{Proofs of the Results in Section \ref{subsec: set up}}\label{app: proof of sec setup}
    \paragraph{Derivation of existing estimators}
    Suppose $(Y_{1}, X_{1}, \bG_{1}), \dots, (Y_{n}, X_{n}, \bG_{n})$ are independent and identically distributed copies of $(Y, X, \bG)$. 
    Let $\bbY = (Y_{1} - \bar{Y},\dots, Y_{n} - \bar{Y})^{\T}$, $\bbX = (X_{1} - \bar{X},\dots, X_{n} - \bar{X})^{\T}$ and $\bbG = (\bG_{1} - \bar{\bG},\dots, \bG_{n} - \bar{\bG})^{\T}$ be the centered outcome vector, exposure vector, and the genotype matrix, respectively, where $\bar{Y} = \sum_{i=1}^{n}Y_{i}/n$, $\bar{X} = \sum_{i=1}^{n}X_{i}/n$, and $\bar{\bG} = \sum_{i=1}^{n}\bG_{i}/n$. 
    For $j = 1,\dots,d$, let $\widehat{\gamma}_{j} = \bbG_{\cdot j}^{\T}\bbX / (\bbG_{\cdot j}^{\T}\bbG_{\cdot j})$ and $\widehat{\Gamma}_{j} = \bbG_{\cdot j}^{\T}\bbY / (\bbG_{\cdot j}^{\T}\bbG_{\cdot j})$ be the marginal regression coefficients of $\bbX$ and $\bbY$ on $\bbG_{\cdot j}$, respectively, where $\bbG_{\cdot j}$ is the $j$th column of $\bbG$.
    Under the fixed design (or equivalently conditional on $\bbG$), it can be verified that the covariance matrices of $\hgamma$, $\hGamma$ and $\hGamma - \betax\hgamma$ are all proportional to $\bbPhi\bbG^{\T}\bbG \bbPhi$ where $\bbPhi = \diag\{\bbG_{\cdot 1}^{\T}\bbG_{\cdot 1}, \dots, \bbG_{\cdot d}^{\T}\bbG_{\cdot d}\}^{-1}$. Suppose $\bbG^{\T}\bbG$ is invertible. Then one can use $\bbV = \bbPhi\bbG^{\T}\bbG \bbPhi$ and the resulting solution of \eqref{eq: plug-in EE} is 
    \[\frac{\bbX^{\T}\bbG(\bbG^{\T}\bbG)^{-1}\bbG^{\T}\bbY}{\bbX^{\T}\bbG(\bbG^{\T}\bbG)^{-1}\bbG^{\T}\bbX},\]
    which is the widely used TSLS estimator \citep{anderson2005origins}. When $\hgamma$ and $\hGamma$ come from two independent samples, the optimal two-sample IV estimator proposed by \cite{zhao2019two} can be derived similarly. 
    
    The components of $\hgamma$ and $\hGamma$ are approximately independent \citep{zhao2020statistical} if the SNPs in $\bG$ are mutually independent. Moreover, if $\betax$ is small, we have $\bbSig \approx \bbSig_{\Gamma} \approx \diag \{\sigma_{\Gamma 1}^{2},\dots, \sigma_{\Gamma d}^{2}\}$ where $\sigma_{\Gamma j}^{2}$ is the variance of the $\widehat{\Gamma}_{j}$ for $j = 1,\dots,d$. In this case, the IVW estimator can be obtained if we take $\bbV = \diag\{\widehat{\sigma}_{\Gamma 1}^{2},\dots, \widehat{\sigma}_{\Gamma d}^{2}\}$ in Equation \eqref{eq: plug-in EE}, where $\widehat{\sigma}_{\Gamma j}^{2}$ is an estimator for $\sigma_{\Gamma j}^{2}$ for $j = 1,\dots, d$. Similarly, the WLR estimator and the PCA-IVW estimator can be viewed as solutions of Equation \eqref{eq: plug-in EE} when some specific working covariance matrices are adopted.
    
    % \paragraph{Randomness of $\bbV$}
    % \begin{proposition}
    	% 	Suppose $\|\bbV - \bbV_{*}\| = o_{P}\left(\|\bbV_{*}\|\right)$ for some deterministic matrix $\bbV_{*}$. If $\bbV_{*}$ has bounded condition number, then
    	% 	\begin{equation*}
    		% 		\|\bbV^{-1}\hgamma - \bbV_{*}^{-1}\bgamma\| = \|\bbV_{*}^{-1}\hgamma - \bbV_{*}^{-1}\bgamma\| + o_{P}\left(\|\bbV_{*}^{-1}\hgamma - \bbV_{*}^{-1}\bgamma\|\right).
    		% 	\end{equation*}
    	% \end{proposition}
    %\begin{proof}
    %	Because $\|\bbV - \bbV_{*}\| = o_{P}\left(\|\bbV_{*}\|\right)$ and $\bbV_{*}$ has bounded condition number, we have
    %	\[
    %	\begin{aligned}
    	%		\|\hgamma^{\T}(\bbV^{-1}  - \bbV_{*}^{-1})\| &\leq \|\hgamma\|\|\bbV^{-1} - \bbV_{*}^{-1}\|\\
    	%		& = o_{P}(\|\hgamma - \bgamma\|\|\bbV_{*}^{-1}\|)\\
    	%		& = o_{p}(\|\|\hgamma^{\T}\bbV_{*}^{-1} - \bgamma^{\T}\bbV_{*}^{-1}\|),
    	%	\end{aligned}
    %	\]
    %	where the last equality follows because the condition number of $\bbV_{*}$ is bounded.
    %	Thus
    %	\[
    %	\begin{aligned}
    	%		\|\hgamma^{\T}\bbV^{-1} - \bgamma^{\T}\bbV_{*}^{-1}\|
    	%		& = \|\hgamma^{\T}\bbV_{*}^{-1} - \bgamma^{\T}\bbV_{*}^{-1} + \hgamma^{\T}(\bbV^{-1}  - \bbV_{*}^{-1})\| \\
    	%		& = \|\hgamma^{\T}\bbV_{*}^{-1} - \bgamma^{\T}\bbV_{*}^{-1}\| + o_{p}( \|\hgamma^{\T}\bbV_{*}^{-1} - \bgamma^{\T}\bbV_{*}^{-1}\|).
    	%	\end{aligned}
    %	\]
    %\end{proof}
    \subsection{Proofs of the Results in Section \ref{app: weak IV}}\label{app: proof of sec weak IV}
    \paragraph{Proof of the asymptotic bias of $\hbeta_{\rm PlugIn}$}
    \begin{proof}
    	To assess the bias of $\hbeta_{\rm PlugIn}$, suppose
    	$\hgamma^{\T}\bbV^{-1}\hgamma$ is close to its expectation in the sense that $\hgamma^{\T}\bbV^{-1}\hgamma / E(\hgamma^{\T}\bbV^{-1}\hgamma) \to 1$ in probability, which holds under general regularity conditions. Then, we have
    	\[
    	\begin{aligned}
    		\hbeta_{\rm PlugIn} - \betax	
    		& = \left[\frac{\hgamma^{\T}\bbV^{-1}(\hGamma - \betax \hgamma) - E\{\hgamma^{\T}\bbV^{-1}(\hGamma - \betax \hgamma)\}}{E(\hgamma^{\T}\bbV^{-1}\hgamma)} + \frac{E\{\hgamma^{\T}\bbV^{-1}(\hGamma - \betax \hgamma)\}}{E(\hgamma^{\T}\bbV^{-1}\hgamma)}\right]\times\\
    		&\quad (1 + o_{P}(1))\\
    		& \equalscolon (r_{1} + r_{2})(1 + o_{P}(1)) \\
            & = r_{1} + r_{2} + o_{P}(r_{1} + r_{2}).
    	\end{aligned}
    	\]
    	The term $r_{1}$ is mean zero, and the term $r_{2}$ is the bias term.
    	Notice that $E[\hgamma^{\T}\bbV^{-1}\hgamma] = \bgamma^{T}\bbV^{-1}\bgamma + \tr\{\bbV^{-1}\bbSig_{\gamma}\}$. We have
        \[
    	r_{2} = - \frac{\betax\tr\{\bbV^{-1}\bbSig_{\gamma}\}}{\bgamma^{T}\bbV^{-1}\bgamma + \tr\{\bbV^{-1}\bbSig_{\gamma}\}} = - \frac{\betax}{\Theta(\|\bgamma\|^{2} / \tr\{\bbSig_{\gamma}\}) + 1}.
    	\]
    \end{proof}
        
    \subsection{Proofs of the Results in Section \ref{subsec: orDEEM}}\label{app: proof of sec orDEEM}
    \paragraph{A counterexample when $\bbD_{\gamma}$ and $\bbD_{\Gamma}$ consist of the diagonal elements of $\bbSig_{\gamma}$ and $\bbSig_{\Gamma}$}~\\
    Consider a case where $\betax = 1$,
    \[
      \bbSig_{\gamma} = \begin{pmatrix}
      	0.5 & 0.1\\
      	0.1& 0.5
      \end{pmatrix},
      \bbSig_{\Gamma} = \begin{pmatrix}
      	1.5 & 0.9\\
      	0.9& 1.5
      \end{pmatrix},
      \enspace \text{and} \enspace
      \bbV^{-1} = (\bbSig_{\Gamma} + \betax^{2}\bbSig_{\gamma})^{-1}
      = \frac{1}{3}
      \begin{pmatrix}
      	2 & - 1\\
      	- 1& 2
      \end{pmatrix}. 
    \]
    We have $\bbQ(\betax) = - 1/4 \bbI$ if $\bbD_{\gamma}$ and $\bbD_{\Gamma}$ consist of the diagonal elements of $\bbSig_{\gamma}$ and $\bbSig_{\Gamma}$, respectively.  Then the expectation of the EE \eqref{eq: DEEM pop} at $\betax$ is $-0.1$, which indicates that \eqref{eq: DEEM pop} is biased in this case.
    
    \paragraph{Proof of Equation \eqref{eq: inner product representation}}
    \begin{proof}
    	By straightforward calculation, we have
    	\begin{equation}\label{eq: bias representation}
    		\begin{aligned}
    			&E[\{\hgamma - \bbQ(\betax)(\hGamma - \betax \hgamma)\}^{\T}\bbV^{-1}(\hGamma - \betax \hgamma)]\\
    			&= - \tr\left\{\bbQ(\betax)^{\T}\bbV^{-1}(\bbSig_{\Gamma} + \betax^{2}\bbSig_{\gamma})\right\} - \betax \tr\left\{\bbV^{-1}\bbSig_{\gamma}\right\} \\
    			& = -\tr\{\bbQ(\betax)^{\T}\bbV^{-1}(\bbSig_{\Gamma} - \bbD_{\Gamma})\} - \betax^2 \tr\{\bbQ(\betax)^{\T}\bbV^{-1}(\bbSig_{\gamma} - \bbD_{\gamma})\}\\
    			&\quad \ - \tr\{\bbQ(\betax)^{\T}\bbV^{-1}(\bbD_{\Gamma} + \betax^{2}\bbD_{\gamma})\} - \betax \tr\left\{\bbV^{-1}\bbSig_{\gamma}\right\}.
    		\end{aligned}
    	\end{equation}
    	Note that 
    	\begin{equation*}
    		\begin{aligned}
    			&\tr\{\bbQ(\betax)^{\T}\bbV^{-1}(\bbD_{\Gamma} + \betax^{2}\bbD_{\gamma})\}\\
    			&= - \betax \tr\{(\bbD_{\Gamma} + \betax^{2}\bbD_{\gamma})^{-1} \bbD_{\gamma}\bbV^{-1}(\bbD_{\Gamma} + \betax^{2}\bbD_{\gamma})\}\\
    			&= - \betax\tr\{\bbV^{-1}\bbD_{\gamma}\}.
    		\end{aligned}
    	\end{equation*}
    	Combining this with \eqref{eq: bias representation}, we have
    	\begin{equation*}
    		\begin{aligned}
    			&E[\{\hgamma - \bbQ(\betax)(\hGamma - \betax \hgamma)\}^{\T}\bbV^{-1}(\hGamma - \betax \hgamma)]\\
    			& = -\tr\{\bbQ(\betax)^{\T}\bbV^{-1}(\bbSig_{\Gamma} - \bbD_{\Gamma})\} -  \tr\{(\betax^{2}\bbQ(\betax) + \betax \bbI)^{\T}\bbV^{-1}(\bbSig_{\gamma} - \bbD_{\gamma})\}\\
    			& = -\langle \bbQ(\betax), \bbSig_{\Gamma} - \bbD_{\Gamma}\rangle_{\rm \bbV} - \langle \betax^{2}\bbQ(\betax) + \betax \bbI, \bbSig_{\gamma} - \bbD_{\gamma}\rangle_{\rm \bbV}.
    		\end{aligned}
    	\end{equation*}
    \end{proof}
    
    \paragraph{Proof of Proposition \ref{prop: D and unbiasedness}}
    \begin{proof}
    	We first establish the expression of $\bbD_{\gamma}$ and $\bbD_{\Gamma}$.
    	We only prove the result for $\bbD_{\gamma}$. The result for $\bbD_{\Gamma}$ follows similarly.
    	According to the property of projection, we have
    	\[
    	\tr\{\bbD \bbV^{-1}(\bbSig_{\gamma} - \bbD_{\gamma})\} = 0
    	\]
    	for any diagonal matrix $\bbD$. This implies every diagonal element of $\bbV^{-1}(\bbSig_{\gamma} - \bbD_{\gamma})$ is zero. Thus $\bbV^{-1}\bbSig_{\gamma}$ and $\bbV^{-1}\bbD_{\gamma}$ have identical diagonal elements. Then we have $[\bbV^{-1}\bbSig_{\gamma}]_{jj} = [\bbV^{-1}]_{jj}[\bbD_{\gamma}]_{jj}$ for $j = 1,\dots, d$. Hence $[\bbD_{\gamma}]_{jj} = [\bbV^{-1}\bbSig_{\gamma}]_{jj}/[\bbV^{-1}]_{jj}$.
    	
    	Next, we prove the unbiasedness result.
    	Note that $\cD$ is a linear space.
    	According to the property of a projection onto a linear space, we have 
    	$\tr\{\bbD \bbV^{-1}(\bbSig_{\gamma} - \bbD_{\gamma})\} = 0$ and $\tr\{\bbD \bbV^{-1}(\bbSig_{\Gamma} - \bbD_{\Gamma})\} = 0$ for any $\bbD\in\cD$. This implies
    	\begin{equation*}
    		\begin{aligned}
    			& E[\{\hgamma - \bbQ(\betax)(\hGamma - \betax \hgamma)\}^{\T} \bbV^{-1}(\hGamma - \betax \hgamma)] \\
    			& = 
    			-\betax\tr\{\bbV^{-1}\bbSig_{\gamma}\} - \tr\{\bbQ(\betax)^{\T}\bbV^{-1}(\bbSig_{\Gamma} + \betax^{2} \bbSig_{\gamma})\} \\
    			& = -\betax\tr\{\bbV^{-1}(\bbSig_{\gamma} - \bbD_{\gamma})\} - \betax\tr\{\bbV^{-1}\bbD_{\gamma}\} \\
    			&\quad - \tr\{\bbQ(\betax)^{\T}\bbV^{-1}(\bbSig_{\Gamma} + \betax^{2} \bbSig_{\gamma} - \bbD_{\Gamma} - \betax^{2}\bbD_{\gamma})\}\\
    			&\quad - \tr\{\bbQ(\betax)^{\T}\bbV^{-1}(\bbD_{\Gamma} + \betax^{2}\bbD_{\gamma})\}\\
    			& =  - \betax\tr\{\bbV^{-1}\bbD_{\gamma}\} - \tr\{\bbQ(\betax)^{\T}\bbV^{-1}(\bbD_{\Gamma} + \betax^{2}\bbD_{\gamma})\},
    		\end{aligned}
    	\end{equation*}
    	where the last equality holds because $\bbQ(\betax)$ is diagonal. Thus
    	\begin{equation*}
    		\begin{aligned}
    			& E\{\{\hgamma - \bbQ(\betax)(\hGamma - \betax \hgamma)\}^{\T} \bbV^{-1}(\hGamma - \betax \hgamma)\} \\
    			& =  - \betax\tr\{\bbV^{-1}\bbD_{\gamma}\} - \tr\{\bbQ(\betax)^{\T}\bbV^{-1}(\bbD_{\Gamma} + \betax^{2}\bbD_{\gamma})\} \\
    			& = - \betax\tr\{\bbV^{-1}\bbD_{\gamma}\} + \betax\tr\{(\bbD_{\Gamma} + \betax^{2}\bbD_{\gamma})^{-1}\bbD_{\gamma}\bbV^{-1}(\bbD_{\Gamma} + \betax^{2}\bbD_{\gamma})\} \\
    			& = - \betax\tr\{\bbV^{-1}\bbD_{\gamma}\} + \betax\tr\{\bbV^{-1}\bbD_{\gamma}\}\\
    			& = 0.
    		\end{aligned}
    	\end{equation*}
    \end{proof}
     The unbiasedness result in Proposition \ref{prop: D and unbiasedness} can be extended to allow $\cD$ to be other sets of structured matrices such as commuting symmetric matrix algebras. Next, we provide the extension of Proposition \ref{prop: D and unbiasedness} for the sake of completeness. In the following, we assume $\cD$ is a linear space of symmetric matrices. Then Proposition \ref{prop: D and unbiasedness} can be extended under the following condition.
    \begin{condition}\label{cond: structured set}
    	(i) For any $\bbA \in \cD$, if $\bbA$ is nonsingular, then $\bbA^{-1}\in \cD$; 
    	(ii) for any  $\bbA_{1}, \bbA_{2} \in \cD$, we have $\bbA_{1}\bbA_{2} \in \cD$.
    \end{condition} 
    The following proposition can be proved similarly to Proposition \ref{prop: D and unbiasedness}.
    \begin{proposition}
    	If $\cD$ satisfies Condition \ref{cond: structured set} and $\bbD_{\Gamma} + \betax^{2}\bbD_{\gamma}$ is nonsingular, then
    	\[
    	E\left[\{\hgamma - \bbQ(\betax)(\hGamma - \betax \hgamma)\}^{\T} \bbV^{-1}(\hGamma - \betax \hgamma)\right] = 0.
    	\] 
    \end{proposition}
    Condition \ref{cond: structured set} is satisfied if $\cD$ is the set of diagonal matrices.
    To be specific, we focus on the case where $\cD$ is the set of diagonal matrices in the main text as it performs fairly well in our simulation study.
    
    \subsection{Proofs of the Results in Section \ref{subsec: AN and combine}}\label{app: proof of sec AN and combine}
    In the proofs of Theorems \ref{thm: AN origin}, \ref{thm: AN pleiotropy} and \ref{thm: AN os}, we let $\bxi$ be a $2q$-dimensional random vector with independent standard normal components, and use the notation $\stackrel{d}{=}$ to represent ``equality in distribution".
    \paragraph{Proof of Theorem \ref{thm: AN origin}}
    \begin{proof}
    	We first prove the existence and uniqueness of the solution $\hbeta_{1}$.
    Condition \ref{cond: matrices origin}(i), (iii) gives $c/n\leq[\bbD_{\gamma}]_{jj},[\bbD_{\Gamma}]_{jj}\leq C/n$ uniformly in $j$, so $\bbD_{\Gamma}+\beta^{2}\bbD_{\gamma}\succeq(c/n)\bbI$ for every $\beta$. By part (iv), the estimated denominator has the same lower bound with $c/2$ in place of $c$, with probability tending to one. Thus, on any fixed bounded $\beta$-interval, the population and estimated inverse norms are $O(n)$ and $O_P(n)$, respectively; $\bbQ(\beta)$ and its first two derivatives are uniformly bounded, with corresponding $O_P(1)$ bounds for $\hQ(\beta)$. In particular, $\|\bbH\|=O(1)$.
    	Let $\hs_{1}(\beta) = \{\hgamma - \hQ(\beta)(\hGamma - \beta \hgamma)\}^{\T} \bbV^{-1}(\hGamma - \beta \hgamma)$ and $s_{1}(\beta) = \{\hgamma - \bbQ(\beta)(\hGamma - \beta \hgamma)\}^{\T} \bbV_{*}^{-1}(\hGamma - \beta \hgamma)$. Then \eqref{eq: orDEEM} becomes $\hs_{1}(\beta) = 0$. According to Condition \ref{cond: matrices origin} and Condition \ref{cond: selection}, we have
    	\begin{equation}\label{eq: est s1}
    		\hs_{1}(\beta) - s_{1}(\beta) = O_{P}( n^{1-\kappa}\|\bgamma\|^{2}).
    	\end{equation}
    	for $\beta$ in a neighborhood of $\betax$.
    	Note that $E\{s_{1}(\betax)\} = 0$. For any sufficiently small $\epsilon > 0$, we have 
    	\begin{equation}\label{eq: order expect s1}
    		E\{s_{1}(\betax - \epsilon)\} \geq c\bgamma^{\T}\bbV_{*}^{-1}\bgamma\ \text{and}\ E\{s_{1}(\betax + \epsilon)\} \leq - c\bgamma^{\T}\bbV_{*}^{-1}\bgamma
    	\end{equation}
    	for some constant $c > 0$. Let 
    	\[
    	\bbSig^{\dag} = 
    	\begin{pmatrix}
    		\bbSig_{\gamma} & -\betax\bbSig_{\gamma}\\
    		-\betax\bbSig_{\gamma} & \bbSig_{\Gamma} + \betax^{2}\bbSig_{\gamma}
    	\end{pmatrix}
    	\]
    	be the joint variance-covariance matrix of $\hgamma$ and $\hGamma - \betax\hgamma$.
    	Let 
    	\[
    	\bbM = \frac{1}{2}
    	\begin{pmatrix}
    		\mathbf{0} & \bbV_{*}^{-1}\\
    		\bbV_{*}^{-1} & - \bbQ(\betax)^{\T}\bbV_{*}^{-1} - \bbV_{*}^{-1}\bbQ(\betax)
    	\end{pmatrix}.
    	\] 
    	Then
    	\[
    	s_{1}(\betax) = 
    	\begin{pmatrix}
    		\hgamma \\
    		\hGamma - \betax\hgamma
    	\end{pmatrix}^{\T}
    	\bbM
    	\begin{pmatrix}
    		\hgamma \\
    		\hGamma - \betax\hgamma
    	\end{pmatrix}.
    	\]
    	By some algebra, it can be verified that 
    	\begin{equation}\label{eq: order variance s1}
    		\begin{aligned}
    			\var\{s_{1}(\betax)\}
    			& = \bgamma^{\T}\bbV_{*}^{-1}\bbSig\bbV_{*}^{-1}\bgamma + 2\tr\{\bbM\bbSig^{\dag}\bbM\bbSig^{\dag}\}\\
    			& = \Theta(\bgamma^{\T}\bbV_{*}^{-1}\bgamma)
    		\end{aligned}
    	\end{equation}
    	according to Condition \ref{cond: matrices origin} (i), (ii), (iii), Condition \ref{cond: selection}, Condition \ref{cond: angle} (iii) and Theorem 3.2d.4 in \cite{mathai1992quadratic}.
    	Then, according to \eqref{eq: est s1}, \eqref{eq: order expect s1}, and \eqref{eq: order variance s1}, we have
    	\[
    	P\left(\hs_{1}(\betax - \epsilon) > 0\right) \to 1 \ \text{and} \ P\left(\hs_{1}(\betax + \epsilon) < 0\right) \to 1.
    	\]
    	This implies $\hs_{1}(\beta) = 0$ has a solution $\hbeta_{1}$ in $[\betax - \epsilon, \betax + \epsilon]$ with probability tending to one. 
    	Next we prove the uniqueness of $\hbeta_{1}$.
    	Similar to \eqref{eq: est s1}, we have
    	\begin{align}\label{eq: est first derivative}
    			\frac{\partial}{\partial \beta}\hs_{1}(\beta)
    			& = - \hgamma^{\T}\bbV^{-1}\hgamma + \hgamma^{\T}\bbV^{-1}\hQ(\beta)(\hGamma - \beta \hgamma) + (\hGamma - \beta \hgamma)^{\T}\bbV^{-1}\hQ(\beta)\hgamma\nn\\
    			&\quad + (\hGamma - \beta \hgamma)^{\T}\bbV^{-1}\hD_{\gamma}(\hD_{\Gamma} + \beta^{2}\hD_{\gamma})^{-1}(\hGamma - \beta\hgamma)\nn\\
    			&\quad - 2(\hGamma - \beta \hgamma)^{\T}\bbV^{-1}\hQ(\beta)^{2}(\hGamma - \beta \hgamma)\nn\\
    			& =  - \hgamma^{\T}\bbV_{*}^{-1}\hgamma + \hgamma^{\T}\bbV_{*}^{-1}\bbQ(\beta)(\hGamma - \beta \hgamma) + (\hGamma - \beta \hgamma)^{\T}\bbV_{*}^{-1}\bbQ(\beta)\hgamma\nn\\
    			&\quad + (\hGamma - \beta \hgamma)^{\T}\bbV_{*}^{-1}\bbD_{\gamma}(\bbD_{\Gamma} + \beta^{2}\bbD_{\gamma})^{-1}(\hGamma - \beta\hgamma)\nn\\
    			&\quad - 2(\hGamma - \beta \hgamma)^{\T}\bbV_{*}^{-1}\bbQ(\beta)^{2}(\hGamma - \beta \hgamma) + O_{P}\left\{(q + n\|\bgamma\|^{2})n^{-\kappa}\right\}\nn\\
    			& = \frac{\partial}{\partial \beta}s_{1}(\beta) + O_{P}( n^{1-\kappa}\|\bgamma\|^{2}).
    	\end{align}
    	for $\beta$ in a neighborhood of $\betax$.
    	Note that
    	\begin{equation}\label{eq: expect first derivative}
    		\begin{aligned}
    			E\left\{\frac{\partial}{\partial \beta}s_{1}(\betax)\right\}
    			& = - \bgamma^{\T}\bbV_{*}^{-1}\bgamma - \tr\{\bbV_{*}^{-1}\bbD_{\gamma}\} - 2\betax\tr\{\bbV_{*}^{-1}\bbQ(\betax)\bbD_{\gamma}\} +  \tr\{\bbV_{*}^{-1}\bbD_{\gamma}\}\\
    			&\quad + 2\betax\tr\{\bbV_{*}^{-1}\bbQ(\betax)\bbD_{\gamma}\} \\
    			& = - \bgamma^{\T}\bbV_{*}^{-1}\bgamma.
    		\end{aligned}
    	\end{equation}
    	This implies $E\left\{\partial s_{1}(\beta)/\partial \beta\right\} \leq -c \bgamma^{\T}\bbV_{*}^{-1}\bgamma$ in a neighborhood of $\betax$ for some constant $c > 0$.
    	Equation \eqref{eq: expect first derivative} implies $E\left\{\partial s_{1}(\betax)/\partial \beta\right\} = -\bgamma^{\T}\bbV_{*}^{-1}\bgamma$.
    	By Condition \ref{cond: matrices origin} (ii), we have $\bgamma^{\T}\bbV_{*}^{-1}\bgamma = \Theta(n\|\bgamma\|^{2})$.
    	It is straightforward to show that $\var\{\partial s_{1}(\betax) / \partial \beta\} = O(n\|\bgamma\|^{2}) = O(\bgamma^{\T}\bbV_{*}^{-1}\bgamma)$ according to Condition \ref{cond: matrices origin} (i), (ii), (iii), Condition \ref{cond: selection} and Theorem 3.2d.4 in \cite{mathai1992quadratic}. Hence 
    	\begin{equation}\label{eq: order first derivative}
    		\frac{\partial}{\partial \beta}s_{1}(\betax) = - \bgamma^{\T}\bbV_{*}^{-1}\bgamma + o_{P}( \bgamma^{\T}\bbV_{*}^{-1}\bgamma).
    	\end{equation} Moreover, we have
    	\begin{align}\label{eq: est second derivative}
    			\frac{\partial^{2}}{\partial \beta^2}\hs_{1}(\beta)
    			& = -\hgamma^{\T}\bbV_{*}^{-1}\bbD_{\gamma}(\bbD_{\Gamma} + \beta^{2}\bbD_{\gamma})^{-1}(\hGamma - \beta \hgamma) + 2 \hgamma^{\T}\bbV_{*}^{-1}\bbQ(\beta)^{2}(\hGamma - \beta\hgamma)\nn\\
    			&\quad - 2\hgamma^{\T}\bbV_{*}^{-1}\bbQ(\beta)\hgamma - (\hGamma - \beta \hgamma)^{\T}\bbV_{*}^{-1}\bbD_{\gamma}(\bbD_{\Gamma} + \beta^{2}\bbD_{\gamma})^{-1}\hgamma\nn\\
    			&\quad + 2(\hGamma - \beta\hgamma)^{\T}\bbV_{*}^{-1}\bbQ(\beta)^{2}\hgamma -\hgamma^{\T}\bbV_{*}^{-1}\bbD_{\gamma}(\bbD_{\Gamma} + \beta^{2}\bbD_{\gamma})^{-1}(\hGamma - \beta\hgamma)\nn\\
    			&\quad - (\hGamma - \beta\hgamma)^{\T}\bbV_{*}^{-1}\bbD_{\gamma}(\bbD_{\Gamma} + \beta^{2}\bbD_{\gamma})^{-1}\hgamma\nn\\
    			&\quad + 2 (\hGamma - \beta\hgamma)^{\T}\bbV_{*}^{-1}\bbQ(\beta)\bbD_{\gamma}(\bbD_{\Gamma} + \beta^{2}\bbD_{\gamma})^{-1}(\hGamma - \beta\hgamma)\nn\\
    			&\quad + 2\hgamma^{\T}\bbV_{*}^{-1}\bbQ(\beta)^{2}(\hGamma - \beta\hgamma) + 2(\hGamma - \beta\hgamma)^{\T}\bbV_{*}^{-1}\bbQ(\beta)^{2}\hgamma\nn\\
    			&\quad + 2 (\hGamma - \beta\hgamma)^{\T}\bbV_{*}^{-1}\bbQ(\beta)\{\bbD_{\gamma}(\bbD_{\Gamma} + \beta^{2}\bbD_{\gamma})^{-1} - 2\bbQ(\beta)^{2}\}(\hGamma - \beta\hgamma)\nn\\
    			&\quad + 2 (\hGamma - \beta\hgamma)^{\T}\bbV_{*}^{-1}\{\bbD_{\gamma}(\bbD_{\Gamma} + \beta^{2}\bbD_{\gamma})^{-1} - 2\bbQ(\beta)^{2}\}\bbQ(\beta)(\hGamma - \beta\hgamma)\nn\\
    			&\quad + O_{P}( n^{1-\kappa}\|\bgamma\|^{2})\nn\\
    			& = O_{P}(q + n\|\bgamma\|^{2})\nn\\
    			& = O_{P}(\bgamma^{\T}\bbV_{*}^{-1}\bgamma)
    	\end{align}
    	for $\beta$ in a neighborhood of $\betax$ according to Condition \ref{cond: matrices origin} and Condition \ref{cond: selection}. Combining \eqref{eq: est first derivative}, \eqref{eq: order first derivative} with \eqref{eq: est second derivative}, we conclude that
    	\[
    	P\left(\sup_{\beta \in [\betax - c, \betax + c]}	\frac{\partial}{\partial \beta}\hs_{1}(\beta) < -c\bgamma^{\T}\bbV_{*}^{-1}\bgamma\right) \to 1
    	\]
    	for some constant $c > 0$. And hence $\hs_{1}(\beta)$ is strictly decreasing in a neighborhood of $\betax$ with probability tending to one. Thus, with probability tending to one, $\hs_{1}(\beta)$ has a unique solution $\hbeta_{1}$ in $[\betax - c, \betax + c]$ which also belongs to $[\betax - \epsilon, \betax + \epsilon]$. Furthermore, $\hbeta_{1}$ is consistent for $\betax$ due to the arbitrariness of $\epsilon$.
    	
    	Next, we move on to prove the asymptotic normality. Because $\hs_{1}(\hbeta_{1}) = 0$, we have
    	\[
    	\begin{aligned}
    		-\hs_{1}(\betax)
    		& = \hs_{1}(\hbeta_{1}) - \hs_{1}(\betax)\\
    		& = \frac{\partial}{\partial \beta}\hs_{1}(\betax)(\hbeta_{1} - \betax) + \frac{1}{2}\frac{\partial^{2}}{\partial \beta^{2}}\hs_{1}(\bar{\beta})(\hbeta_{1} - \betax)^{2}
    	\end{aligned}
    	\]
    	according to the mean value theorem. According to Condition \ref{cond: matrices origin} and Condition \ref{cond: selection},
    	we have
    	\[
    	\begin{aligned}
    		&\hs_{1}(\betax) - s_{1}(\betax) \\
    		& = \{\hgamma - \hQ(\betax)(\hGamma - \betax \hgamma)\}^{\T} \bbV^{-1}(\hGamma - \betax \hgamma) - \{\hgamma - \bbQ(\betax)(\hGamma - \betax \hgamma)\}^{\T} \bbV^{-1}(\hGamma - \betax \hgamma)\\
    		& \quad + \{\hgamma - \bbQ(\betax)(\hGamma - \betax \hgamma)\}^{\T} \bbV^{-1}(\hGamma - \betax \hgamma) - 
    		\{\hgamma - \bbQ(\betax)(\hGamma - \betax \hgamma)\}^{\T} \bbV_{*}^{-1}(\hGamma - \betax \hgamma)\\
    		& \leq \|\hGamma - \betax\hgamma\|^{2} \|\bbV^{-1}\|\|\hQ(\betax) - \bbQ(\betax)\| + 
    		\|\hGamma - \betax\hgamma\|\|\hgamma - \bbQ(\betax)(\hGamma - \betax \hgamma)\|\|\bbV^{-1} - \bbV_{*}^{-1}\|\\
    		& = O_{P}(n^{-\kappa}q) + O_{P}(n^{-\kappa}\sqrt{q}\sqrt{n\|\bgamma\|^{2}}) \\
    		& = o_{P}(s_{1}(\betax)).
    	\end{aligned}
    	\] Combining this with \eqref{eq: est first derivative}, \eqref{eq: order first derivative} and \eqref{eq: est second derivative}, we have
    	\begin{equation}\label{eq: expan hbeta1}
    		\hbeta_{1} - \betax = \mu_{1}^{-1}s_{1}(\betax) + o_{P}(\hbeta_{1} - \betax),
    	\end{equation}
    	where $\mu_{1} = \bgamma^{\T}\bbV_{*}^{-1}\bgamma$. 
    	On the other hand, by Condition \ref{cond: matrices origin} (ii), (iv), Condition \ref{cond: selection}, and Condition \ref{cond: angle}, it can be shown that
    	\begin{equation}\label{eq: expan hbeta2}
    		\hbeta_{2} - \betax = \mu_{2}^{-1}s_{2}(\betax) + o_{P}(\hbeta_{2} - \betax),
    	\end{equation}
    	where $s_{2}(\betax) = \tgamma^{\T}\bbV_{*}^{-1}(\hGamma - \betax\hgamma)$ and $\mu_{2} = \bgamma^{\T}\bbV_{*}^{-1}\tgamma$.
    	In the following, we write $s_{1}(\betax)$ and $s_{2}(\betax)$ as $s_{1}$ and $s_{2}$ for short, respectively. For any given weight vector $\bw = (w_{1}, w_{2})^{\T}$, we have 
    	\begin{equation}\label{eq: expan tbeta}
    		\enbeta_{\rm DEEM} - \betax = w_{1}\hbeta_{1} + w_{2}\hbeta_{2} - \betax = w_{1}\mu_{1}^{-1}s_{1} + w_{2}\mu_{2}^{-1}s_{2} + o_{P}(\enbeta_{\rm DEEM} - \betax)
    	\end{equation}
    	according to \eqref{eq: expan hbeta1} and \eqref{eq: expan hbeta2}. According to \eqref{eq: normal model}, it can be verified that
    	\begin{equation}\label{eq: dist equiv origin}
    		w_{1}\mu_{1}^{-1}s_{1} + w_{2}\mu_{2}^{-1}s_{2} 
    		\stackrel{d}{=} w_{1}\mu_{1}^{-1}
    		\bxi^{\T}\bbH\bxi + \bg^{\T}\bxi,
    	\end{equation}
    	where $\bxi$ is defined at the beginning of Section \ref{app: proof of sec AN and combine} and
    	\[
    	\bg = 
    	\begin{pmatrix}
    		-\betax\bbSig_{\gamma}^{\frac{1}{2}}&
    		\bbSig_{\Gamma}^{\frac{1}{2}}
    	\end{pmatrix}^{\T}\bbV_{*}^{-1}(w_{1}\mu_{1}^{-1}\bgamma + w_{2}\mu_{2}^{-1}\tgamma).
    	\]
    	We have $E(w_{1}\mu_{1}^{-1}
    	\bxi^{\T}\bbH\bxi + \bg^{\T}\bxi) = w_{1}\mu_{1}^{-1}\tr\{\bbH\} = 0$ according to Proposition \ref{prop: D and unbiasedness}.
    	By Theorem 3.2d.4 in \cite{mathai1992quadratic}, we have
    	\[
    	\var( w_{1}\mu_{1}^{-1}
    	\bxi^{\T}\bbH\bxi + \bg^{\T}\bxi) = \bg^{\T}\bg + 2w_{1}^{2}\mu_{1}^{-2}\tr\{\bbH^{2}\} \equalscolon \psi_{w}
    	\]
    	Hence according to equation (2.2) in \cite{knight1985joint}, the characteristic function of
    	\[
    	\psi_{w}^{-\frac{1}{2}}(w_{1}\mu_{1}^{-1}\bxi^{\T}\bbH\bxi + \bg^{\T}\bxi)
    	\]
    	is 
    	\[
    	\begin{aligned}
    		f_{w}(t) 
    		&= \left|\bbI - \frac{2itw_{1}}{\sqrt{\psi_{w}}\mu_{1}} \bbH\right|^{-\frac{1}{2}}
    		\exp\left\{-\frac{t^{2}}{2\psi_{w}}\bg^{\T}\left(\bbI - \frac{2itw_{1}}{\sqrt{\psi_{w}}\mu_{1}} \bbH\right)^{-1}\bg\right\}\\
    		&= \exp\left\{-\frac{t^{2}}{2\psi_{w}}\bg^{\T}\left(\bbI - \frac{2itw_{1}}{\sqrt{\psi_{w}}\mu_{1}} \bbH\right)^{-1}\bg - \frac{1}{2}\sum_{j}\log\left(1 - \frac{2itw_{1}}{\sqrt{\psi_{w}}\mu_{1}}\lambda_{0j}\right) \right\},
    	\end{aligned}
    	\]
    	where $i$ is the imaginary unit and $\lambda_{0j}$ is the $j$th eigenvalue of $\bbH$ for $j = 1,\dots 2q$. We have $\max_j|\lambda_{0j}| \leq C$ for some constant $C > 0$ by Condition \ref{cond: matrices origin}. Thus $\max_{j}|tw_{1}\lambda_{0j}/(\sqrt{\psi_{w}}\mu_{1})| \to 0$ according to Condition \ref{cond: balance origin}. Then we have
    	\[
    	\left\|\bbI - \left(\bbI - \frac{2itw_{1}}{\sqrt{\psi_{w}}\mu_{1}} \bbH\right)^{-1}\right\| \to 0,
    	\]
    	and
    	\[
    	\begin{aligned}
    		\log\left(1 - \frac{2itw_{1}}{\sqrt{\psi_{w}}\mu_{1}}\lambda_{0j}\right) = - \frac{2itw_{1}}{\sqrt{\psi_{w}}\mu_{1}}\lambda_{0j} + \frac{2t^{2}w_{1}^{2}}{\psi_{w}\mu_{1}^{2}}\lambda_{0j}^{2} + o\left(\frac{t^{2}w_{1}^{2}}{\psi_{w}\mu_{1}^{2}}\lambda_{0j}^{2}\right).
    	\end{aligned}
    	\]
    	This implies 
    	\[
    	\begin{aligned}
    		f_{w}(t) 
    		&= (1 + o_{P}(1))\exp\left\{-\frac{t^{2}}{2\psi_{w}}\bg^{\T}\bg + \sum_{j}\frac{itw_{1}}{\sqrt{\psi_{w}}\mu_{1}}\lambda_{0j} -
    		\sum_{j}\frac{t^{2}w_{1}^{2}}{\psi_{w}\mu_{1}^{2}}\lambda_{0j}^{2}\right\}.
    	\end{aligned}
    	\]
    	Note that $\sum_{j}\lambda_{0j} = \tr\{\bbH\} = 0$ and $\sum_{j}\lambda_{0j}^{2} = \tr\{\bbH^{2}\}$. We have
    	\begin{equation}\label{eq: characteristic origin}
    		f_{w}(t) = (1 + o_{P}(1))\exp\left(-\frac{t^{2}}{2}\right).
    	\end{equation}
    	This implies $\psi_{w}^{-\frac{1}{2}}(w_{1}\mu_{1}^{-1}\bxi^{\T}\bbH\bxi + \bg^{\T}\bxi) \to N(0, 1)$ in distribution by Corollary 3 in \cite{bulinski2017conditional}. Combining this with \eqref{eq: expan tbeta} and \eqref{eq: dist equiv origin}, we conclude $\psi_{w}^{-1/2}(\enbeta_{\rm DEEM} -\betax) \to N(0, 1)$ in distribution. The proof is completed by noting that $\psi_{w} = \bw^{\T}\bbPsi\bw$, where 
    	\[
    	\bbPsi =
    	\begin{pmatrix}
    		\mu_{1}^{-2}[\bgamma^{\T}\bbV_{*}^{-1}\bbSig\bbV_{*}^{-1}\bgamma + 2\tr\{\bbM\bbSig^{\dag}\bbM\bbSig^{\dag}\}] & \mu_{1}^{-1}\mu_{2}^{-1}\bgamma^{\T}\bbV_{*}^{-1}\bbSig\bbV_{*}^{-1}\tgamma\\
    		\mu_{1}^{-1}\mu_{2}^{-1}\bgamma^{\T}\bbV_{*}^{-1}\bbSig\bbV_{*}^{-1}\tgamma &
    		\mu_{2}^{-2}\tgamma^{\T}\bbV_{*}^{-1}\bbSig\bbV_{*}^{-1}\tgamma
    	\end{pmatrix}.
    	\]
    \end{proof}
    \paragraph{Consistent estimator for $\bbPsi$}
    In order to consistently estimate $\bbPsi$, we need some estimators $\hSig_{\gamma}$ and $\hSig_{\Gamma}$ for $\bbSig_{\gamma}$ and $\bbSig_{\Gamma}$. The construction of $\hSig_{\gamma}$ and $\hSig_{\Gamma}$ based on summary statistics is discussed in Section \ref{app: est Sig}. We assume these estimators are consistent in the sense that $\|\hSig_{\gamma} - \bbSig_{\gamma}\| / \|\bbSig_{\gamma}\| = o_{P}(1)$ and $\|\hSig_{\Gamma} - \bbSig_{\Gamma}\| / \|\bbSig_{\Gamma}\| = o_{P}(1)$. In contrast to $\bbD_{\gamma}$ and $\bbD_{\Gamma}$, we only require $\hSig_{\gamma}$ and $\hSig_{\Gamma}$ to be consistent without further assumptions on the convergence rate. Let
    \[
       	\hSig = 
    	\hSig_{\Gamma} + \hbeta_{2}^{2}\hSig_{\gamma},\
    	\hSig^{\dag} = 
    	\begin{pmatrix}
    		\hSig_{\gamma} & -\hbeta_{2}\hSig_{\gamma}\\
    		-\hbeta_{2}\hSig_{\gamma} & \hSig
    	\end{pmatrix},
    \]
    \[
    	\widehat{\bbM} = \frac{1}{2}
    	\begin{pmatrix}
    		\mathbf{0} & \bbV^{-1}\\
    		\bbV^{-1} & - \hQ(\hbeta_{2})^{\T}\bbV^{-1} - \bbV^{-1}\hQ(\hbeta_{2}),
    	\end{pmatrix}.
    \] 
    Note that $\hbeta_{2} - \betax = o_{P}(1)$ according to \eqref{eq: expan hbeta2}, Condition \ref{cond: matrices origin}, Condition \ref{cond: selection} and Condition \ref{cond: angle} (i). Thus $\|\hSig - \bbSig\| / \|\bbSig\| = o_{P}(1)$, $\|\widehat{\bbSig}^{\dag} - \bbSig^{\dag}\| / \|\bbSig^{\dag}\| = o_{P}(1)$, and $\|\widehat{\bbM} - \bbM\| / \|\bbM\| = o_{P}(1)$.
    $\mu_{1}$ and $\mu_{2}$ can be consistently estimated by $\hat{\mu}_{1} = \hgamma^{\T}\bbV^{-1}\hgamma - \tr\{\bbV^{-1}\hSig_{\gamma}\}$ and $\hat{\mu}_{2} = \hgamma^{\T}\bbV^{-1}\tgamma$, respectively.
    Let $\psi_{ij}$ be the $ij$th element of $\bbPsi$ for $i = 1, 2$ and $j = 1, 2$. Define
    \[
      \widehat{\bbPsi} = 
      \begin{pmatrix}
          \hat{\psi}_{11} & \hat{\psi}_{12}\\
          \hat{\psi}_{21} & \hat{\psi}_{22}
      \end{pmatrix},
    \]
    where
    \[
    \begin{aligned}
    	\hat{\psi}_{11} & =\hat{\mu}_{1}^{-2}[\hgamma^{\T}\bbV^{-1}\hSig\bbV^{-1}\hgamma - \tr\{\bbV^{-1}\hSig\bbV^{-1}\hSig_{\gamma}\} + 2\tr\{\widehat{\bbM}\hSig^{\dag}\widehat{\bbM}\hSig^{\dag}\}],\\
    	\hat{\psi}_{12} & = \hat{\psi}_{21} = \hat{\mu}_{1}^{-1}\hat{\mu}_{2}^{-1}\hgamma^{\T}\bbV^{-1}\hSig\bbV^{-1}\tgamma,\\
    	\hat{\psi}_{22} & = \hat{\mu}_{2}^{-2}\tgamma^{\T}\bbV^{-1}\hSig\bbV^{-1}\tgamma.
    \end{aligned}
    \]
    Suppose the condition number $\lambda_{\rm max}(\bbPsi) / \lambda_{\rm min}(\bbPsi)$ of $\bbPsi$ is bounded.
    Then, we have $(\hat{\psi}_{ij} - \psi_{ij}) / \|\bbPsi\| = o_{P}(1)$ for $i = 1, 2$ and $j = 1, 2$ according to the above convergence results and Conditions \ref{cond: matrices origin}, \ref{cond: selection}, and \ref{cond: angle}. Thus, $\|\hPsi - \bbPsi\|/\|\bbPsi\| \leq \|\hPsi - \bbPsi\|_{\rm F}/\|\bbPsi\| = o_{P}(1)$, where $\|\hPsi - \bbPsi\|_{\rm F} = \sqrt{\sum_{i}\sum_{j}(\hat{\psi}_{ij} - \psi_{ij})^{2}}$ is the Frobenius norm of $\hPsi - \bbPsi$. This establishes the consistency of $\hPsi$. $\|\hPsi - \bbPsi\|/\|\bbPsi\| = o_{P}(1)$ implies $\|(\boldsymbol{1}^{\T}\bbPsi^{-1}\boldsymbol{1})^{-1}\boldsymbol{1}^{\T}\bbPsi^{-1} - (\boldsymbol{1}^{\T}\hPsi^{-1}\boldsymbol{1})^{-1}\boldsymbol{1}^{\T}\hPsi^{-1}\| = o_{P}(1)$, which establishes the consistency of the estimated optimal weight. 
    \subsection{Proofs of the Results in Section \ref{subsec: DEEM p}}\label{app: proof of sec p}
    Next, we consider the asymptotic result when SNPs have direct effects on $Y$. See the following diagram for an illustration of the scenario.
    \begin{figure}[h]
    	\centering
    	\begin{tikzpicture}[scale = 0.6]
    		\node [circle, draw=black, fill=white, inner sep=3pt, minimum size=0.5cm] (x) at (0,0) {\large $X$};
    		\node [circle,draw=black,fill=white,inner sep=3pt,minimum size=0.5cm] (z) at (-4, 0) { $G$};
    		\node [circle,draw=black,fill=white,inner sep=3pt,minimum size=0.5cm] (y) at (4,0) {\large $Y$};
    		\node [obs, minimum size=0.7cm] (u) at (2,2.5) {\large $U$};
    		\path [draw,->] (z) edge (x);
    		\path [draw,->,dashed] (z) edge node[ anchor=center, pos=0.5,font=\bfseries]{$\times$} (u);
    		\path [draw,->] (z) edge [out=-30, in=-150] (y);
    		\path [draw,->] (x) edge node[ anchor=center, above, pos=0.5,font=\bfseries] {$\betax$} (y);
    		\path [draw,->] (u) edge (x);
    		\path [draw,->] (u) edge (y);
    	\end{tikzpicture}
    	\caption{Causal diagram for MR analysis with direct effects.}\label{fig: direct effect}
    \end{figure}

    \paragraph{Proof of Theorem \ref{thm: AN pleiotropy}}
    \begin{proof}
    	The existence and uniqueness of the solution $\hbeta_{1, \rm p}$ can be established similarly to that of $\hbeta_{1}$ in the proof of Theorem \ref{thm: AN origin}. We prove the asymptotic normality next. 
{
Write $a_n=\sqrt{\|\bbSig\|}/\|\bgamma\|$ and $t_{n,\rm p}=\|\bbSig\|/(\|\bbSig_{\rm p}\|\sqrt q)$. The remainders below use these deterministic scales (conditional on the supplemental sample).
}
Similar to \eqref{eq: expan hbeta2}, we have
    	\begin{equation}\label{eq: expan hbeta2 p}
    		\hbeta_{2} - \betax = \mu_{2}^{-1}s_{2} + {o_P(a_n)}
    	\end{equation}
    	according to Condition \ref{cond: matrices origin}, Condition \ref{cond: selection}, and Condition \ref{cond: angle}. By calculating the variance of $s_{2}$, Equation \eqref{eq: expan hbeta2 p}, Condition \ref{cond: matrices origin} (i), (ii), Condition \ref{cond: selection}, and Condition \ref{cond: bound pleiotropy} (ii), it can be shown that $\hbeta_{2} - \betax = O_{P}\left(\sqrt{n\|\bbSig\|/\bgamma^{\T}\bbV_{*}^{-1}\bgamma}\right) = o_{P}(1)$. Then
    	\begin{align}\label{eq: approx htau}
    			\htau_{\rm p} - \taup
    			&= \tr\{\bbV^{-1}\hD_{\rm p}\}^{-1}\left[(\hGamma - \hbeta_{2} \hgamma)^{\T}\bbV^{-1}(\hGamma - \hbeta_{2} \hgamma) - \tr\{\bbV^{-1}(\hD_{\Gamma} + \hbeta_{2}^{2}\hD_{\gamma} + \taup\hD_{\rm p})\}\right]\nn\\
    			&= \tr\{\bbV_{*}^{-1}\bbD_{\rm p}\}^{-1}\left[(\hGamma - \hbeta_{2} \hgamma)^{\T}\bbV_{*}^{-1}(\hGamma - \hbeta_{2} \hgamma) - \tr\{\bbV_{*}^{-1}(\bbD_{\Gamma} + \hbeta_{2}^{2}\bbD_{\gamma} + \taup\bbD_{\rm p})\}\right]\nn\\
    			&\quad + O_{P}\left( \frac{n^{1-\kappa}\tr\{\bbSig\}}{\tr\{\bbV_{*}^{-1}\bbD_{\rm p}\}}\right)\nn\\
    			&= \tr\{\bbV_{*}^{-1}\bbD_{\rm p}\}^{-1}\left[(\hGamma - \betax \hgamma)^{\T}\bbV_{*}^{-1}(\hGamma - \betax \hgamma) - \tr\{\bbV_{*}^{-1}(\bbD_{\Gamma} + \betax^{2}\bbD_{\gamma} + \taup\bbD_{\rm p})\}\right]\nn\\
    			&\quad + \tr\{\bbV_{*}^{-1}\bbD_{\rm p}\}^{-1}\left[-2\hgamma^{\T}\bbV_{*}^{-1}(\hGamma - \betax \hgamma) - 2\betax\tr\{\bbV_{*}^{-1}\bbD_{\gamma}\}\right](\hbeta_{2} - \betax)\nn\\
    			&\quad + \tr\{\bbV_{*}^{-1}\bbD_{\rm p}\}^{-1}\left[\hgamma^{\T}\bbV_{*}^{-1}\hgamma - \tr\{\bbV_{*}^{-1}\bbD_{\gamma}\}\right](\hbeta_{2} - \betax)^{2}\nn\\
    			&\quad +  O_{P}\left( \frac{n^{1-\kappa}\tr\{\bbSig\}}{\tr\{\bbV_{*}^{-1}\bbD_{\rm p}\}}\right)\nn\\
    			&= \tr\{\bbV_{*}^{-1}\bbD_{\rm p}\}^{-1}\left[(\hGamma - \betax \hgamma)^{\T}\bbV_{*}^{-1}(\hGamma - \betax \hgamma) - \tr\{\bbV_{*}^{-1}\bbD_{\rm F, p}\}\right]\nn\\
    			&\quad + O_{P}\left(\frac{n\|\bbSig\|}{\tr\{\bbV_{*}^{-1}\bbD_{\rm p}\}}\right) +   O_{P}\left( \frac{n^{1-\kappa}q\|\bbSig\|}{\tr\{\bbV_{*}^{-1}\bbD_{\rm p}\}}\right),
    	\end{align}
    	by Condition \ref{cond: matrices origin}, Condition \ref{cond: matrices pleiotropy} (i), (iv), and Lemma \ref{lem: tr inequality}, $\bbD_{\rm F, p}$ is defined before Condition \ref{cond: linearity pleiotropy}. According to \eqref{eq: normal model} and some algebra, we have
    	\[
    	\begin{aligned}
    		(\hGamma - \betax \hgamma)^{\T}\bbV_{*}^{-1}(\hGamma - \betax \hgamma) 
    		& \stackrel{d}{=}\begin{pmatrix}
    			\bxi\\
    			\taup^{-\frac{1}{2}}\bbeta_{G}
    		\end{pmatrix}^{\T}
    		\bbH_{\rm p, 2}
    		\begin{pmatrix}
    			\bxi\\
    			\taup^{-\frac{1}{2}}\bbeta_{G}
    		\end{pmatrix}\\
    		& = \bxi^{\T}\bbH_{\rm p, 2}^{(11)} \bxi + 2\taup^{-\frac{1}{2}} \bxi^{\T}\bbH_{\rm p, 2}^{(12)}\bbeta_{G} 
    		+ \taup^{-1} \bbeta_{G}^{\T}\bbH_{\rm p, 2}^{(22)}\bbeta_{G},
    	\end{aligned}
    	\]
    	where 
    	\[
    	\begin{aligned}
    		\bbH_{\rm p, 2}^{(11)} 
    		& = 
    		\begin{pmatrix}
    			- \betax \bbSig_{\gamma}^{\frac{1}{2}}\\
    			\bbSig_{\Gamma}^{\frac{1}{2}}
    		\end{pmatrix}\bbV_{*}^{-1}
    		\begin{pmatrix}
    			- \betax \bbSig_{\gamma}^{\frac{1}{2}}\\
    			\bbSig_{\Gamma}^{\frac{1}{2}}
    		\end{pmatrix}^{\T},\\
    		\bbH_{\rm p, 2}^{(12)} 
    		& = 
    		\sqrt{\taup}\begin{pmatrix}
    			- \betax \bbSig_{\gamma}^{\frac{1}{2}}\\
    			\bbSig_{\Gamma}^{\frac{1}{2}}
    		\end{pmatrix}\bbV_{*}^{-1}\bbP,\\
    		\bbH_{\rm p, 2}^{(22)} 
    		& = \taup\bbP^{\T}\bbV_{*}^{-1}\bbP.
    	\end{aligned}
    	\]
    	Thus 
    	\[
    	\begin{aligned}
    		&\var\{(\hGamma - \betax \hgamma)^{\T}\bbV_{*}^{-1}(\hGamma - \betax \hgamma) \}\\
    		&= 2\tr\{\bbH_{\rm p, 2}^{(11)2}\} + 4 \tr\{\bbH_{\rm p, 2}^{(12)}\bbH_{\rm p, 2}^{(12)\T}\} +  2\tr\{(\bbH_{\rm p, 2}^{(22)})^2\} + (\taup^{-2}\tau_{4, \rm p} - 3)h_{\rm p},
    	\end{aligned}
    	\]
    	where $h_{\rm p}$ is the sum of the squares of $\bbH_{\rm p, 2}^{(22)}$'s diagonal components according to the variance formula provided in \cite{goetze2002asymptotic}. According to Condition \ref{cond: matrices pleiotropy} (i), Condition \ref{cond: bound pleiotropy} (i), we have $\tr\{\bbV_{*}^{-1}\bbD_{\rm p}\}^{-2}\var\{(\hGamma - \betax \hgamma)^{\T}\bbV_{*}^{-1}(\hGamma - \betax \hgamma) \} = O\left(\frac{\|\bbSig\|^2}{q\|\bbSig_{\rm p}\|^2}\right) = o(1)$. Combining this with \eqref{eq: approx htau}, we have
    	\begin{equation}\label{eq: expan htaup}
        {\htau_{\rm p}-\taup=s_{\tau,\rm p}+O_P\left(\frac{\|\bbSig\|}{\|\bbSig_{\rm p}\|}(q^{-1}+n^{-\kappa})\right)=s_{\tau,\rm p}+o_P(t_{n,\rm p})}
    	\end{equation}
{
by $q\to\infty$ and $n^{-\kappa}\sqrt q\to0$. In particular, $\htau_{\rm p}-\taup=O_P(t_{n,\rm p})$ and $|\htau_{\rm p}-\taup|\|\bbSig_{\rm p}\|/\|\bbSig\|=O_P(q^{-1/2})=o_P(1)$, where
}
    	\[
    	s_{\tau, \rm p} = \tr\{\bbV_{*}^{-1}\bbD_{\rm p}\}^{-1}\left[(\hGamma - \betax \hgamma)^{\T}\bbV_{*}^{-1}(\hGamma - \betax \hgamma) - \tr\{\bbV_{*}^{-1}\bbD_{\rm F, p}\}\right].
    	\]
    	Let $s_{1, \rm p}(\beta, \tau) = \{\hgamma - \bbQ(\beta; \tau)(\hGamma - \beta\hgamma)\}^{\T}\bbV_{*}^{-1}(\hGamma - \beta\hgamma)$. Write $s_{1, \rm p}(\betax; \taup)$ as $s_{1, \rm p}$ for short.
{
Taylor expansion of $s_{1,\rm p}$ gives the following first-order expansion. In the scale $u=(\tau-\tau_{\rm p})\|\bbSig_{\rm p}\|/\|\bbSig\|$, the nuisance error is $O_P(q^{-1/2})$. The quadratic Taylor terms and fluctuations of the first derivatives, after division by $\mu_1$, are $o_P(a_n)$. Nuisance-matrix substitution has relative order at most $O_P(n^{-\kappa}\sqrt q)$ at this scale.
}
    	\begin{equation}\label{eq: expan hbeta1 p}
    		\hbeta_{1, \rm p} - \betax = \mu_{1}^{-1}s_{1, \rm p} + \mu_{1}^{-1}\mu_{\tau, \rm p}s_{\tau, \rm p} + {o_P(a_n)}
    	\end{equation}
    	according to Condition \ref{cond: matrices origin} and Condition \ref{cond: matrices pleiotropy} (iii), (iv),
    	where 
    	\[
    	\mu_{\tau, \rm p} = E\left\{\frac{\partial}{\partial \tau} s_{1, \rm p}(\betax; \taup)\right\} = - \betax \tr\{\bbV_{*}^{-1}\bbD_{\rm p}\bbD_{\rm F, p}^{-1}\bbD_{\gamma}\},
    	\]
    	$\mu_{1}$ and $\mu_{2}$ are defined after equations \eqref{eq: expan hbeta1} and \eqref{eq: expan hbeta2}, respectively. Equations \eqref{eq: expan htaup} and \eqref{eq: expan hbeta1 p} imply
    	\begin{equation}\label{eq: expan tbeta p}
    		\enbeta_{\rm DEEM, p} - \betax = w_{1}\hbeta_{1, \rm p} + w_{2}\hbeta_{2} - \betax = w_{1}\mu_{1}^{-1}(s_{1, \rm p} + \mu_{\tau, \rm p}s_{\tau, \rm p}) + w_{2}\mu_{2}^{-1}s_{2} + {o_P(\|\bw\|a_n)}
    	\end{equation}
    	for any given $\bw = (w_{1}, w_{2})^{\T}$.
{
Condition \ref{cond: angle} implies $[\bbPsi_{\rm p}]_{22}\asymp a_n^2$. The bounded condition number therefore implies $\lambda_{\min}(\bbPsi_{\rm p})\asymp\lambda_{\max}(\bbPsi_{\rm p})\asymp a_n^2$, so the remainder is negligible after standardization for the allowed, possibly $n$-dependent weights.
}
    	According to \eqref{eq: normal model}, it can be verified that
    	\begin{equation}\label{eq: dist equiv p}
    		\begin{aligned}
    			&w_{1}\mu_{1}^{-1}(s_{1, \rm p} + \mu_{\tau, \rm p}s_{\tau, \rm p}) + w_{2}\mu_{2}^{-1}s_{2}\\
    			&\stackrel{d}{=} w_{1}\mu_{1}^{-1}
    			\begin{pmatrix}
    				\bxi\\
    				\taup^{-\frac{1}{2}}\bbeta_{G}
    			\end{pmatrix}^{\T}
    			\bbH_{\rm p}
    			\begin{pmatrix}
    				\bxi\\
    				\taup^{-\frac{1}{2}}\bbeta_{G}
    			\end{pmatrix} -w_1\mu_1^{-1}\tr(\bbH_{\rm p}) + \bg_{\rm p}^{\T}
    			\begin{pmatrix}
    				\bxi\\
    				\taup^{-\frac{1}{2}}\bbeta_{G}
    			\end{pmatrix},
    		\end{aligned}
    	\end{equation}
    	where $\bg_{\rm p} = w_{1}\mu_{1}^{-1}\bl_{1\rm p} + w_{2}\mu_{2}^{-1}\bl_{2\rm p}$ and $\bbH_{\rm p}$ is defined before Condition \ref{cond: linearity pleiotropy}. Define
    	\[
    	\begin{aligned}
    		\bbSig_{\rm p}^{\dag} & = 
    		\begin{pmatrix}
    			\bbSig_{\gamma} & -\betax\bbSig_{\gamma}\\
    			-\betax\bbSig_{\gamma} & \bbSig_{\Gamma} + \betax^{2}\bbSig_{\gamma} +\taup\bbSig_{\rm p}
    		\end{pmatrix},\\
    		\bbM_{\rm p} & = \frac{1}{2}
    		\begin{pmatrix}
    			\mathbf{0} & \bbV_{*}^{-1}\\
    			\bbV_{*}^{-1} & - \bbQ(\betax; \taup)^{\T}\bbV_{*}^{-1} - \bbV_{*}^{-1}\bbQ(\betax; \taup)
    		\end{pmatrix},\\
    		\bbM_{\tau} & = 
    		\begin{pmatrix}
    			\mathbf{0} & \mathbf{0}\\
    			\mathbf{0} & \bbV_{*}^{-1}
    		\end{pmatrix}.
    	\end{aligned}
    	\]
    	Let
    	\[
    	\begin{aligned}
    		\bbOmega_{\rm p} =
    		\begin{pmatrix}
    			\bgamma^{\T}\bbV_{*}^{-1}\bbSig\bbV_{*}^{-1}\bgamma + 2\tr\{\bbM_{\rm p}\bbSig_{\rm p}^{\dag}\bbM_{\rm p}\bbSig_{\rm p}^{\dag}\} &  \bgamma^{\T}\bbV_{*}^{-1}\bbSig\bbV_{*}^{-1}\tgamma & 2\tr\{\bbV_{*}^{-1}\bbD_{\rm p}\}^{-1}\tr\{\bbM_{\rm p}\bbSig_{\rm p}^{\dag}\bbM_{\tau}\bbSig_{\rm p}^{\dag}\}\\
    			\bgamma^{\T}\bbV_{*}^{-1}\bbSig\bbV_{*}^{-1}\tgamma & 
    			\tgamma^{\T}\bbV_{*}^{-1}\bbSig\bbV_{*}^{-1}\tgamma & 0 \\
    			2\tr\{\bbV_{*}^{-1}\bbD_{\rm p}\}^{-1}\tr\{\bbM_{\rm p}\bbSig_{\rm p}^{\dag}\bbM_{\tau}\bbSig_{\rm p}^{\dag}\} & 0 &2\tr\{\bbV_{*}^{-1}\bbD_{\rm p}\}^{-2}\tr\{\bbM_{\tau}\bbSig_{\rm p}^{\dag}\bbM_{\tau}\bbSig_{\rm p}^{\dag}\}
    		\end{pmatrix}.
    	\end{aligned}
    	\]
    	be the variance-covariance matrix of $(s_{1, \rm p}, s_{2}, s_{\tau, \rm p})^{\T}$ calculated under the assumption that the components of $\bbeta_{G}$ follow a normal distribution. Define
    	\begin{equation}\label{eq: Psi p}
    		\bbPsi_{\rm p} = 
    		\begin{pmatrix}
    			\mu_{1}^{-1} & 0\\
    			0 & \mu_2^{-1}\\
    			\mu_{1}^{-1}\mu_{\tau, \rm p} & 0
    		\end{pmatrix}^{\T}
    		\bbOmega_{\rm p}
    		\begin{pmatrix}
    			\mu_{1}^{-1} & 0\\
    			0 & \mu_2^{-1}\\
    			\mu_{1}^{-1}\mu_{\tau, \rm p} & 0
    		\end{pmatrix}.
    	\end{equation}
    	Let $\bxi^{\dag}$ be a $(2q+d)$-dimensional random vector with independent standard normal components and $f_{w, \rm p}(t)$ and $f_{w, \rm p}^{\prime}(t)$ be the characteristic functions of
    	\[
    	(\bw^{\T}\bbPsi_{\rm p}\bw)^{-\frac{1}{2}}\left\{w_{1}\mu_{1}^{-1}
    	\begin{pmatrix}
    		\bxi\\
    		\taup^{-\frac{1}{2}}\bbeta_{G}
    	\end{pmatrix}^{\T}
    	\bbH_{\rm p}
    	\begin{pmatrix}
    		\bxi\\
    		\taup^{-\frac{1}{2}}\bbeta_{G}
    	\end{pmatrix} -w_1\mu_1^{-1}\tr(\bbH_{\rm p}) + \bg_{\rm p}^{\T}
    	\begin{pmatrix}
    		\bxi\\
    		\taup^{-\frac{1}{2}}\bbeta_{G}
    	\end{pmatrix}
    	\right\}
    	\]
    	and
    	\[
    	(\bw^{\T}\bbPsi_{\rm p}\bw)^{-\frac{1}{2}}\left(w_{1}\mu_{1}^{-1}
    	\bxi^{\dag\T}
    	\bbH_{\rm p}
    	\bxi^{\dag}-w_1\mu_1^{-1}\tr(\bbH_{\rm p}) 
    	+ \bg_{\rm p}^{\T}
    	\bxi^{\dag}
    	\right),
    	\]
    	respectively.
    	{
Conditional on the supplemental sample, write $a=w_1\mu_1^{-1}$,
$s^2=\bw^{\T}\bbPsi_{\rm p}\bw$, and set
$\mathbf A=a\bbH_{\rm p}/s$ and $\boldsymbol u=\bg_{\rm p}/s$.
The Gaussian variance formula gives
\[
2\tr(\mathbf A^2)+\|\boldsymbol u\|^2=1.
\]
Let $\boldsymbol c=(w_1\mu_1^{-1},w_2\mu_2^{-1})^{\T}$.
By definition, $s^2=\boldsymbol c^{\T}\mathbf K_{\rm p}\boldsymbol c$.
Each $\boldsymbol v_{{\rm p},j}$ belongs to the range of $\mathbf K_{\rm p}$,
so the Cauchy--Schwarz inequality implies
\[
(\boldsymbol c^{\T}\boldsymbol v_{{\rm p},j})^2
\leq(\boldsymbol v_{{\rm p},j}^{\T}\mathbf K_{\rm p}^{+}
\boldsymbol v_{{\rm p},j})(\boldsymbol c^{\T}\mathbf K_{\rm p}\boldsymbol c).
\]
Consequently, $\max_{j>2q}|u_j|^2\to0$ by
Condition \ref{cond: linearity pleiotropy}(i).
Condition \ref{cond: linearity pleiotropy}(ii) implies $\|\mathbf A\|\to0$.
For nonzero $\bbH_{\rm p}$, it also gives
\[
\begin{aligned}
\sum_j[\bbH_{\rm p}]_{jj}^2
&\leq\max_j b_j(\bbH_{\rm p})\sum_j|[\bbH_{\rm p}]_{jj}|\\
&\leq\scB(\bbH_{\rm p})\left\{\sum_jb_j(\bbH_{\rm p})\right\}^2
=o\bigl(\tr\{\bbH_{\rm p}^2\}\bigr).
\end{aligned}
\]
Thus $\sum_j A_{jj}^2=o(1)$; the conclusions are immediate when
$\bbH_{\rm p}=\mathbf0$.

Let $\boldsymbol Z=(\bxi^{\T},\taup^{-1/2}\bbeta_G^{\T})^{\T}$.
Its coordinates are independent, centered, have unit variance, and have
uniformly bounded fourth moments by Condition \ref{cond: bound pleiotropy}(i).
The centered diagonal term $\sum_j A_{jj}(Z_j^2-1)$ has conditional second
moment bounded by $C\sum_j A_{jj}^2=o(1)$.
The same bound holds after replacing the final $d$ coordinates by independent
standard normal variables.

For the off-diagonal quadratic polynomial plus its linear term, define
$
I_j=u_j^2+4\sum_{k\ne j}A_{jk}^2$ for $j>2q$.
Then $\max_{j>2q}I_j\to0$ and
$\sum_{j>2q}I_j\leq\|\boldsymbol u\|^2+4\tr(\mathbf A^2)\leq2$.
Replace the final $d$ coordinates successively.
At each replacement the polynomial is affine in that coordinate,
with coefficient $u_j+2\sum_{k\ne j}A_{jk}Z_k$.
Its fourth moment is bounded by $C I_j^2$ uniformly over the intermediate
mixtures of original and Gaussian coordinates.
A second-order Taylor expansion of the characteristic function matches
the first two moments of the replaced coordinate; its absolute remainder
is therefore bounded by $C_t I_j^{3/2}$.
This argument uses bounded absolute third moments and does not require
symmetry of the direct effects.
The total replacement error is bounded by
$
C_t\sum_{j>2q}I_j^{3/2}
\leq C_t\left(\max_{j>2q}I_j\right)^{1/2}\sum_{j>2q}I_j\to0$.
Together with the negligible diagonal terms, this proves
$f_{w,{\rm p}}(t)-f_{w,{\rm p}}^{\prime}(t)\to0$ in probability.
Finally, $\|\mathbf A\|\to0$ and
$2\tr(\mathbf A^2)+\|\boldsymbol u\|^2=1$ allow the Gaussian
characteristic-function calculation used for \eqref{eq: characteristic origin},
with the quadratic term centered by its trace, to give
$f_{w,{\rm p}}^{\prime}(t)\to\exp(-t^2/2)$
}
    	in probability and hence
    	\[
    	f_{w, \rm p}(t) \to \exp\left(-\frac{t^{2}}{2}\right)
    	\]
    	in probability. Combining this with \eqref{eq: expan tbeta p} and \eqref{eq: dist equiv p}, we obtain the asymptotic normality result in Theorem \ref{thm: AN pleiotropy} by Corollary 3 in \cite{bulinski2017conditional}. A plug-in estimator for $\bbPsi_{\rm p}$ can be constructed similarly to $\hPsi$.
    \end{proof}
   
    \paragraph{Consistent estimator for $\bbPsi_{\rm p}$}
    To estimate $\bbPsi_{\rm p}$, we need estimators $\hSig_{\gamma}$, $\hSig_{\Gamma}$, and $\hSig_{\rm p}$ for $\bbSig_{\gamma}$, $\bbSig_{\Gamma}$, and $\bbSig_{\rm p}$. The construction of $\hSig_{\gamma}$, $\hSig_{\Gamma}$, and $\hSig_{\rm p}$ based on summary statistics is discussed in Section \ref{app: est Sig}.
    Suppose $\|\hSig_{\gamma} - \bbSig_{\gamma}\| / \|\bbSig_{\gamma}\| = o_{P}(1)$, $\|\hSig_{\Gamma} - \bbSig_{\Gamma}\| / \|\bbSig_{\Gamma}\| = o_{P}(1)$, and $\|\hSig_{\rm p} - \bbSig_{\rm p}\| / \|\bbSig_{\rm p}\| = o_{P}(1)$. Let $\hD_{\rm F, p} = 
    	\hD_{\Gamma} + \hbeta_{2}^{2}\hD_{\gamma} + \htau_{\rm p}\hD_{\rm p}$,
    \[
       	\hSig_{\rm F, p} = 
    	\hSig_{\Gamma} + \hbeta_{2}^{2}\hSig_{\gamma} + \htau_{\rm p}\hSig_{\rm p},\
    	\hSig_{\rm p}^{\dag} = 
    	\begin{pmatrix}
    		\hSig_{\gamma} & -\hbeta_{2}\hSig_{\gamma}\\
    		-\hbeta_{2}\hSig_{\gamma} & \hSig_{\rm F, p}
    	\end{pmatrix},
    \]
    \[
    	\widehat{\bbM}_{\rm p} = \frac{1}{2}
    	\begin{pmatrix}
    		\mathbf{0} & \bbV^{-1}\\
    		\bbV^{-1} & - \hQ(\hbeta_{2}; \htau_{\rm p})^{\T}\bbV^{-1} - \bbV^{-1}\hQ(\hbeta_{2}; \htau_{\rm p}),
    	\end{pmatrix}.
    \] 
{
The preliminary estimator satisfies $\hbeta_2-\betax=O_P(a_n)=o_P(1)$. More importantly for covariance estimation, \eqref{eq: expan htaup} gives
\[
 |\htau_{\rm p}-\taup|\,\|\bbSig_{\rm p}\|/\|\bbSig\|=O_P(q^{-1/2})=o_P(1).
\]
Since $\bbSig=\bbSig_\Gamma+\betax^2\bbSig_\gamma+\taup\bbSig_{\rm p}$ is a sum of nonnegative definite matrices and $\|\bbSig_\Gamma\|\asymp n^{-1}$, the component covariance errors assumed above are negligible relative to $\|\bbSig\|$.
}
    Combining these with Condition \ref{cond: matrices origin}, Condition \ref{cond: selection}, Condition \ref{cond: angle} (i), and Condition \ref{cond: matrices pleiotropy}, we have $\|\hD_{\rm F, p} - \bbD_{\rm F, p}\| / \|\bbD_{\rm F, p}\| = o_{P}(1)$, $\|\hSig_{\rm F, p} - (\bbSig_{\Gamma} + \betax^{2}\bbSig_{\gamma} +\taup\bbSig_{\rm p})\| / \|\bbSig_{\Gamma} + \betax^{2}\bbSig_{\gamma} +\taup\bbSig_{\rm p}\| = o_{P}(1)$, $\|\widehat{\bbSig}_{\rm p}^{\dag} - \bbSig_{\rm p}^{\dag}\| / \|\bbSig_{\rm p}^{\dag}\| = o_{P}(1)$, and $\|\widehat{\bbM}_{\rm p} - \bbM_{\rm p}\| / \|\bbM_{\rm p}\| = o_{P}(1)$.
    $\mu_{1}$, $\mu_{2}$, and $\mu_{\tau, \rm p}$ can be consistently estimated by $\hat{\mu}_{1} = \hgamma^{\T}\bbV^{-1}\hgamma - \tr\{\bbV^{-1}\hSig_{\gamma}\}$, $\hat{\mu}_{2} = \hgamma^{\T}\bbV^{-1}\tgamma$, and $\hat{\mu}_{\tau, \rm p} = - \hbeta_{2} \tr\{\bbV^{-1}\hD_{\rm p}\hD_{\rm F, p}^{-1}\hD_{\gamma}\}$, respectively.
    Let $\psi_{{\rm p}, ij}$ be the $ij$th element of $\bbPsi_{\rm p}$ for $i = 1, 2$ and $j = 1, 2$.
    Define
    \[
    	\begin{aligned}
    		\widehat{\bbOmega}_{\rm p} =
    		\begin{pmatrix}
    			\hat{v}_{\gamma, {\rm p}} + 2\tr\{\widehat{\bbM}_{\rm p}\widehat{\bbSig}_{\rm p}^{\dag}\widehat{\bbM}_{\rm p}\widehat{\bbSig}_{\rm p}^{\dag}\} &  \hgamma^{\T}\bbV^{-1}\hSig_{\rm F, p}\bbV^{-1}\tgamma & 2\tr\{\bbV^{-1}\hD_{\rm p}\}^{-1}\tr\{\widehat{\bbM}_{\rm p}\widehat{\bbSig}_{\rm p}^{\dag}\widehat{\bbM}_{\tau}\widehat{\bbSig}_{\rm p}^{\dag}\}\\
    			\hgamma^{\T}\bbV^{-1}\hSig_{\rm F, p}\bbV^{-1}\tgamma & 
    			\tgamma^{\T}\bbV^{-1}\hSig_{\rm F, p}\bbV^{-1}\tgamma & 0 \\
    			2\tr\{\bbV^{-1}\hD_{\rm p}\}^{-1}\tr\{\widehat{\bbM}_{\rm p}\widehat{\bbSig}_{\rm p}^{\dag}\widehat{\bbM}_{\tau}\widehat{\bbSig}_{\rm p}^{\dag}\} & 0 &2\tr\{\bbV^{-1}\hD_{\rm p}\}^{-2}\tr\{\widehat{\bbM}_{\tau}\widehat{\bbSig}_{\rm p}^{\dag}\widehat{\bbM}_{\tau}\widehat{\bbSig}_{\rm p}^{\dag}\}
    		\end{pmatrix}.
    	\end{aligned}
    	\]	
    and
    \[
    \begin{aligned}
        \widehat{\bbPsi}_{\rm p} = 
      \begin{pmatrix}
          \hat{\psi}_{{\rm p}, 11} & \hat{\psi}_{{\rm p}, 12}\\
          \hat{\psi}_{{\rm p}, 21} & \hat{\psi}_{{\rm p}, 22}
      \end{pmatrix}
      =
      \begin{pmatrix}
    			\hat{\mu}_{1}^{-1} & 0\\
    			0 & \hat\mu_2^{-1}\\
    			\hat{\mu}_{1}^{-1}\hat{\mu}_{\tau, \rm p} & 0
    		\end{pmatrix}^{\T}
    		\widehat{\bbOmega}_{\rm p}
    		\begin{pmatrix}
    			\hat{\mu}_{1}^{-1} & 0\\
    			0 & \hat\mu_2^{-1}\\
    			\hat{\mu}_{1}^{-1}\hat{\mu}_{\tau, \rm p} & 0
    		\end{pmatrix},
    \end{aligned}
    \]
    where 
    \[
     \widehat{\bbM}_{\tau} = \begin{pmatrix}
    			\mathbf{0} & \mathbf{0}\\
    			\mathbf{0} & \bbV^{-1}
    		\end{pmatrix}
    \]
    and $\hat{v}_{\gamma, {\rm p}} = \hgamma^{\T}\bbV^{-1}\hSig_{\rm F, p}\bbV^{-1}\hgamma - \tr\{\bbV^{-1}\hSig_{\rm F, p}\bbV^{-1}\hSig_{\gamma}\}$.
    Suppose the condition number $\lambda_{\rm max}(\bbPsi_{\rm p}) / \lambda_{\rm min}(\bbPsi_{\rm p})$ of $\bbPsi_{\rm p}$ is bounded.
    Then, we have $(\hat{\psi}_{{\rm p}, ij} - \psi_{{\rm p}, ij}) / \|\bbPsi_{\rm p}\| = o_{P}(1)$ for $i = 1, 2$ and $j = 1, 2$ according to the above convergence results and Conditions \ref{cond: matrices origin}, \ref{cond: selection}, \ref{cond: angle}, and \ref{cond: matrices pleiotropy}. Thus, $\|\hPsi_{\rm p} - \bbPsi_{\rm p}\|/\|\bbPsi_{\rm p}\| \leq \|\hPsi_{\rm p} - \bbPsi_{\rm p}\|_{\rm F}/\|\bbPsi_{\rm p}\| = o_{P}(1)$. This establishes the consistency of $\hPsi_{\rm p}$. The optimal weight $(\boldsymbol{1}^{\T}\bbPsi_{\rm p}^{-1}\boldsymbol{1})^{-1}\boldsymbol{1}^{\T}\bbPsi_{\rm p}^{-1}$ that minimizes the asymptotic variance in this case can be estimated by $(\boldsymbol{1}^{\T}\hPsi_{\rm p}^{-1}\boldsymbol{1})^{-1}\boldsymbol{1}^{\T}\hPsi_{\rm p}^{-1}$, whose consistency follows from the consistency of $\hPsi_{\rm p}$.
    
    \begin{lemma}\label{lem: tr inequality}
    	For matrices $\bbA_{1}$, $\bbA_{2}$, and $\bbA_{3}$ such that $\bbA_{1} - \bbA_{2}$ and $\bbA_{3}$ are non-negative definite, we have $\tr\{\bbA_{1}\bbA_{3}\} \geq \tr\{\bbA_{2}\bbA_{3}\}$.
    \end{lemma}
    \begin{proof}
    	Note that $\bbA_{3}^{\frac{1}{2}}(\bbA_{1} - \bbA_{2})\bbA_{3}^{\frac{1}{2}}$ is non-negative definite under the assumption of the lemma. Hence
    	\[
    	0 \leq \tr\{\bbA_{3}^{\frac{1}{2}}(\bbA_{1} - \bbA_{2})\bbA_{3}^{\frac{1}{2}}\} = \tr\{(\bbA_{1} - \bbA_{2})\bbA_{3}\} = \tr\{\bbA_{1}\bbA_{3}\} - \tr\{\bbA_{2}\bbA_{3}\},
    	\]
    	which completes the proof.
    \end{proof}
    \subsection{Proofs of the Results in Section \ref{subsec: DEEM os}}\label{app: proof of sec os}
{
\paragraph{Motivation for the proportional covariance assumption}
}
{
    Here we adopt the notation in Section \ref{app: proof of sec setup} and analyze conditional on the full genotype design $\{\bG_i\}_{i=1}^n$ and the direct-effect vector $\bbeta_G$.
} Then we have $\hgamma = \bbPhi \bbG^{\T}\bbX$ and $\hGamma = \bbPhi \bbG^{\T}\bbY$. Let $\mathbf{e}_{X} = (\epsilon_{X1},\dots, \epsilon_{Xn})^{\T}$, $\mathbf{e}_{Y} = (\epsilon_{Y1},\dots, \epsilon_{Yn})^{\T}$, $\bU_{g} = (g(\bU_{1}) - E[g(\bU)],\dots, g(\bU_{n}) - E[g(\bU)])^{\T}$, and $\bU_{f} = (f(\bU_{1}) - E[f(\bU)],\dots, f(\bU_{n}) - E[f(\bU)])^{\T}$, where $\epsilon_{Xi}$,  $\epsilon_{Yi}$ and $\bU_{i}$ are the error terms in $X_{i}$, $Y_{i}$ and the vector of unmeasured confounders for the $i$th observation. Under model \eqref{eq: pleiotropy outcome model}, we have $\hgamma - E(\hgamma) = \bbPhi \bbG^{\T}(\bU_{g} + \mathbf{e}_{X})$ and $\hGamma - E(\hGamma) = \bbPhi \bbG^{\T}(\bU_{f} + \bU_{g}\betax + \mathbf{e}_{Y} + \mathbf{e}_{X}\betax)$. Then we have $\bbSig_{\gamma} = \var(\hgamma) = \{\var(g(\bU)) + \var(\epsilon_{X})\}\bbPhi\bbG^{\T}\bbG\bbPhi$ and $\cov(\hgamma,\hGamma) = \{\cov(g(\bU), f(\bU)) + \betax \var(g(\bU)) + \betax\var(\epsilon_{X})\}\bbPhi\bbG^{\T}\bbG\bbPhi$. {
This proves the proportional relation for these conditional sampling covariances.
}

    \paragraph{Proof of the unbiasedness of EE $\{\hgamma - \bbQ(\beta; \taup, \rho_{U})(\hGamma - \beta \hgamma)\}^{\T} \bbV^{-1}(\hGamma - \beta \hgamma) = 0$}
    \begin{proof}
        Recall that $\cD$ is the set of all diagonal matrices. By definition, we have
        $\tr\{\bbD \bbV^{-1}(\bbSig_{\gamma} - \bbD_{\gamma})\} = 0$, $\tr\{\bbD \bbV^{-1}(\bbSig_{\Gamma} - \bbD_{\Gamma})\} = 0$ and $\tr\{\bbD \bbV^{-1}(\bbSig_{\rm p} - \bbD_{\rm p})\} = 0$ for any $\bbD\in\cD$. Thus,
        \begin{equation*}
    		\begin{aligned}
    			& E[\{\hgamma - \bbQ(\betax; \taup, \rho_{U})(\hGamma - \betax \hgamma)\}^{\T} \bbV^{-1}(\hGamma - \betax \hgamma)] \\
    			& = 
    			(\rho_{U} -\betax)\tr\{\bbV^{-1}\bbSig_{\gamma}\} - \tr\{\bbQ(\betax; \taup, \rho_{U})^{\T}\bbV^{-1}(\bbSig_{\Gamma} + \taup\bbSig_{\rm p} + (\betax^{2} - 2\rho_{U}\betax) \bbSig_{\gamma})\} \\
    			& = (\rho_{U} -\betax)\tr\{\bbV^{-1}\bbD_{\gamma}\} - \tr\{\bbQ(\betax; \taup, \rho_{U})^{\T}\bbV^{-1}(\bbD_{\Gamma} + \taup\bbD_{\rm p} + (\betax^{2} - 2\rho_{U}\betax) \bbD_{\gamma})\},
    		\end{aligned}
    	\end{equation*}
    	where the last equality holds because $\bbQ(\betax; \taup, \rho_{U})$ is diagonal. Recall that $\bbQ(\betax; \taup, \rho_{U}) = (\rho_{U} - \betax) \bbD_{\gamma}(\bbD_{\Gamma} + \taup \bbD_{\rm p} + (\betax^{2} - 2\rho_{U}\betax) \bbD_{\gamma})^{-1}$. Thus
    	\begin{equation*}
    		\begin{aligned}
    			& E\{\{\hgamma - \bbQ(\betax;\taup,\rho_U)(\hGamma - \betax \hgamma)\}^{\T} \bbV^{-1}(\hGamma - \betax \hgamma)\} \\
    			& =  (\rho_{U} -\betax)\tr\{\bbV^{-1}\bbD_{\gamma}\} - \tr\{\bbQ(\betax; \taup, \rho_{U})^{\T}\bbV^{-1}(\bbD_{\Gamma} + \taup\bbD_{\rm p} + (\betax^{2} - 2\rho_{U}\betax) \bbD_{\gamma})\} \\
    			& = (\rho_{U} -\betax)\tr\{\bbV^{-1}\bbD_{\gamma}\} - 
                (\rho_{U} -\betax)\tr\{\bbV^{-1}\bbD_{\gamma}\}\\
    			& = 0.
    		\end{aligned}
    	\end{equation*}
    \end{proof}
    
    \paragraph{Proof of Theorem \ref{thm: AN os}}
    \begin{proof}
{
For the one-sample argument, let $a_n=\sqrt{\|\bbSig\|}/\|\bgamma\|$, let $R_n$ be as in Condition \ref{cond: matrices os}, and write $p_n=\|\bbSig_{\rm p}\|$ and
\[
 t_{n,\rm os}=\frac{\|\bbSig\|}{p_n\sqrt q}+\frac{R_n}{np_n}.
\]
The additional conditions give $a_n=O(R_n)$, $q^{-1/2}=O(R_n)$, and $R_n\to0$. The explicit nuisance estimators satisfy $\hat\rho-\rho_U=O_P(R_n)$ and $\htau_{\rm os}-\taup=O_P(t_{n,\rm os})$. Thus their errors in the denominator of $\bbQ$ are $o_P(\|\bbSig\|)$, even when $\taup=0$ or $\rho_U=\betax$.

To verify the first-order remainders, the ratio expansion for $\hbeta_2$ has error $o_P(a_n)$. The explicit expression for $\hat\rho$ has error $O_P\{R_n(q^{-1/2}+n^{-\kappa}\sqrt q)\}=o_P(R_n)$ after retaining $s_\rho+\mu_\beta\mu_2^{-1}s_2$. After retaining $s_{\tau,\rm os}+\mu_{\rho,\tau}(\hat\rho-\rho_U)$, the remaining terms in the explicit expression for $\htau_{\rm os}$ are $O_P\{(\|\bbSig\|/p_n)(q^{-1}+n^{-\kappa})\}$. In the first score, the second-order Taylor terms have order $O_P(qR_n^2)$; the first-derivative fluctuations and nuisance-matrix substitution are also $o_P(qR_n)$. Since $\mu_1a_n\asymp qR_n$, these terms give an $o_P(a_n)$ remainder for $\hbeta_{1,\rm os}$. Stability of the first derivative in the same neighborhood establishes local existence and uniqueness as in Theorem \ref{thm: AN origin}. We now give the resulting expansions.
}
    	
    	Define $s_{\rho}(\beta, \rho) = \tr\{\bbV_{*}^{-1}\bbD_{\gamma}\}^{-1}[\hgamma^{\T}\bbV_{*}^{-1}(\hGamma - \beta\hgamma) - (\rho - \beta)\tr\{\bbV_{*}^{-1}\bbD_{\gamma}\}]$, $s_{\tau, \rm os}(\beta, \tau, \rho) = \tr\{\bbV_{*}^{-1}\bbD_{\rm p}\}^{-1}[(\hGamma - \beta \hgamma)^{\T}\bbV_{*}^{-1}(\hGamma - \beta \hgamma) - \tr\{\bbV_{*}^{-1}(\bbD_{\Gamma} + (\beta^{2} - 2 \rho\beta)\bbD_{\gamma} + \tau \bbD_{\rm p})\}]$ and $s_{1, \rm os}(\beta, \tau, \rho) = \{\hgamma - \bbQ(\beta; \tau, \rho)(\hGamma - \beta\hgamma)\}^{\T}\bbV_{*}^{-1}(\hGamma - \beta\hgamma)$. Denote $s_{\rho}(\betax, \rho_{U})$, $s_{\tau, \rm os}(\betax, \taup, \rho_{U})$, and $s_{1, \rm os}(\betax, \taup, \rho_{U})$ as $s_{\rho}$, $s_{\tau, \rm os}$, and $s_{1, \rm os}$, respectively. According to Conditions \ref{cond: matrices origin}, \ref{cond: selection}, \ref{cond: angle}, \ref{cond: matrices pleiotropy}, \ref{cond: bound pleiotropy}, and \ref{cond: matrices os}, arguments similar to the proof of \eqref{eq: expan htaup} and \eqref{eq: expan hbeta1 p} can show that
    	\begin{align}
    		&\hbeta_{2} - \betax = \mu_{2}^{-1}s_{2} + {o_P(a_n)}, \label{eq: expan hbeta2 os}\\
    		&\hat{\rho} - \rho_{U} = s_{\rho} + \mu_{\beta}(\hbeta_{2} - \betax) + {o_P(R_n)}
    		= s_{\rho} + \mu_{\beta}\mu_{2}^{-1}s_{2} + {o_P(R_n)}, \label{eq: expan rho hat os}\\
    		&\htau_{\rm os} - \taup = s_{\tau, \rm os} + \mu_{\rho, \tau}(\hat{\rho} - \rho_{U}) + {o_P(t_{n,\rm os})},\nonumber\\
    		&\phantom{\htau_{\rm os} - \taup} 
    		= s_{\tau, \rm os} + \mu_{\rho, \tau}s_{\rho} + \mu_{\rho, \tau}\mu_{\beta}\mu_{2}^{-1}s_{2} + {o_P(t_{n,\rm os})}, \label{eq: expan tau hat os}
    	\end{align}
    	and
    	\begin{equation}\label{eq: expan hbeta1 os}
    		\begin{aligned}
    			\hbeta_{1, \rm os} - \betax 
    			& = \mu_{1}^{-1}\{s_{1,\rm os} + \mu_{\tau, \rm os}(\htau_{\rm os} - \taup) + \mu_{\rho}(\hat{\rho} - \rho_{U})\} + {o_P(a_n)}\\
    			& = \mu_{1}^{-1}\{s_{1,\rm os} + \mu_{\tau, \rm os}s_{\tau, \rm os} + (\mu_{\tau, \rm os}\mu_{\rho, \tau} + \mu_{\rho})s_{\rho} + (\mu_{\tau, \rm os}\mu_{\rho, \tau} + \mu_{\rho})\mu_{\beta}\mu_{2}^{-1}s_{2}\}\\
    			& \quad  + {o_P(a_n)}.
    		\end{aligned}
    	\end{equation}
    	where $\mu_{\beta}$, $\mu_{\tau, \rm os}$, $\mu_{\rho, \tau}$, and $\mu_{\rho}$ are defined before Condition \ref{cond: linearity os}.
    	Combining \eqref{eq: expan hbeta2 os} and \eqref{eq: expan hbeta1 os}, we have
    	\begin{equation}\label{eq: expan tbeta os}
    		\begin{aligned}
    			\enbeta_{\rm DEEM, os} - \betax 
    			&= w_{1}\mu_{1}^{-1}\{s_{1,\rm os} + \mu_{\tau, \rm os}s_{\tau, \rm os} + (\mu_{\tau, \rm os}\mu_{\rho, \tau} + \mu_{\rho})s_{\rho}\}\\
    			& \quad  + \{w_{1}\mu_1^{-1}(\mu_{\tau, \rm os}\mu_{\rho, \tau} + \mu_{\rho})\mu_{\beta} + w_{2}\}\mu_{2}^{-1}s_{2} + {o_P(\|\bw\|a_n)}.
    		\end{aligned}
    	\end{equation}
{
Condition \ref{cond: angle} gives $[\bbPsi_{\rm os}]_{22}\asymp a_n^2$. Its bounded condition number therefore implies $\lambda_{\min}(\bbPsi_{\rm os})\asymp\lambda_{\max}(\bbPsi_{\rm os})\asymp a_n^2$, making the ensemble remainder negligible after standardization, including for $n$-dependent weights.
}
    	By \eqref{eq: normal model}, we have
    	\begin{equation}\label{eq: dist equiv os}
    		\begin{aligned}
    			&w_{1}\mu_{1}^{-1}\{s_{1,\rm os} + \mu_{\tau, \rm os}s_{\tau, \rm os} + (\mu_{\tau, \rm os}\mu_{\rho, \tau} + \mu_{\rho})s_{\rho}\} + \{w_{1}\mu_1^{-1}(\mu_{\tau, \rm os}\mu_{\rho, \tau} + \mu_{\rho})\mu_{\beta} + w_{2}\}\mu_{2}^{-1}s_{2} \\
    			&\stackrel{d}{=}
    			w_{1}\mu_{1}^{-1}
    			\begin{pmatrix}
    				\bxi\\
    				\taup^{-\frac{1}{2}}\bbeta_{G}
    			\end{pmatrix}^{\T}
    			\bbH_{\rm os}
    			\begin{pmatrix}
    				\bxi\\
    				\taup^{-\frac{1}{2}}\bbeta_{G}
    			\end{pmatrix} -w_1\mu_1^{-1}\tr(\bbH_{\rm os}) + \bg_{\rm os}^{\T}
    			\begin{pmatrix}
    				\bxi\\
    				\taup^{-\frac{1}{2}}\bbeta_{G}
    			\end{pmatrix},
    		\end{aligned}
    	\end{equation}
    	where $\bg_{\rm os} = w_{1}\mu_{1}^{-1}\bl_{1\rm os} + w_{2}\mu_{2}^{-1}\bl_{2\rm os}$, $\bbH_{\rm os}$ is defined before Condition \ref{cond: linearity os}.
    	
    	Define
    	\[
    	\begin{aligned}
    		\bbSig_{\rm os}^{\dag} & = 
    		\begin{pmatrix}
    			\bbSig_{\gamma} & (\rho_{U} - \betax)\bbSig_{\gamma}\\
    			(\rho_{U} - \betax)\bbSig_{\gamma} & \bbSig_{\Gamma} + (\betax^{2} - 2\rho_{U}\betax)\bbSig_{\gamma} +\taup\bbSig_{\rm p}
    		\end{pmatrix},\\
    		\bbM_{\rm os} & = \frac{1}{2}
    		\begin{pmatrix}
    			\mathbf{0} & \bbV_{*}^{-1}\\
    			\bbV_{*}^{-1} & - \bbQ(\betax; \taup, \rho_{U})^{\T}\bbV_{*}^{-1} - \bbV_{*}^{-1}\bbQ(\betax; \taup, \rho_{U})
    		\end{pmatrix},\\
    		\bbM_{\rho} & = 
    		\frac{1}{2}\begin{pmatrix}
    			\mathbf{0} & \bbV_{*}^{-1}\\
    			\bbV_{*}^{-1} & 0
    		\end{pmatrix}.
    	\end{aligned}
    	\]
    	Let
    	\begin{equation}\label{eq: Psi os}
    		\begin{aligned}
    			\bbOmega_{\rm os} =
    			\begin{pmatrix}
    				\omega_{11} &  \omega_{12} & \omega_{13} 
    				& \omega_{14}\\
    				\omega_{21} & 
    				\omega_{22} & 0 & \omega_{24}\\
    				\omega_{31} & 0 & \omega_{33} & \omega_{34}\\
    				\omega_{41} & \omega_{42} & \omega_{43} & \omega_{44}
    			\end{pmatrix}.
    		\end{aligned}
    	\end{equation}
    	be the variance-covariance matrix of $(s_{1, \rm os}, s_{2}, s_{\tau, \rm os},s_\rho)^{\T}$ calculated under the assumption that the components of $\bbeta_{G}$ follow a normal distribution, where 
    	\begin{align*}
    		\omega_{11} & = \bgamma^{\T}\bbV_{*}^{-1}\bbSig\bbV_{*}^{-1}\bgamma + 2\tr\{\bbM_{\rm os}\bbSig_{\rm os}^{\dag}\bbM_{\rm os}\bbSig_{\rm os}^{\dag}\},\\
    		\omega_{12} & = \omega_{21} = \bgamma^{\T}\bbV_{*}^{-1}\bbSig\bbV_{*}^{-1}\tgamma\\
    		\omega_{13} & = \omega_{31} = 2\tr\{\bbV_{*}^{-1}\bbD_{\rm p}\}^{-1}\tr\{\bbM_{\rm os}\bbSig_{\rm os}^{\dag}\bbM_{\tau}\bbSig_{\rm os}^{\dag}\},\\
    		\omega_{14} & = \omega_{41} = \tr\{\bbV_{*}^{-1}\bbD_{\gamma}\}^{-1}[\bgamma^{\T}\bbV_{*}^{-1}\bbSig\bbV_{*}^{-1}\bgamma + 2\tr\{\bbM_{\rm os}\bbSig_{\rm os}^{\dag}\bbM_{\rho}\bbSig_{\rm os}^{\dag}\}],\\
    		\omega_{22} & = \tgamma^{\T}\bbV_{*}^{-1}\bbSig\bbV_{*}^{-1}\tgamma\\
    		\omega_{24} & = \omega_{42} = \tr\{\bbV_*^{-1}\bbD_\gamma\}^{-1}\bgamma^{\T}\bbV_{*}^{-1}\bbSig\bbV_{*}^{-1}\tgamma\\
    		\omega_{33} & = 2\tr\{\bbV_{*}^{-1}\bbD_{\rm p}\}^{-2}\tr\{\bbM_{\tau}\bbSig_{\rm os}^{\dag}\bbM_{\tau}\bbSig_{\rm os}^{\dag}\},\\
    		\omega_{34} & = \omega_{43} = 2\tr\{\bbV_{*}^{-1}\bbD_{\rm p}\}^{-1}\tr\{\bbV_{*}^{-1}\bbD_{\rm \gamma}\}^{-1} \tr\{\bbM_{\tau}\bbSig_{\rm os}^{\dag}\bbM_{\rho}\bbSig_{\rm os}^{\dag}\},\\
    		\omega_{44} & = \tr\{\bbV_{*}^{-1}\bbD_{\rm \gamma}\}^{-2}[\bgamma^{\T}\bbV_{*}^{-1}\bbSig\bbV_{*}^{-1}\bgamma + 2\tr\{\bbM_{\rho}\bbSig_{\rm os}^{\dag}\bbM_{\rho}\bbSig_{\rm os}^{\dag}\}].
    	\end{align*}
    	Define
    	\[
    	\bbPsi_{\rm os} = 
    	\begin{pmatrix}
    		\mu_{1}^{-1} & 0\\
    		(\mu_{\tau, \rm os}\mu_{\rho, \tau} + \mu_{\rho})\mu_{\beta}\mu_{1}^{-1}\mu_{2}^{-1} & \mu_2^{-1}\\
    		\mu_{1}^{-1}\mu_{\tau, \rm os}
    		& 0\\
    		\mu_{1}^{-1}(\mu_{\tau, \rm os}\mu_{\rho, \tau} + \mu_{\rho}) & 0
    	\end{pmatrix}^{\T}
    	\bbOmega_{\rm os}
    	\begin{pmatrix}
    		\mu_{1}^{-1} & 0\\
    		(\mu_{\tau, \rm os}\mu_{\rho, \tau} + \mu_{\rho})\mu_{\beta}\mu_{1}^{-1}\mu_{2}^{-1} & \mu_2^{-1}\\
    		\mu_{1}^{-1}\mu_{\tau, \rm os}
    		& 0\\
    		\mu_{1}^{-1}(\mu_{\tau, \rm os}\mu_{\rho, \tau} + \mu_{\rho}) & 0
    	\end{pmatrix}.
    	\]
    	Let $\bxi^{\dag}$ be a $(2q+d)$-dimensional random vector with independent standard normal components and $f_{w, \rm os}(t)$ and $f_{w, \rm os}^{\prime}(t)$ be the characteristic functions of
    	\[
    	(\bw^{\T}\bbPsi_{\rm os}\bw)^{-\frac{1}{2}}\left\{w_{1}\mu_{1}^{-1}
    	\begin{pmatrix}
    		\bxi\\
    		\taup^{-\frac{1}{2}}\bbeta_{G}
    	\end{pmatrix}^{\T}
    	\bbH_{\rm os}
    	\begin{pmatrix}
    		\bxi\\
    		\taup^{-\frac{1}{2}}\bbeta_{G}
    	\end{pmatrix} -w_1\mu_1^{-1}\tr(\bbH_{\rm os}) + \bg_{\rm os}^{\T}
    	\begin{pmatrix}
    		\bxi\\
    		\taup^{-\frac{1}{2}}\bbeta_{G}
    	\end{pmatrix}
    	\right\}
    	\]
    	and
    	\[
    	(\bw^{\T}\bbPsi_{\rm os}\bw)^{-\frac{1}{2}}\left(w_{1}\mu_{1}^{-1}
    	\bxi^{\dag\T}
    	\bbH_{\rm os}
    	\bxi^{\dag}-w_1\mu_1^{-1}\tr(\bbH_{\rm os}) 
    	+ \bg_{\rm os}^{\T}
    	\bxi^{\dag}
    	\right),
    	\]
    	respectively.
    	{
Apply the Gaussian-replacement argument in the proof of
Theorem \ref{thm: AN pleiotropy} with $\bbH_{\rm p}$, $\bg_{\rm p}$,
$\mathbf K_{\rm p}$, and $\bbPsi_{\rm p}$ replaced by $\bbH_{\rm os}$,
$\bg_{\rm os}$, $\mathbf K_{\rm os}$, and $\bbPsi_{\rm os}$, respectively.
Condition \ref{cond: linearity os} supplies the same bounds on the normalized
linear coefficients, diagonal terms, and off-diagonal quadratic influences.
Thus $f_{w,{\rm os}}(t)-f_{w,{\rm os}}^{\prime}(t)\to0$ in probability and
$f_{w,{\rm os}}^{\prime}(t)\to\exp(-t^2/2)$
}
    	in probability and hence
    	\[
    	f_{w, \rm os}(t) \to \exp\left(-\frac{t^{2}}{2}\right)
    	\]
    	in probability. Combining this with \eqref{eq: expan tbeta os} and \eqref{eq: dist equiv os}, we obtain the asymptotic normality result in Theorem \ref{thm: AN os} according to Corollary 3 in \cite{bulinski2017conditional}. A plug-in estimator for $\bbPsi_{\rm os}$ can be constructed similarly to $\hPsi$.
    \end{proof}
    
    \paragraph{Consistent estimator for $\bbPsi_{\rm os}$}
    To estimate $\bbPsi_{\rm os}$, we need estimators $\hSig_{\gamma}$, $\hSig_{\Gamma}$, and $\hSig_{\rm p}$ for $\bbSig_{\gamma}$, $\bbSig_{\Gamma}$, and $\bbSig_{\rm p}$. The construction of $\hSig_{\gamma}$, $\hSig_{\Gamma}$, and $\hSig_{\rm p}$ based on summary statistics is discussed in Section \ref{app: est Sig}.
    Suppose $\|\hSig_{\gamma} - \bbSig_{\gamma}\| / \|\bbSig_{\gamma}\| = o_{P}(1)$, $\|\hSig_{\Gamma} - \bbSig_{\Gamma}\| / \|\bbSig_{\Gamma}\| = o_{P}(1)$, and $\|\hSig_{\rm p} - \bbSig_{\rm p}\| / \|\bbSig_{\rm p}\| = o_{P}(1)$. Let $\hD_{\rm F, os} = 
    	\hD_{\Gamma} + (\hbeta_{2}^{2} - 2 \hat{\rho}\hbeta_{2})\hD_{\gamma} + \htau_{\rm os}\hD_{\rm p}$,
    \[
       	\hSig_{\rm F, os} = 
    	\hSig_{\Gamma} + (\hbeta_{2}^{2} - 2 \hat{\rho}\hbeta_{2})\hSig_{\gamma} + \htau_{\rm os}\hSig_{\rm p},\
    	\hSig_{\rm os}^{\dag} = 
    	\begin{pmatrix}
    		\hSig_{\gamma} & (\hat{\rho} - \hbeta_{2})\hSig_{\gamma}\\
    		(\hat{\rho} - \hbeta_{2})\hSig_{\gamma} & \hSig_{\rm F, os}
    	\end{pmatrix},
    \]
    \[
    	\widehat{\bbM}_{\rm os} = \frac{1}{2}
    	\begin{pmatrix}
    		\mathbf{0} & \bbV^{-1}\\
    		\bbV^{-1} & - \hQ(\hbeta_{2}; \htau_{\rm os}, \hat{\rho})^{\T}\bbV^{-1} - \bbV^{-1}\hQ(\hbeta_{2}; \htau_{\rm os}, \hat{\rho}),
    	\end{pmatrix}.
    \] 
{
The preceding expansions give $\hat\rho-\rho_U=O_P(R_n)=o_P(1)$ and
\[
 \frac{|\htau_{\rm os}-\taup|\,\|\bbSig_{\rm p}\|}{\|\bbSig\|}
 =O_P\left(q^{-1/2}+\frac{R_n}{n\|\bbSig\|}\right)=o_P(1).
\]
These scaled statements establish the required covariance accuracy. Absolute consistency of $\htau_{\rm os}$ also follows from Condition \ref{cond: matrices pleiotropy}(ii), $\|\bgamma\|\le C$, and $\|\bbSig\|\to0$. By Condition \ref{cond: matrices os}, $\|\bbSig_\gamma\|+\|\bbSig_\Gamma\|=O(\|\bbSig\|)$; the identity defining $\bbSig$ then also gives $\taup\|\bbSig_{\rm p}\|=O(\|\bbSig\|)$. Consequently the component covariance errors assumed above and the nuisance substitution errors are negligible relative to $\|\bbSig\|$.
}
    Combining these with Condition \ref{cond: matrices origin}, Condition \ref{cond: selection}, Condition \ref{cond: angle} (i), Condition \ref{cond: matrices pleiotropy}, and Condition \ref{cond: matrices os}, we have $\|\hD_{\rm F, os} - \bbD_{\rm F, os}\| / \|\bbD_{\rm F, os}\| = o_{P}(1)$, $\|\hSig_{\rm F, os} - \{\bbSig_{\Gamma} + (\betax^{2} - 2\rho_{U}\betax)\bbSig_{\gamma} +\taup\bbSig_{\rm p}\}\| / \|\bbSig_{\Gamma} + (\betax^{2} - 2\rho_{U}\betax)\bbSig_{\gamma} +\taup\bbSig_{\rm p}\| = o_{P}(1)$, $\|\widehat{\bbSig}_{\rm os}^{\dag} - \bbSig_{\rm os}^{\dag}\| / \|\bbSig_{\rm os}^{\dag}\| = o_{P}(1)$, and $\|\widehat{\bbM}_{\rm os} - \bbM_{\rm os}\| / \|\bbM_{\rm os}\| = o_{P}(1)$.
    $\mu_{1}$, $\mu_{2}$, $\mu_{\beta}$, $\mu_{\tau, \rm os}$, $\mu_{\rho, \tau}$, and $\mu_{\rho}$ can be consistently estimated by
    \[
    \begin{aligned}
        &\hat{\mu}_{1} = \hgamma^{\T}\bbV^{-1}\hgamma - \tr\{\bbV^{-1}\hSig_{\gamma}\}, \\
        &\hat{\mu}_{2} = \hgamma^{\T}\bbV^{-1}\tgamma, \\   
    	&\hat{\mu}_{\beta} = - \tr\{\bbV^{-1}\hD_{\gamma}\}^{-1}\hat{\mu}_{1},\\
    	&\hat{\mu}_{\tau, \rm os} = (\hat{\rho} - \hat{\beta}_{2}) \tr\{\bbV^{-1}\hD_{\rm p}\hD_{\rm F, os}^{-1}\hD_{\gamma}\},\\
    	&\hat{\mu}_{\rho, \tau} = 2\hat{\beta}_{2}\tr\{\bbV^{-1}\hD_{\rm p}\}^{-1}\tr\{\bbV^{-1}\hD_{\gamma}\},\\ 
    	&\hat{\mu}_{\rho} = - \tr\{\bbV^{-1}\hD_{\gamma}\} - 2(\hat{\rho} - \hat{\beta}_{2})\hat{\beta}_{2}\tr\{\bbV^{-1}\hD_{\gamma}\hD_{\rm F, os}^{-1}\hD_{\gamma}\}, 
    \end{aligned}
    \]
    respectively.
    Let $\psi_{{\rm os}, ij}$ be the $ij$th element of $\bbPsi_{\rm os}$ for $i = 1, 2$ and $j = 1, 2$.
    Define
    \[
    	\begin{aligned}
    		\widehat{\bbOmega}_{\rm os} =
    		\begin{pmatrix}
    				\widehat{\omega}_{11} &  \widehat{\omega}_{12} & \widehat{\omega}_{13} 
    				& \widehat{\omega}_{14}\\
    				\widehat{\omega}_{21} & 
    				\widehat{\omega}_{22} & 0 & \widehat{\omega}_{24}\\
    				\widehat{\omega}_{31} & 0 & \widehat{\omega}_{33} & \widehat{\omega}_{34}\\
    				\widehat{\omega}_{41} & \widehat{\omega}_{42} & \widehat{\omega}_{43} & \widehat{\omega}_{44}
    		\end{pmatrix},
    	\end{aligned}
    	\]	
    where
    \[
       \begin{aligned}
    		\widehat{\omega}_{11} & = \hat{v}_{\gamma, {\rm os}} + 2\tr\{\widehat{\bbM}_{\rm os}\widehat{\bbSig}_{\rm os}^{\dag}\widehat{\bbM}_{\rm os}\widehat{\bbSig}_{\rm os}^{\dag}\},\\
    		\widehat{\omega}_{12} & = \widehat{\omega}_{21} = \hgamma^{\T}\bbV^{-1}\hSig_{\rm F, os}\bbV^{-1}\tgamma\\
    		\widehat{\omega}_{13} & = \widehat{\omega}_{31} = 2\tr\{\bbV^{-1}\hD_{\rm p}\}^{-1}\tr\{\widehat{\bbM}_{\rm os}\widehat{\bbSig}_{\rm os}^{\dag}\widehat{\bbM}_{\tau}\widehat{\bbSig}_{\rm os}^{\dag}\},\\
    		\widehat{\omega}_{14} & = \widehat{\omega}_{41} = \tr\{\bbV^{-1}\hD_{\gamma}\}^{-1}[\hat{v}_{\gamma, {\rm os}} + 2\tr\{\widehat{\bbM}_{\rm os}\widehat{\bbSig}_{\rm os}^{\dag}\widehat{\bbM}_{\rho}\widehat{\bbSig}_{\rm os}^{\dag}\}],\\
    		\widehat{\omega}_{22} & = \tgamma^{\T}\bbV^{-1}\hSig_{\rm F, os}\bbV^{-1}\tgamma\\
    		\widehat{\omega}_{24} & = \widehat{\omega}_{42} = \tr\{\bbV^{-1}\hD_\gamma\}^{-1}\hgamma^{\T}\bbV^{-1}\hSig_{\rm F, os}\bbV^{-1}\tgamma\\
    		\widehat{\omega}_{33} & = 2\tr\{\bbV^{-1}\hD_{\rm p}\}^{-2}\tr\{\widehat{\bbM}_{\tau}\widehat{\bbSig}_{\rm os}^{\dag}\widehat{\bbM}_{\tau}\widehat{\bbSig}_{\rm os}^{\dag}\},\\
    		\widehat{\omega}_{34} & = \widehat\omega_{43} = 2\tr\{\bbV^{-1}\hD_{\rm p}\}^{-1}\tr\{\bbV^{-1}\hD_{\rm \gamma}\}^{-1} \tr\{\widehat{\bbM}_{\tau}\widehat{\bbSig}_{\rm os}^{\dag}\widehat{\bbM}_{\rho}\widehat{\bbSig}_{\rm os}^{\dag}\},\\
    		\widehat{\omega}_{44} & = \tr\{\bbV^{-1}\hD_{\rm \gamma}\}^{-2}[\hat{v}_{\gamma, {\rm os}} + 2\tr\{\widehat{\bbM}_{\rho}\widehat{\bbSig}_{\rm os}^{\dag}\widehat{\bbM}_{\rho}\widehat{\bbSig}_{\rm os}^{\dag}\}],
    	\end{aligned}
    \] 
    \[
     \widehat{\bbM}_{\tau} = 
                \begin{pmatrix}
    			\mathbf{0} & \mathbf{0}\\
    			\mathbf{0} & \bbV^{-1}
    		\end{pmatrix},\quad
    \widehat{\bbM}_{\rho} = \frac{1}{2}\begin{pmatrix}
    			\mathbf{0} & \bbV^{-1}\\
    			\bbV^{-1} & \mathbf{0}
    		\end{pmatrix},
    \]
    and $\hat{v}_{\gamma, {\rm os}} = \hgamma^{\T}\bbV^{-1}\hSig_{\rm F, os}\bbV^{-1}\hgamma - \tr\{\bbV^{-1}\hSig_{\rm F, os}\bbV^{-1}\hSig_{\gamma}\}$.
    Let
    \[
    \begin{aligned}
        \widehat{\bbPsi}_{\rm os} & = 
      \begin{pmatrix}
          \hat{\psi}_{{\rm os}, 11} & \hat{\psi}_{{\rm os}, 12}\\
          \hat{\psi}_{{\rm os}, 21} & \hat{\psi}_{{\rm os}, 22}
      \end{pmatrix}\\
      &     =
      \begin{pmatrix}
    		\hat{\mu}_{1}^{-1} & 0\\
    		(\hat{\mu}_{\tau, \rm os}\hat{\mu}_{\rho, \tau} + \hat{\mu}_{\rho})\hat{\mu}_{\beta}\hat{\mu}_{1}^{-1}\hat{\mu}_{2}^{-1} & \hat\mu_2^{-1}\\
    		\hat{\mu}_{1}^{-1}\hat{\mu}_{\tau, \rm os}
    		& 0\\
    		\hat{\mu}_{1}^{-1}(\hat{\mu}_{\tau, \rm os}\hat{\mu}_{\rho, \tau} + \hat{\mu}_{\rho}) & 0
    	\end{pmatrix}^{\T}
    	\widehat{\bbOmega}_{\rm os}
    	\begin{pmatrix}
    		\hat{\mu}_{1}^{-1} & 0\\
    		(\hat{\mu}_{\tau, \rm os}\hat{\mu}_{\rho, \tau} + \hat{\mu}_{\rho})\hat{\mu}_{\beta}\hat{\mu}_{1}^{-1}\hat{\mu}_{2}^{-1} & \hat\mu_2^{-1}\\
    		\hat{\mu}_{1}^{-1}\hat{\mu}_{\tau, \rm os}
    		& 0\\
    		\hat{\mu}_{1}^{-1}(\hat{\mu}_{\tau, \rm os}\hat{\mu}_{\rho, \tau} + \hat{\mu}_{\rho}) & 0
    	\end{pmatrix}.
    \end{aligned}
    \]
    Assume that the condition number of $\bbPsi_{\rm os}$ is bounded.
    Then, we have $(\hat{\psi}_{{\rm os}, ij} - \psi_{{\rm os}, ij}) / \|\bbPsi_{\rm os}\| = o_{P}(1)$ for $i = 1, 2$ and $j = 1, 2$ according to the above convergence results and Conditions \ref{cond: matrices origin}, \ref{cond: selection}, \ref{cond: angle}, \ref{cond: matrices pleiotropy}, and \ref{cond: matrices os}. Thus, $\|\hPsi_{\rm os} - \bbPsi_{\rm os}\|/\|\bbPsi_{\rm os}\| \leq \|\hPsi_{\rm os} - \bbPsi_{\rm os}\|_{\rm F}/\|\bbPsi_{\rm os}\| = o_{P}(1)$. This establishes the consistency of $\hPsi_{\rm os}$. The optimal weight $(\boldsymbol{1}^{\T}\bbPsi_{\rm os}^{-1}\boldsymbol{1})^{-1}\boldsymbol{1}^{\T}\bbPsi_{\rm os}^{-1}$ that minimizes the asymptotic variance in this case can be estimated by $(\boldsymbol{1}^{\T}\hPsi_{\rm os}^{-1}\boldsymbol{1})^{-1}\boldsymbol{1}^{\T}\hPsi_{\rm os}^{-1}$, whose consistency follows from the consistency of $\hPsi_{\rm os}$.
    
    \section{Estimation based on Summary Statistics}\label{app: est summary}
    \subsection{Estimate $\bbSig_{\gamma}$, $\bbSig_{\Gamma}$, $\bbD_{\gamma}$, $\bbD_{\Gamma}$, and $\bbSig_{\rm p}$}\label{app: est Sig}
    In this subsection, we introduce the summary statistics based estimators for $\bbSig_{\gamma}$ and $\bbSig_{\Gamma}$. Let $(X_{1}, \bG_{1}^{(e)}), \dots, (X_{n_{e}}, \bG_{n_{e}}^{(e)})$ be independent and identically distributed copies of $(X, \bG)$ from the exposure sample, $(Y_{1}, \bG_{1}^{(o)}), \dots, (Y_{n_{o}}, \bG_{n_{o}}^{(o)})$ be independent and identically distributed copies of $(Y, \bG)$ from the outcome sample, where $n_{e}$ and $n_{o}$ are the sample sizes of the exposure and the outcome sample, respectively. The exposure and outcome samples are independent of each other in the two sample setting and they are the same sample in the one-sample setting. 
    Let $\bbY = (Y_{1} - \bar{Y},\dots, Y_{n} - \bar{Y})^{\T}$, $\bbX = (X_{1} - \bar{X},\dots, X_{n} - \bar{X})^{\T}$, $\bbG^{(e)} = (\bG_{1}^{(e)} - \bar{\bG}^{(e)},\dots, \bG_{n_{e}}^{(e)} - \bar{\bG}^{(e)})^{\T}$, and $\bbG^{(o)} = (\bG_{1}^{(o)} - \bar{\bG}^{(o)},\dots, \bG_{n_{o}}^{(o)} - \bar{\bG}^{(o)})^{\T}$ be the centered outcome vector, exposure vector, and the genotype matrices, respectively, where $\bar{Y} = \sum_{i=1}^{n_{o}}Y_{i}/n_{o}$, $\bar{X} = \sum_{i=1}^{n_{e}}X_{i}/n_{e}$, and $\bar{\bG}^{(e)} = \sum_{i=1}^{n_{e}}\bG_{i}^{(e)}/n_{e}$, and $\bar{\bG}^{(o)} = \sum_{i=1}^{n_{o}}\bG_{i}^{(o)}/n_{o}$. 
    For $j = 1,\dots,p$, the marginal regression coefficients $\widehat{\gamma}_{j} = \bbG_{\cdot j}^{(e)\T}\bbX / (\bbG_{\cdot j}^{(e)\T}\bbG_{\cdot j}^{(e)})$ and $\widehat{\Gamma}_{j} = \bbG_{\cdot j}^{(o)\T}\bbY / (\bbG_{\cdot j}^{(o)\T}\bbG_{\cdot j}^{(o)\T})$ of $\bbX$, $\bbY$ on $\bbG_{\cdot j}^{(e)}$, $\bbG_{\cdot j}^{(o)}$ are available from GWAS summary statistics,  where $\bbG_{\cdot j}^{(e)}$, $\bbG_{\cdot j}^{(o)}$ is the $j$th column of $\bbG^{(e)}$, $\bbG^{(o)}$, respectively.  Besides the marginal regression coefficients $\widehat{\gamma}_{j}$ and $\widehat{\Gamma}_{j}$, the variance estimates
    $\widehat{\sigma}_{\gamma j}^{2} = n_{e}^{-1}(\bbG_{\cdot j}^{(e)\T}\bbG_{\cdot j}^{(e)})^{-1}(\bbX - \bbG_{\cdot j}^{(e)} \widehat{\gamma}_{j})^{\T}(\bbX - \bbG_{\cdot j}^{(e)} \widehat{\gamma}_{j})$ and 
    $\widehat{\sigma}_{\Gamma j}^{2} = n_{o}^{-1}(\bbG_{\cdot j}^{(o)\T}\bbG_{\cdot j}^{(o)})^{-1}(\bbY - \bbG_{\cdot j}^{(o)} \widehat{\Gamma}_{j})^{\T}(\bbY - \bbG_{\cdot j}^{(o)} \widehat{\Gamma}_{j})$ are also available from GWAS summary statistics for $j = 1,\dots, d$ \citep{zhu2017bayesian}. In addition, the block diagonal approximation $\hR_{\rm ref}$ of the correlations matrix among SNPs  (``LD matrix") can be obtained from public reference databases such as the 1000 Genomes Project \citep{10002010map}. 
    
    According to Equation (2.4) in \cite{zhu2017bayesian}, $\hgamma$ approximately follows $N(\bgamma, \bbS_{\gamma}\bbR\bbS_{\gamma})$ and $\hGamma$ approximately follows $N(\bGamma, \bbS_{\Gamma}\bbR\bbS_{\Gamma})$, where $\bbR$ is the population LD matrix, 
    \[
    \bbS_{\gamma} = \diag\{\sqrt{n_{e}^{-1}\var(G_{1})^{-1}\var(X)}, \dots, \sqrt{n_{e}^{-1}\var(G_{d})^{-1}\var(X)}\},
    \] 
    and 
    \[
    \bbS_{\Gamma} = \diag\{\sqrt{n_{o}^{-1}\var(G_{1})^{-1}\var(Y)}, \dots, \sqrt{n_{o}^{-1}\var(G_{d})^{-1}\var(Y)}\}.
    \]
    In this paper, we assume $\hgamma$ and $\hGamma$ follow exactly the normal distribution \eqref{eq: normal model} with $\bbSig_{\gamma} = \bbS_{\gamma}\bbR\bbS_{\gamma}$ and $\bbSig_{\Gamma} = \bbS_{\Gamma}\bbR\bbS_{\Gamma}$.
    Define $\hS_{\gamma}^{2} = \diag\{\widehat{\sigma}_{\gamma 1}^{2} + n_{e}^{-1}\widehat{\gamma}_{1}^{2}, \dots, \widehat{\sigma}_{\gamma d}^{2} + n_{e}^{-1}\widehat{\gamma}_{d}^{2}\}$
    and $\hS_{\Gamma}^{2} = \diag\{\widehat{\sigma}_{\Gamma 1}^{2} + n_{o}^{-1}\widehat{\Gamma}_{1}^{2}, \dots, \widehat{\sigma}_{\Gamma d}^{2} + n_{o}^{-1}\widehat{\Gamma}_{d}^{2}\}$. 
    Notice that $n_{e}\widehat{\sigma}_{\gamma j}^{2} + \widehat{\gamma}_{j}^{2} = \bbX^{\T}\bbX/\bbG_{\cdot j}^{(e)\T}\bbG_{\cdot j}^{(e)} \to \var(G_{j})^{-1}\var(X)$ and  $n_{o}\widehat{\sigma}_{\Gamma j}^{2} + \widehat{\Gamma}_{j}^{2} = \bbY^{\T}\bbY/\bbG_{\cdot j}^{(o)\T}\bbG_{\cdot j}^{(o)} \to \var(G_{j})^{-1}\var(Y)$ in probability.
    Thus, $\bbSig_{\gamma}$ and $\bbSig_{\Gamma}$ can be estimated by $\hSig_{\gamma} = \hS_{\gamma}\hR_{\rm ref}\hS_{\gamma}$ and $\hSig_{\Gamma} = \hS_{\Gamma}\hR_{\rm ref}\hS_{\Gamma}$, respectively. Based on 
    $\hSig_{\gamma}$ and $\hSig_{\Gamma}$, the estimates $\hD_{\gamma}$, $\hD_{\Gamma}$ for $\bbD_{\gamma}$, $\bbD_{\Gamma}$ can be obtained according to the formulation provided in Section \ref{subsec: orDEEM} in the main text. 
    
    We then show the convergence result of $\hD_{\gamma}$ and $\hD_{\Gamma}$ in Remark \ref{remark: diag rate} in Section \ref{subsec: orDEEM} of the main text. We state the result for $\hD_{\gamma}$ and the result for $\hD_{\Gamma}$ follows similarly.
    Recall that \[
    \bbD_{\gamma} = \diag\left\{[\bbV^{-1}\bbSig_{\gamma}]_{11}/[\bbV^{-1}]_{11},\dots, [\bbV^{-1}\bbSig_{\gamma}]_{11}/[\bbV^{-1}]_{dd}\right\}
    \]
    and  
    \[
    \hD_{\gamma} = 
    \diag\left\{[\bbV^{-1}\hSig_{\gamma}]_{11}/[\bbV^{-1}]_{11}, \dots, [\bbV^{-1}\hSig_{\gamma}]_{dd}/[\bbV^{-1}]_{dd}\right\}.
    \]
    Since both $\hD_{\gamma}$ and $\bbD_{\gamma}$ are diagonal matrices, we have
    \[
    \|\hD_{\gamma} - \bbD_{\gamma}\| = \max_{j}\left\{\frac{|[\bbV^{-1}(\hSig_{\gamma} - \bbSig_{\gamma})]_{jj}|}{[\bbV^{-1}]_{jj}}\right\} \leq s(\bbV) \max_{i, j}|[\hSig_{\gamma} - \bbSig_{\gamma}]_{ij}|,
    \] 
    where $s(\bbV) = \max_{j} \sum_{i=1}^{d}|[\bbV^{-1}]_{ij}|/[\bbV^{-1}]_{jj}$ and the $\max$ is taken over $\{1, \dots, d\}$. The factor $s(\bbV)$ can be upper bounded by a universal constant if the off-diagonal elements of $\bbV^{-1}$ decay exponentially. The exponential decay assumption can be easily met as $\bbV$ is user-specified. Then, the convergence rate in the main text is established if we show that 
    \begin{equation}\label{eq: element-wise convergence}
    	\max_{i, j}|[\hSig_{\gamma} - \bbSig_{\gamma}]_{ij}| = O_{P}\left(\max_{j}[\bbSig_{\gamma}]_{jj}\sqrt{\frac{\log d}{n_{\rm ref}}}\right),
    \end{equation}
    where $n_{\rm ref}$ is the size of the reference sample.
    Suppose $X$ is subgaussian and $\var(G_{j})$ is bounded away from zero for $j = 1,\dots, d$. Considering that SNPs are bounded random variables, the Bernstein inequality combined with the union bound \citep{wainwright2019high} can show that $\max_{j}|[\hS_{\gamma} - \bbS_{\gamma}]_{jj} / [\bbS_{\gamma}]_{jj}| = O_{P}(\sqrt{\log d / n_{e}})$ and $\max_{i,j}|[\hR_{\rm ref} - \bbR]_{ij}| = O_{P}(\sqrt{\log d / n_{\rm ref}})$. Thus, 
    \[
    \max_{i,j}|[\hSig_{\gamma} - \bbSig_{\gamma}]_{ij}| = O_{P}(\max_{j}[S_{\gamma}]_{jj}^{2}\sqrt{\log d / n_{\rm ref}}) = O_{P}(\max_{j}[\bbSig_{\gamma}]_{jj}\sqrt{\log d / n_{\rm ref}}),
    \] 
    provided $n_{\rm ref} = O(n_{e})$, which establishes the convergence rate result \eqref{eq: element-wise convergence}.
    
    The convergence rate of the spectral norm $\|\hSig_{\gamma} - \bbSig_{\gamma}\|$ is $\|\bbSig_{\gamma}\|\sqrt{d / n_{\rm ref}}$ according to Theorem 1 in \cite{koltchinskii2017concentration}. Similar results can be established for $\hD_{\Gamma}$ and $\hSig_{\Gamma}$. 
    
    Next, we move on to the estimation of $\bbSig_{\rm p}$. Note that 
    \[
    \bbP = \diag\{\var(G_{1})^{-1/2}, \dots, \var(G_{d})^{-1/2}\}\bbR\diag\{\var(G_{1})^{1/2}, \dots, \var(G_{d})^{1/2}\}.
    \]
    According to the above arguments, we have $n_{o}\hS_{\Gamma}^{2} \approx \var(Y)\diag\{\var(G_{1}^{2})^{-1}, \dots, \var(G_{d}^{2})^{-1}\}$ for sufficiently large $n_{o}$. Thus
    $\bbP$ can be estimated by $\hS_{\Gamma}\hR_{\rm ref}\hS_{\Gamma}^{-1}$.
    Since $\bbSig_{\rm p} = \bbP\bbP^{\T}$, we can estimate  $\bbSig_{\rm p}$ using
    \[
    \hSig_{\rm p} = \hS_{\Gamma}\hR_{\rm ref}\hS_{\Gamma}^{-2}\hR_{\rm ref}\hS_{\Gamma}.
    \]
    For the case with SNP selection, we estimate  $\bbSig_{\rm p}$ using
    $\hSig_{\rm p} = \big[\hS_{\Gamma}\hR_{\rm ref}\hS_{\Gamma}^{-2}\hR_{\rm ref}\hS_{\Gamma}\big]_{\widetilde{\cA},\widetilde{\cA}}$ and other estimates can be constructed in the same manner.

    \subsection{Construction of the Working Covariance Matrix $\bbV$}\label{app: construction of V}
    Recall that the optimal choice for $\bbV$ is the variance-covariance matrix $\bbSig$ of $\hGamma - \betax\hgamma$. In general, $\bbSig$ can depend on the causal effect $\betax$ and the unknown quantities $\betax$, $\taup$, and $\rho_{U}$ defined in Sections \ref{subsec: DEEM p} and \ref{subsec: DEEM os}, respectively. To avoid the additional uncertainty in estimating these unknown quantities, we construct the working covariance matrix to approximate $\bbSig$ under the setting where $\betax = \taup = \rho_{U} = 0$, which corresponds to the two sample setting without the causal effect or direct effects. Then, we have $\bbSig = \bbSig_{\Gamma}$. According to the discussion in Section \ref{app: est Sig}, $\bbSig_{\Gamma}$ can be estimated by $\hSig_{\Gamma} = \hS_{\Gamma}\hR_{\rm ref}\hS_{\Gamma}$. However, $\hSig_{\Gamma}$ can be ill-conditioned due to the (near) singularity of $\hR_{\rm ref}$. To mitigate this problem, we replace $\hR_{\rm ref}$ by its regularized blockwise diagonal approximation $\hR$ defined in Section \ref{subsec: LD & V} in the main text and set $\bbV = \bbV_{d}^{1/2}\hR\bbV_{d}^{1/2}$ with $\bbV_{d} = \hS_{\Gamma}^{2}$.
    
    Next, we show that the above working covariance matrix satisfies Condition \ref{cond: matrices origin} in Section \ref{app: notations and conds} under regularity conditions. Recall that there are $K$ LD blocks, $d_{k}$ is the size of the $k$-th LD block, $\hR_{{\rm ref}, k}$ is the estimated correlation matrix for the $k$-th LD block obtained from a reference sample for $k = 1,\dots, K$, $\hR = \diag\{\hR_{1}, \dots, \hR_{K}\}$ with $\hR_{k} = \hR_{{\rm ref}, k} + c \bbI_{k}$ for some constant $c$ and $k = 1,\dots, K$. 
    Recall that $\bbV = \hS_{\Gamma}\hR\hS_{\Gamma}$ and the population counterpart of $\hS_{\Gamma}$ is
    \[
    \bbS_{\Gamma} = \diag\{\sqrt{n_{o}^{-1}\var(G_{1})^{-1}\var(Y)}, \dots, \sqrt{n_{o}^{-1}\var(G_{d})^{-1}\var(Y)}\}.
    \]
    Let $\bbV_{*} = \bbS_{\Gamma}\bbR_{*}\bbS_{\Gamma}$  where $\bbR_{*} = \diag\{\bbR_{*, 1}, \dots, \bbR_{*, K}\}$, $\bbR_{*, k} = \bbR_{{\rm pop}, k} + c\bbI_{k}$, and $\bbR_{{\rm pop}, k}$ is the population LD matrix of the $k$-th block for $k = 1,\dots, K$. Suppose $\|\bbR_{{\rm pop}, k}\| \leq C$ for some constant $C$. The condition number of $\bbR_{*, k}$ is bounded by $(C+c)/c$. Suppose $\var(Y)$ is bounded above and away from zero and $\var(G_{j})$ is bounded away from zero for $j = 1,\dots, d$. Note that $\var(G_{j})$ is uniformly bounded for $j = 1,\dots, d$ as SNPs are bounded random variables. Then, $\|\bbS_{\Gamma}\| \asymp n_{o}^{-1/2}$ according to its definition. Hence, Condition \ref{cond: matrices origin} (ii) is satisfied by the above defined $\bbV_{*}$.Because $\hS_{\Gamma}$ and $\bbS_{\Gamma}$ are diagonal matrices, the Bernstein inequality combined with the union bound \citep{wainwright2019high} can show that $\|\hS_{\Gamma} - \bbS_{\Gamma}\| / \|\bbS_{\Gamma}\| = O_{P}(\sqrt{\log d / n_{e}})$. Furthermore, it can be shown that $\max_{k}\|\hR_{k} - \bbR_{*, k}\| = O_{P}(\sqrt{\max_{k}d_{k} / n_{\rm ref}} + \sqrt{\log K / n_{\rm ref}})$ according to Theorem 1 in \cite{koltchinskii2017concentration}, where $n_{\rm ref}$ is the size of the reference sample. Suppose 
    \[
    \max\{\sqrt{\log d / n_{e}}, \sqrt{\max_{k} d_{k} / n_{\rm ref}}, \sqrt{\log K / n_{\rm ref}}\} = O(n^{- \kappa})
    \]
    for some $\kappa > 0$. Then, we have $\|\bbV^{-1} - \bbV_{*}^{-1}\| / \|\bbV^{-1}\| = O_{P}(n^{-\kappa})$ and Condition \ref{cond: matrices origin} (iv) is satisfied.

    \subsection{Adjusting for covariates}\label{app: adjust for covariate}
    Suppose the outcome model involves some covariates, i.e.,
    \begin{equation*}
      Y = \beta_{0} + X\betax + \boldsymbol{C}^{\T}\boldsymbol{\beta}_{c} + f(\bU)  + \epsilon_{Y},
    \end{equation*}
    where $\boldsymbol{C}$ is a vector of covariates and $(\bG^{\T}, \boldsymbol{C}^{\T})^{\T}\Perp \bU$. Then, we redefine $\widehat{\gamma}_{j}$ and $\widehat{\Gamma}_{j}$ as the regression coefficients of $G_{j}$ under the marginal working models 
    \[X = \gamma_{0j} + G_{j}\gamma_{j} + \boldsymbol{C}^{\T}\boldsymbol{\gamma}_{C,j} + \epsilon_{X,j}
    \]
    and 
    \[
    Y = \Gamma_{0j} + G_{j}\Gamma_{j} + \boldsymbol{C}^{\T}\boldsymbol{\Gamma}_{C,j} + \epsilon_{Y,j}
    \]
    for $j = 1,\dots, d$. Define $\gamma_{j}$ and $\Gamma_{j}$ as the probability limits of $\widehat{\gamma}_{j}$ and $\widehat{\Gamma}_{j}$, respectively. Letting $\bgamma = (\gamma_{1},\dots,\gamma_{d})^{\T}$ and $\bGamma = (\Gamma_{1},\dots,\Gamma_{d})^{\T}$, we have $\bGamma = \betax \bgamma$ and the derivations and results in this paper still hold in the presence of covariates.
    
    \section{Details of Simulation Studies in Section \ref{sec: sim} and Some Additional Simulations}
    \subsection{Generation of Genotype data in Section \ref{sec: sim}}\label{app: gen genotype}
    The genotype data used in Section \ref{sec: sim}, comprising 120,000 subjects, was adopted from a prior study from our group \citep{zhang2023new}, generated via HAPGEN2 \citep{su2011hapgen2} based on the 1000 Genomes Project (phase 3) \citep{10002015global} with a European reference panel. Given the computational intensity required to generate genotype data anew via Hapgen2, we utilized the existing 120,000 subjects dataset from the prior study as a pool for our simulation runs. From this pool, samples were drawn with replacement for each simulation.
    
    \subsection{Additional Details of the Simulation Studies in Section \ref{sec: sim}}\label{app: detail sim}
    Table \ref{table: nSNP main} shows the average number of SNPs included in each method over $200$ repetitions under the two-sample and one-sample settings in Section \ref{sec: sim} in the main text.  
    \begin{table}[H]
    	\caption{Average number of SNPs included in each method under the two-sample and one-sample settings in Section \ref{sec: sim} in the main text}\label{table: nSNP main}
    	\resizebox{\textwidth}{!}{
    		\begin{tabular}{*{8}{l}}
    			\toprule
    			Setting &
    			Method&$(0.005, 0.1)$&$(0.05, 0.1)$ & $(0.1, 0.1)$
    			&$(0.005, 0.5)$&$(0.05, 0.5)$ & $(0.1, 0.5)$\\
    			\midrule
    			\multirow{4}{*}{Two-sample}&
    			IVW(s) & 9.21& 5.45& 5.93& 18.92& 21.66& 21.88\\
    			&IVW, dIVW, pIVW, RAPS  & 27.44& 27.64& 27.95& 25.63& 26.37& 26.36\\
    			&WLR(s), PCA-IVW(s)  & 29.91& 10.17& 14.80& 204.82& 159.63& 187.05\\
    			&WLR, PCA-IVW, DEEM  & 1726.68& 1834.87& 1860.63& 2402.57& 2976.05& 3025.56\\
    			\midrule
    			\multirow{4}{*}{One-sample}&
    			IVW(s) & 9.36& 6.52& 5.10& 19.03& 21.72& 21.08\\
    			&IVW, dIVW, pIVW, RAPS  & 27.53& 27.46& 27.96& 25.59& 26.22& 26.52\\
    			&WLR(s), PCA-IVW(s)  & 30.96& 21.20& 9.52& 203.98& 193.14& 147.74\\
    			&WLR, PCA-IVW, DEEM  & 1732.61& 1848.47& 1834.90& 2411.93& 2944.24& 3017.00\\
    			\bottomrule
    		\end{tabular}
    	}
    \end{table}
    
    Figures~\ref{fig: sim ts} and \ref{fig: sim os} show the results of all methods in comparison  in Section \ref{sec: sim} in the main text and two additional methods, MR-CUE \citep{cheng2022mendelian} and GSMR \citep{zhu2018causal}, that can incorporate correlated SNPs under the two-sample and one-sample settings, respectively.

    \begin{figure}[h]
    	\centering
    	\includegraphics[scale = 0.3]{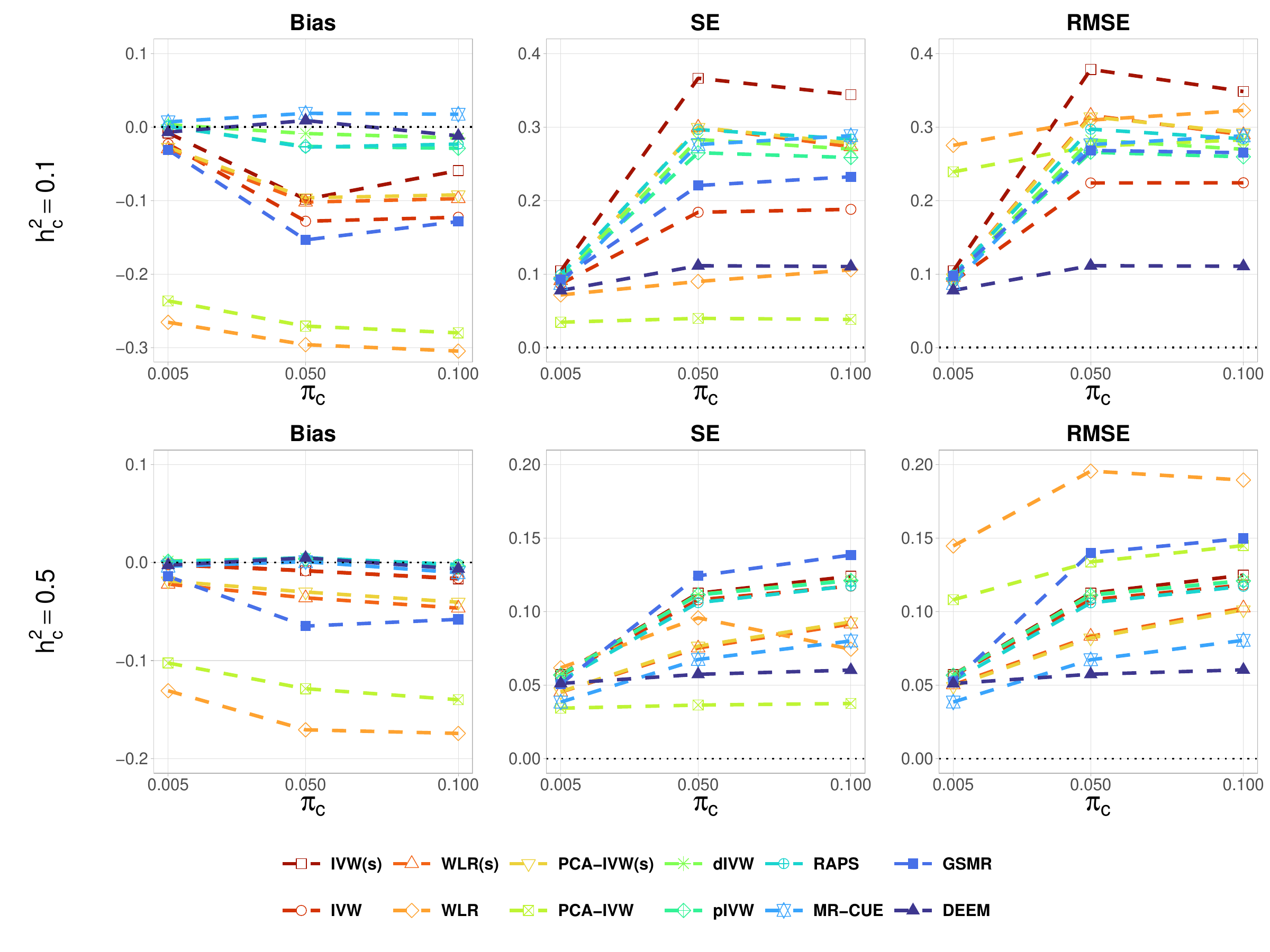}
    	\caption{Simulation results on the biases, SEs, and RMSEs of different methods in the two-sample setting with different combinations of $(\pi_{c}, h_{c}^{2})$. The ``s" in the parentheses means the stringent p-value threshold $10^{-4}$ is adopted.}\label{fig: sim ts}
    \end{figure}
    
    \begin{figure}[h]
    	\centering
    	\includegraphics[scale = 0.3]{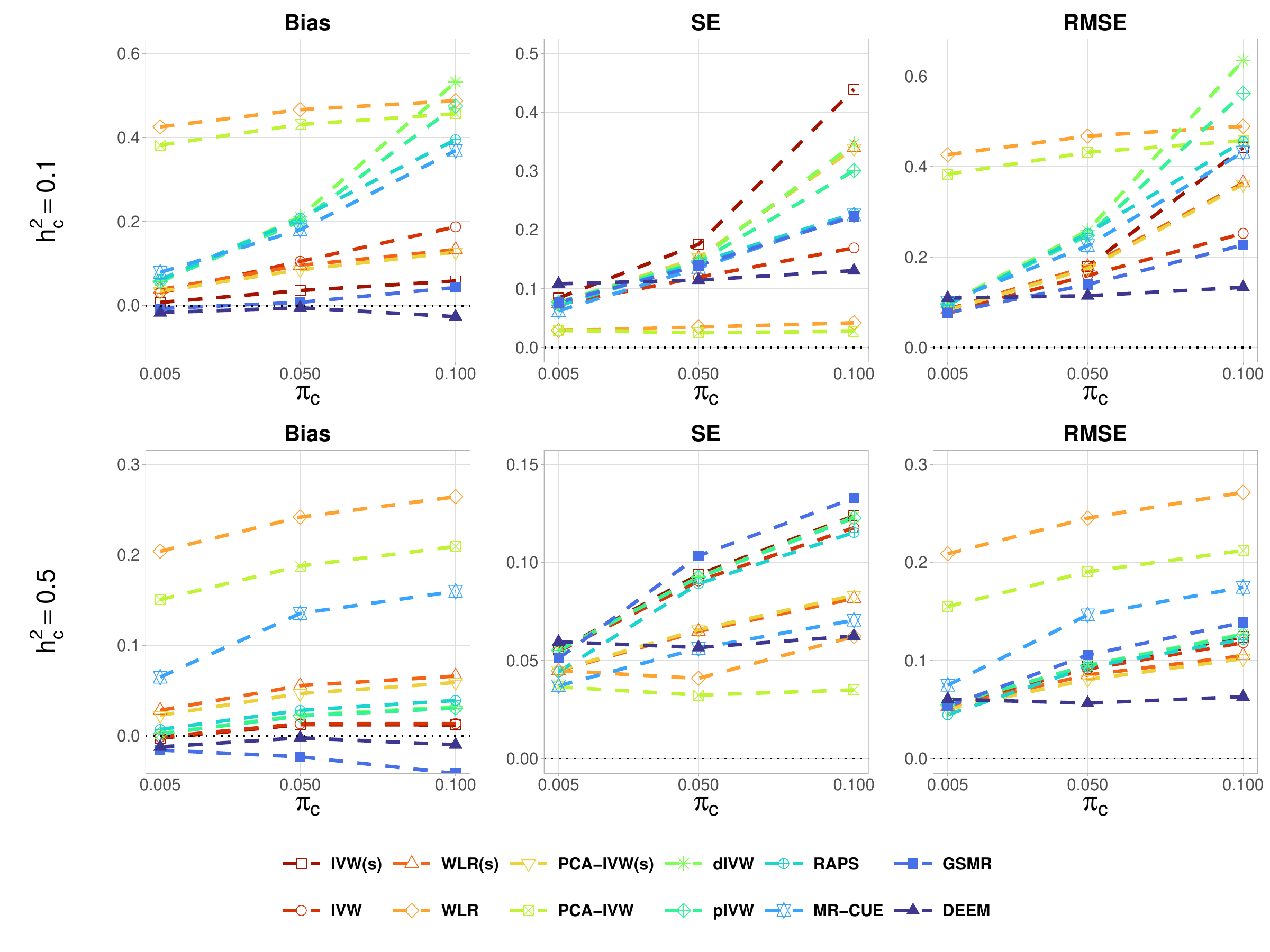}
    	\caption{Simulation results on the biases, SEs, and RMSEs of different methods in the one-sample setting with different combinations of $(\pi_{c}, h_{c}^{2})$. The ``s" in the parentheses means the stringent p-value threshold $10^{-4}$ is adopted.}\label{fig: sim os}
    \end{figure}
    
    In addition to the methods mentioned in the main text, we further compare DEEM with some robust MR methods which are designed to handle pleiotropic SNPs, including MR-APSS \citep{hu2022mendelian}, the weighted median method \citep{bowden2016consistent}, MRMix \citep{qi2019mendelian}, MR-PRESSO \citep{verbanck2018detection}, ContMix \citep{burgess2020robust}, and cML \citep{xue2021constrained}. Results are presented in Figure~\ref{fig: sim ts robust MR}. We adopt the absolute correlation coefficient threshold $0.01$ and p-value threshold $10^{-4}$ in the implementation of all robust MR methods.  
     \begin{figure}[h]
    	\centering
    	\includegraphics[scale = 0.3]{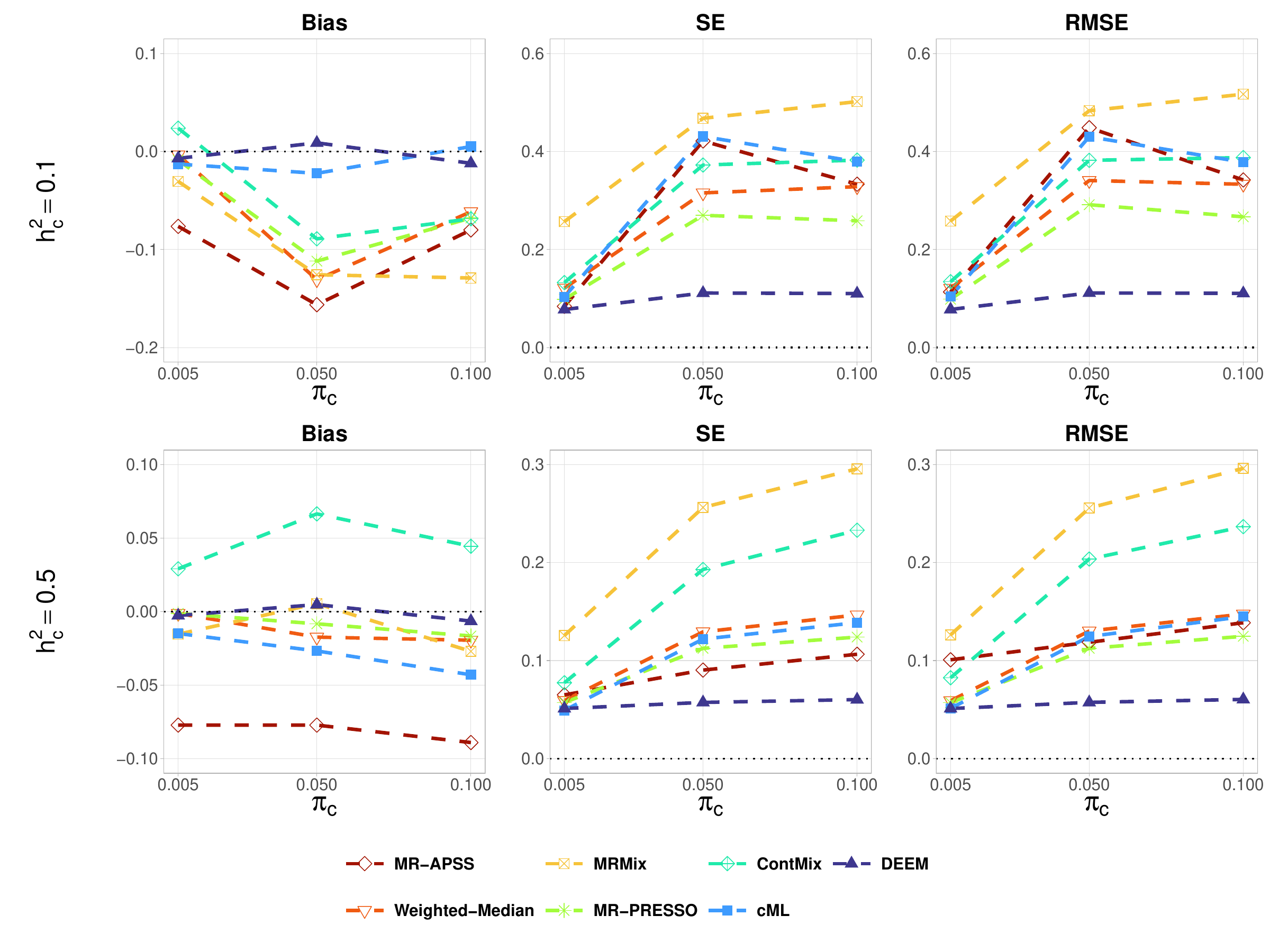}
    	\caption{Simulation results on the biases, SEs, and RMSEs of DEEM and robust MR methods in the two-sample setting with different combinations of $(\pi_{c}, h_{c}^{2})$.}\label{fig: sim ts robust MR}
    \end{figure}
    
    Figure \ref{fig: sim ts robust MR} shows that the proposed DEEM outperforms robust MR methods in terms of bias, SE, and RMSE in most settings.
    
    \subsection{Comparison between different choices of $\bbV$}\label{app: choice V}
    We demonstrate the impact of different choices of $\bbV$ on the variance of the resulting DEEM estimator through a simulation study under the two-sample setting in the main text. Two working covariance matrices are considered: the working covariance matrix $\hS_{\Gamma}\hR\hS_{\Gamma}$ proposed in Section \ref{subsec: LD & V} in the main text with $c = 1$ and the matrix $\hS_{\Gamma}^{2}$ which replaces the approximated correlation matrix $\hR$ in $\hS_{\Gamma}\hR\hS_{\Gamma}$ by an identity matrix. The two matrices are denoted by Cor and Id, respectively. Note that the correlation information among SNPs is incorporated in $\hS_{\Gamma}\hR\hS_{\Gamma}$ while ignored in $\hS_{\Gamma}^{2}$. Figure~\ref{fig: sim ts choice V} presents the biases and SEs with different choices of $\bbV$ across different combinations of $(\pi_{c}, h_{c}^{2})$ and $200$ replications.
    \begin{figure}[h]
    	\centering
    	\includegraphics[scale = 0.3]{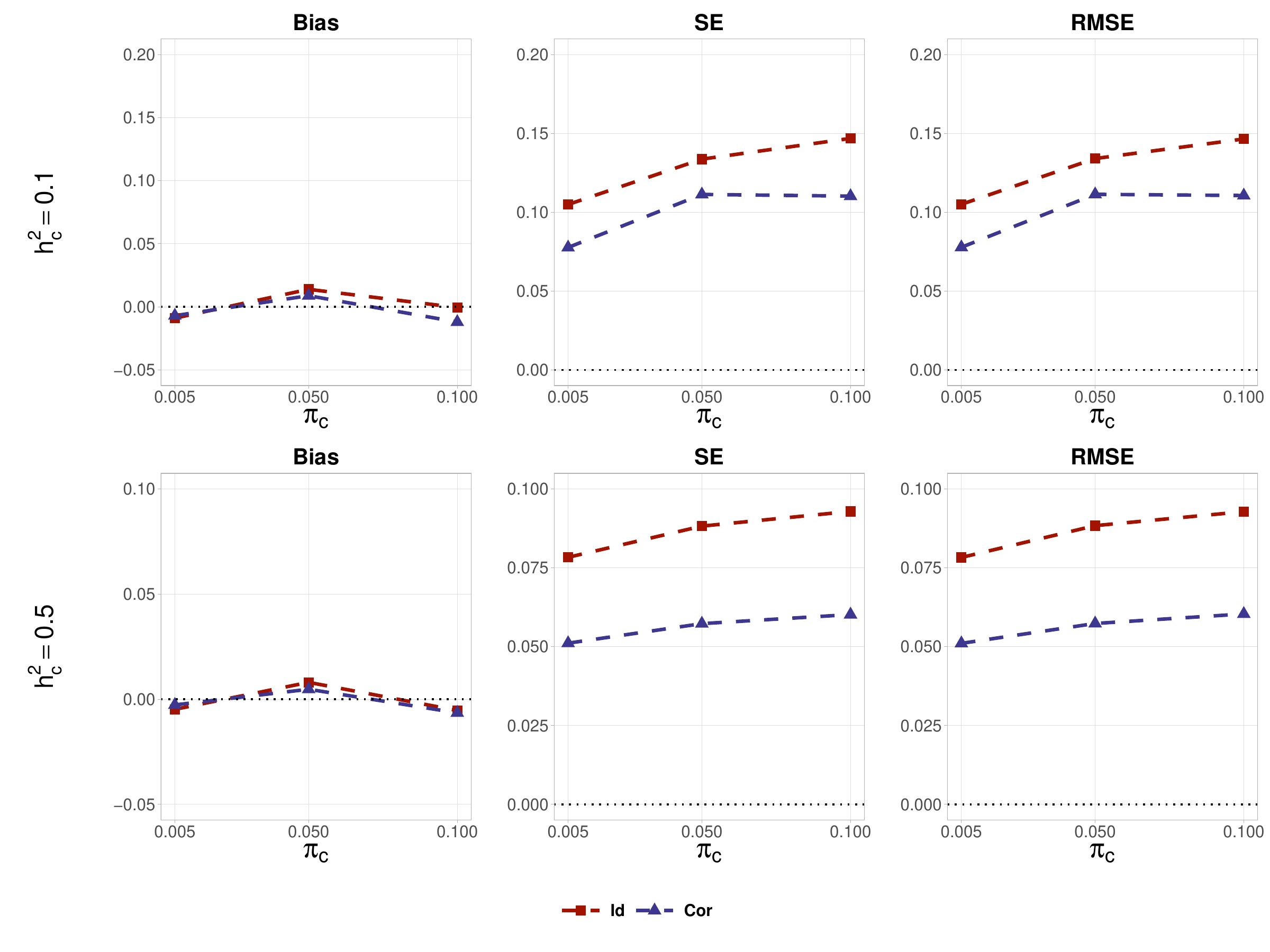}
    	\caption{Simulation results on the biases and SEs under different choices of $\bbV$ in the two-sample setting with different combinations of $(\pi_{c}, h_{c}^{2})$. $\pi_c$ denotes the proportion of SNPs that affect $X$.}\label{fig: sim ts choice V}
    \end{figure}
    The results in Figure~\ref{fig: sim ts choice V} indicate that while incorporating the correlation information among SNPs has a minor impact on the bias of the resulting estimator, it substantially reduces the SE and RMSE.
     
    \subsection{Comparison between different choices of $c$ in the construction of $\bbV$ in Section \ref{subsec: LD & V}}\label{app: choice c}
    We evaluate the performance of DEEM in terms of Bias, SE, and coverage rate under different choices of $c$ in the construction of $\bbV$ in Section \ref{subsec: LD & V}.
    The simulation is conducted under the two-sample setting in Section \ref{sec: sim} in the main text with $(\pi_{c}, h_{c}^{2}) = (0.1, 0.5)$. The SE and coverage rate of DEEM under different choices of $c$ are shown in Table \ref{table: choice c}. 
    \begin{table}[h]
    \centering
    	\caption{The bias, SE, and coverage rate of DEEM under different choices of $c$ under the two-sample setting in Section \ref{sec: sim} in the main text with $(\pi_{c}, h_{c}^{2}) = (0.1, 0.5)$. Biases and SEs are multiplied by ten}\label{table: choice c}
    		\begin{tabular}{*{5}{c}}
    			\toprule
    			c & 0.1 & 0.5 & 1 & 5\\ 
                \midrule
                Bias & -0.09 & -0.07& -0.06& -0.06
                \\
                SE & 0.56& 0.58 & 0.60& 0.68
                \\
                Coverage Rate & 93.5\% & 93.5\% & 94.0\%& 95.0\%\\
    			\bottomrule
    		\end{tabular}
    \end{table}
    Table \ref{table: choice c} shows that the bias of DEEM is small and stable across different choices of $c$. The observed coverage increases from $93.5\%$ to $95.0\%$ as $c$ increases from $0.1$ to $5$, while the empirical SE increases.
    \subsection{Comparison between the ensemble estimator and the two base estimators}\label{app: en}
    In this section, we compare the performance of the ensemble estimator $\hbeta_{\rm DEEM, p}$ with its two base estimators $\hbeta_{1, \rm p}$ and $\hbeta_{2}$ under the two-sample setting in the main text. Furthermore, the ensemble estimator $\hbeta_{\rm DEEM, os}$ is compared with its two base estimators $\hbeta_{1, \rm os}$ and $\hbeta_{2}$ under the one-sample setting in the main text. The comparative results, including the bias and SE of these different estimators, are presented in Figure \ref{fig: En}.
    
    \begin{figure}[h]
    	\centering
    	\subcaptionbox{}{\includegraphics[scale = 0.25]{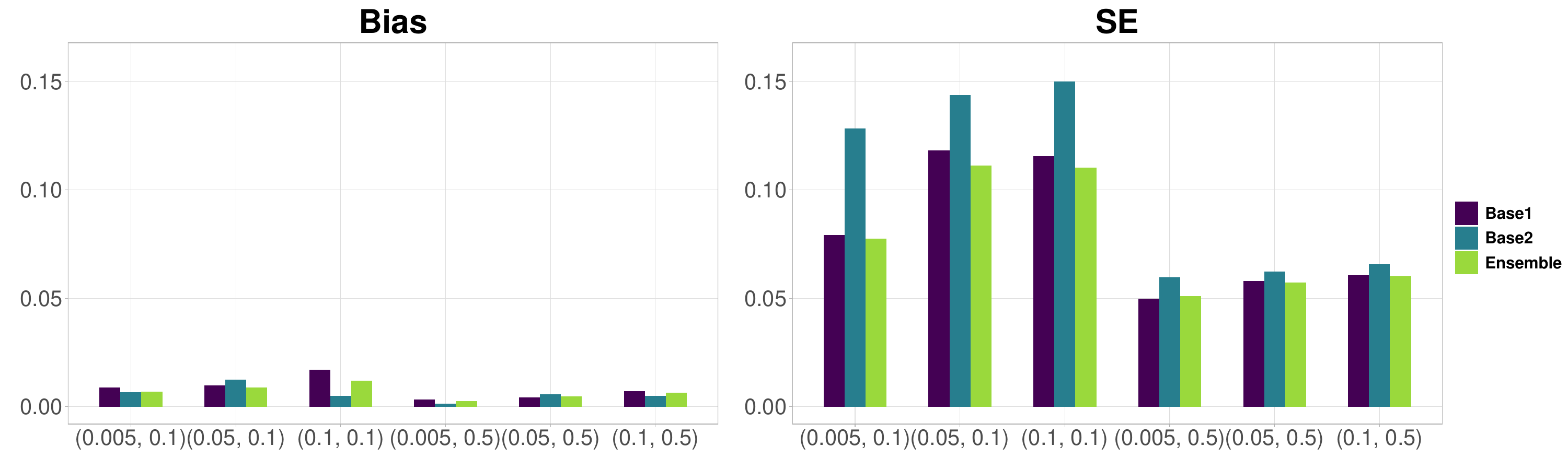}}
    	\subcaptionbox{}{\includegraphics[scale = 0.25]{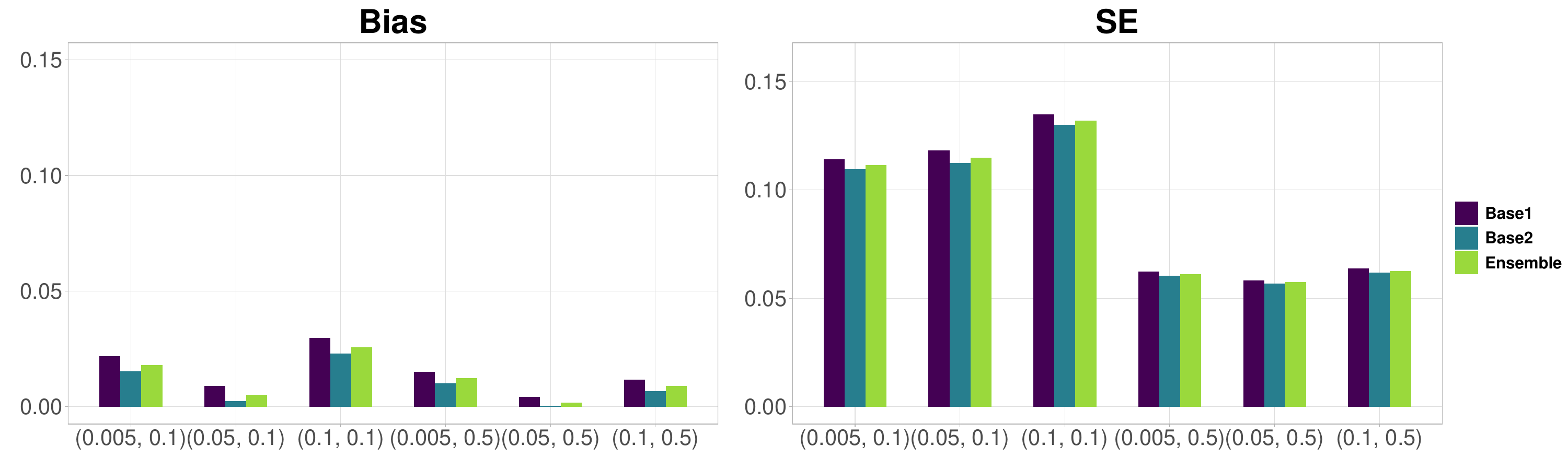}}
    	\caption{Simulation results on the bias and SE of the ensemble estimator and its two base estimators under different combinations of $(\pi_{c}, h_{c}^{2})$. Base 1 stands for $\hbeta_{1, \rm p}$ in the two-sample setting and stands for $\hbeta_{1, \rm os}$ in the one-sample setting. Base 2 stands for $\hbeta_{2}$ in both settings. (a) Bias and SE in the two-sample setting; (b) Bias and SE in the one-sample setting.}\label{fig: En}
    \end{figure} 
    
    Figure \ref{fig: En} shows that the ensemble estimator perform no worth than both of its base estimators across all different settings.
    The ensemble estimator outperforms both base estimators in terms of the bias and SE when $h_{c}^{2} = 0.1$ in the two sample setting.
    
    \subsection{Simulation with Non-normal Effects}\label{app: sim laplace}
    Here we provide some additional simulation results. The simulation settings are identical to those in the main text, except that the effects $\balpha_{G}$ and $\bbeta_{G}$ are generated from the Laplace distribution with the same variances considered in Section \ref{sec: sim}. 
    
    The results under two-sample and one-sample settings are summarized in Figure~\ref{fig: sim ts laplace} and Figure~\ref{fig: sim os laplace}, respectively.
    Performances of different methods are similar to those in the main text.
    In the two-sample setting, the coverage rates of $95\%$ CI based on DEEM are $96.5\%$, $96.5\%$, $97.0\%$, $95.0\%$, $96.0\%$, and $95.5\%$
    when $(\pi_{c}, h_{c}^{2}) = (0.005, 0.1)$, $(0.05, 0.1)$, $(0.1, 0.1)$, $(0.005, 0.5)$, $(0.05, 0.5)$, and $(0.1, 0.5)$, respectively.
    In the one-sample setting, the coverage rates of $95\%$ CI based on the proposed method are $96.0\%$, $96.5\%$, $94.5\%$, $93.5\%$, $95.0\%$, and $93.0\%$
    when $(\pi_{c}, h_{c}^{2}) = (0.005, 0.1)$, $(0.05, 0.1)$, $(0.1, 0.1)$, $(0.005, 0.5)$, $(0.05, 0.5)$, and $(0.1, 0.5)$, respectively.
    \begin{figure}[h]
    	\centering
    	\includegraphics[scale = 0.3]{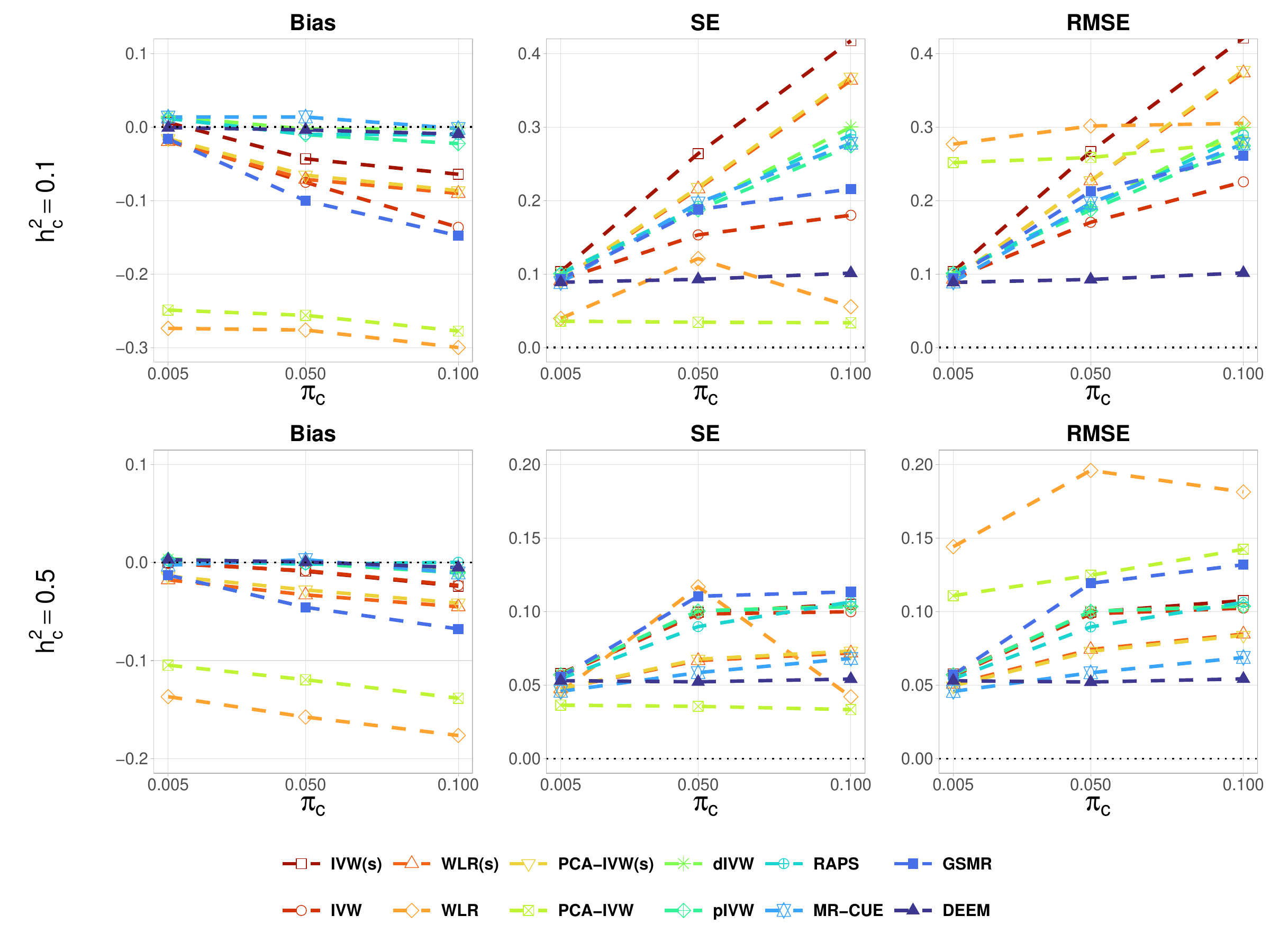}
    	\caption{Simulation results on the biases, SEs, and RMSEs of different methods in the two-sample setting with different combinations of $(\pi_{c}, h_{c}^{2})$. Coefficients generated from the Laplace distribution. The ``s" in the parentheses means the stringent p-value threshold $10^{-4}$ is adopted.}\label{fig: sim ts laplace}
    \end{figure}
    
    \begin{figure}[h]
    	\centering
    	\includegraphics[scale = 0.3]{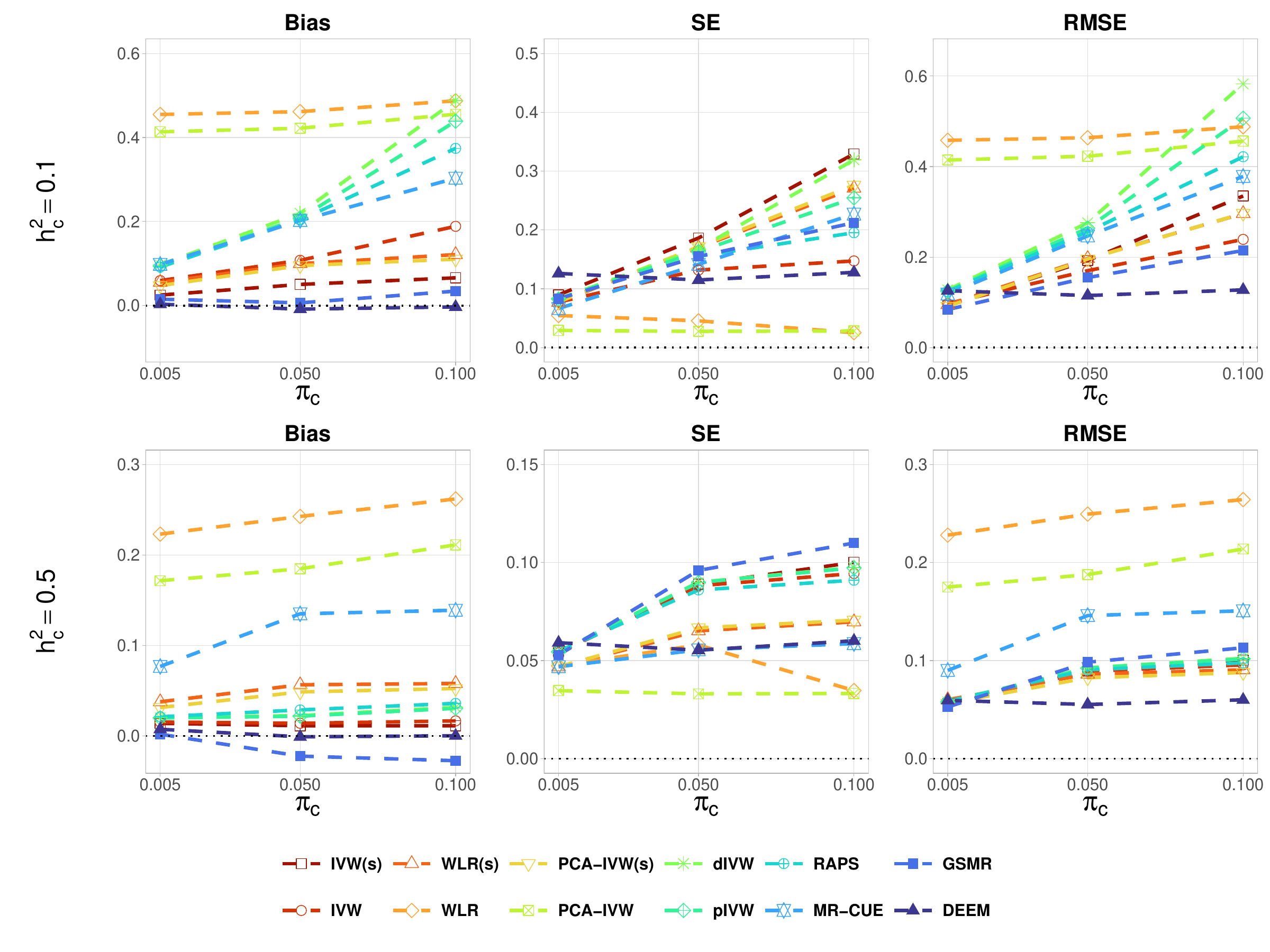}
    	\caption{Simulation results on the biases, SEs, and RMSEs of different methods in the one-sample setting with different combinations of $(\pi_{c}, h_{c}^{2})$. Coefficients generated from the Laplace distribution. The ``s" in the parentheses means the stringent p-value threshold $10^{-4}$ is adopted.}\label{fig: sim os laplace}
    \end{figure}
    
    \subsection{Correlated Pleiotropy}\label{app: sim CHP}
    To further explore the robustness of different methods, we investigate the performance of different estimators when the SNPs may affect the exposure and outcome through an unmeasured confounder, which is commonly referred to as ``correlated pleiotropy" phenomenon in the literature of MR \citep{morrison2020mendelian,cheng2022mendelian}. The IV assumptions and the pleiotropic effect model \eqref{eq: pleiotropy outcome model} are both violated in the presence of correlated pleiotropy, which may introduce biases to many MR methods.
    
    In this section, we generate $\bbG$ in the same way as that in the previous sections and generate the exposure and outcome from the data generation process
    \[
    \begin{aligned}
    	&U_{1} = \bG^{\T}\boldsymbol{\eta}_{G} + \epsilon_{U}\\
    	&X = \bG^{\T}\balpha_{G} + U_{1} + U_{2} + \epsilon_{X};\\
    	&Y =  X\betax + \bG^{\T}\bbeta_{G} + U_{1}\beta_{U} + U_{2} +  \epsilon_{Y},
    \end{aligned}
    \]
    where $\epsilon_{U}$, $\epsilon_{X}$ and $\epsilon_{Y}$ are independent and follow $N(0, 0.5)$, and $U_{2} \sim N(0, 1)$. The effects $\balpha_{G}$ and $\bbeta_{G}$ are generated in the same way as in the above sections. For $j$ with $I_{\alpha_{G,j}} = 1$, generate $I_{\eta_{G,j}} \sim {\rm Bernoulli}(0.1)$. Set $I_{\eta_{G,j}} = 0$ when $I_{\alpha_{G,j}} = 0$.
    Let
    $\eta_{G,j} = 0$ if $I_{\eta_{G}, j} = 0$, and  $\eta_{G,j} \sim N(0, 0.1 / \sum_{j=1}^{d}I_{\eta_{G,j}})$ if $I_{\eta_{G,j}} = 1$, where $\eta_{G,j}$ is the $j$th component of $\boldsymbol{\eta}_{G}$. 
    Set $\betax = 0.4$,  $\pi_{d} = 0.1$, $h_{d}^{2} = h_{c}^{2} = 0.5$, $\pi_{c} = 0.005, 0.05$ or $0.1$ and $\beta_{U} = 0, 0.25, 0.5, 0.75, 1$. The effect $\boldsymbol{\eta}_{G}$ is treated as fixed similarly to $\balpha_{G}$. In this section, the two estimates $\hgamma$ and $\hGamma$ come from two independent samples.

    Next, we evaluate the performance of different MR methods in the presence of correlated pleiotropy.
    Figure \ref{fig: CHP} shows different methods' biases, SEs and RMSEs in the presence of correlated pleiotropy across $200$ simulation runs. WLR and PCA-IVW adopt the p-value threshold $10^{-4}$ in this comparison because simulations in the main text show that these two methods are biased with the p-value threshold $10^{-1}$.
    \begin{figure}
    	\centering
    	\includegraphics[scale = 0.3]{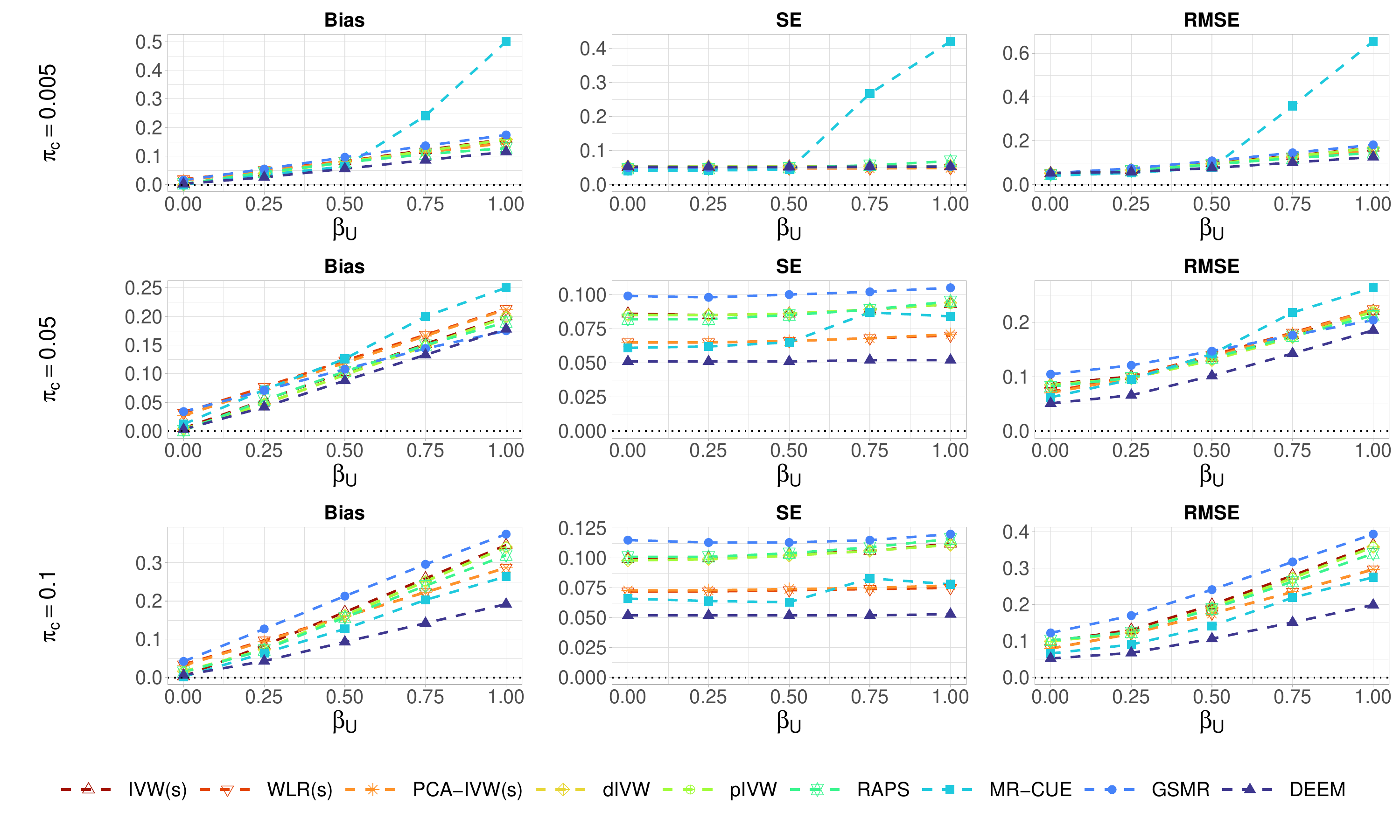}
    	\caption{Simulation results on the absolute biases, SEs, and RMSEs of different methods in the presence of pleiotropic effects through an unmeasured confounder in the two-sample setting with $h_{c}^{2} = 0.5$ and different $\pi_{c}$'s.}\label{fig: CHP}
    \end{figure}
    
    Figure \ref{fig: CHP} shows that the absolute biases of all methods increase as the value of $\beta_{U}$ grows. The proposed DEEM exhibits relatively small biases compared to other methods in most cases. This may be attributed to the fact that DEEM includes more valid instruments than other methods, which reduces the influence of the invalid IVs. As a result, DEEM appears to be less impacted by correlated pleiotropy compared to other methods. Similar to the previous simulations, DEEM tends to have higher efficiency than other methods when the effects are dense, i.e., $\pi_{c}$ is large. The RMSEs of DEEM are smaller than those of other methods in most settings.
    
    Next, we compare DEEM with robust MR methods considered in Figure~\ref{fig: sim ts robust MR} in the presence of correlated pleiotropy. Figure~\ref{fig: CHP robust MR} presents absolute biases, SEs, and RMSEs of different methods across $200$ replications.
    \begin{figure}
    	\centering
    	\includegraphics[scale = 0.3]{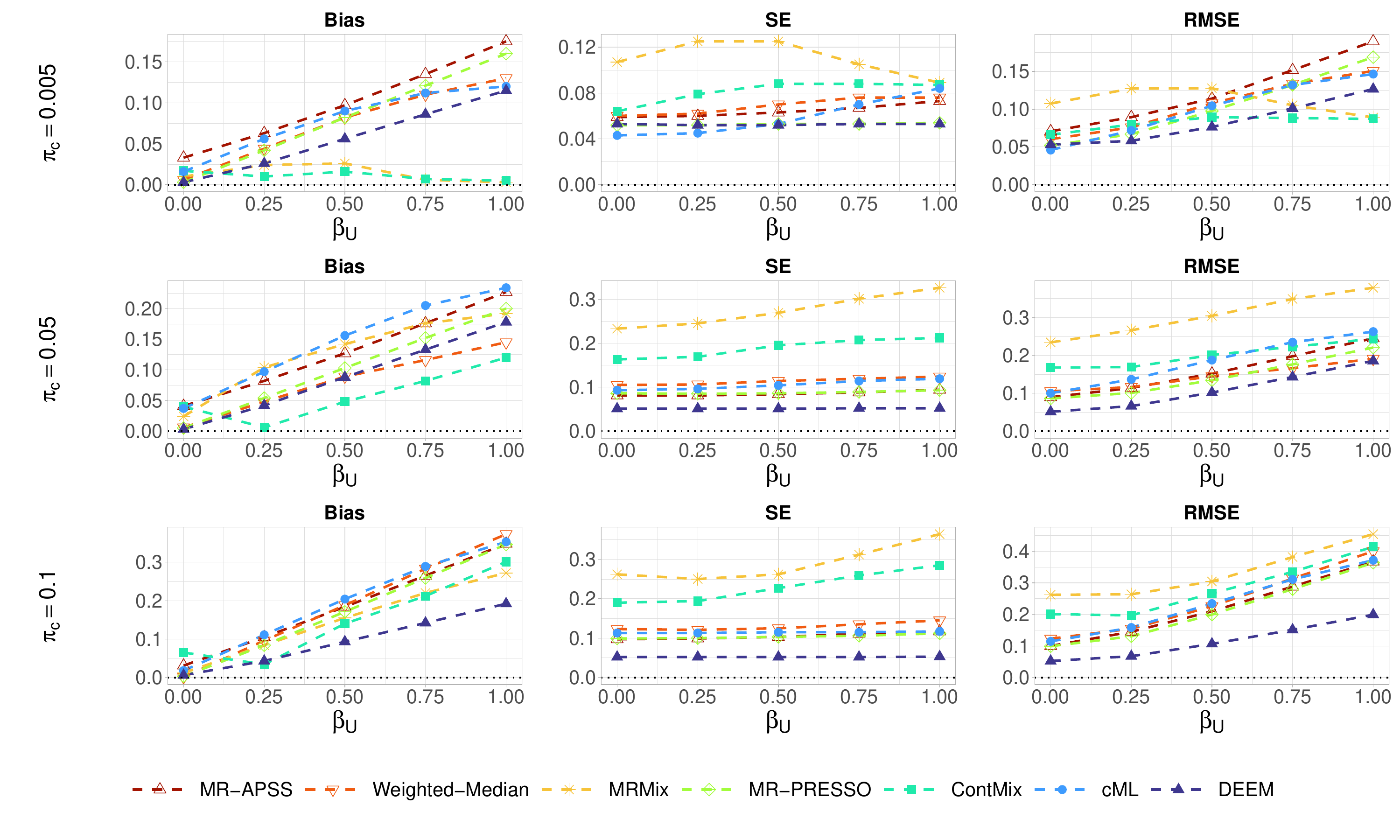}
    	\caption{Simulation results on the absolute biases, SEs, and RMSEs of DEEM and robust MR methods in the presence of pleiotropic effects through an unmeasured confounder in the two-sample setting with $h_{c}^{2} = 0.5$ and different $\pi_{c}$'s.}\label{fig: CHP robust MR}
    \end{figure}
    
    Figure~\ref{fig: CHP robust MR} demonstrates that MRMix and ContMix can mitigate bias caused by correlated pleiotropy when $\pi_{c} = 0.005$. However, they remain susceptible to bias when $\pi_{c} = 0.05$ or $0.1$. The biases of all other methods generally increase as $\beta_{U}$ grows. Notably, DEEM exhibits competitively small biases, particularly in dense-effect settings ($\pi_{c} = 0.1$). Furthermore, DEEM achieves the smallest SEs and RMSEs among all methods across most settings.
    
    \section{The Effect of Height on CAD}\label{sec: height-CAD}
    In this section, we conduct a two-sample MR analysis to study the causal effect of height on CAD risk.   The supplemental exposure ($X$) dataset is the GWAS for height from the UK biobank; the exposure ($X$) dataset is the GWAS for height conducted by the GIANT Consortium; the outcome ($Y$) dataset is the GWAS for CAD risk conducted by the CARDIoGRAMplusC4D Consortium. Details about the employed GWASs are in Section \ref{app: package and data}.
    Implementation details are the same as those in Section \ref{subsec: LDL-CAD}.
    Figure~\ref{fig: height-CAD} summarized the findings.
    
    \begin{figure}
    	\centering
    	\subcaptionbox*{}{\includegraphics[scale = 0.5]{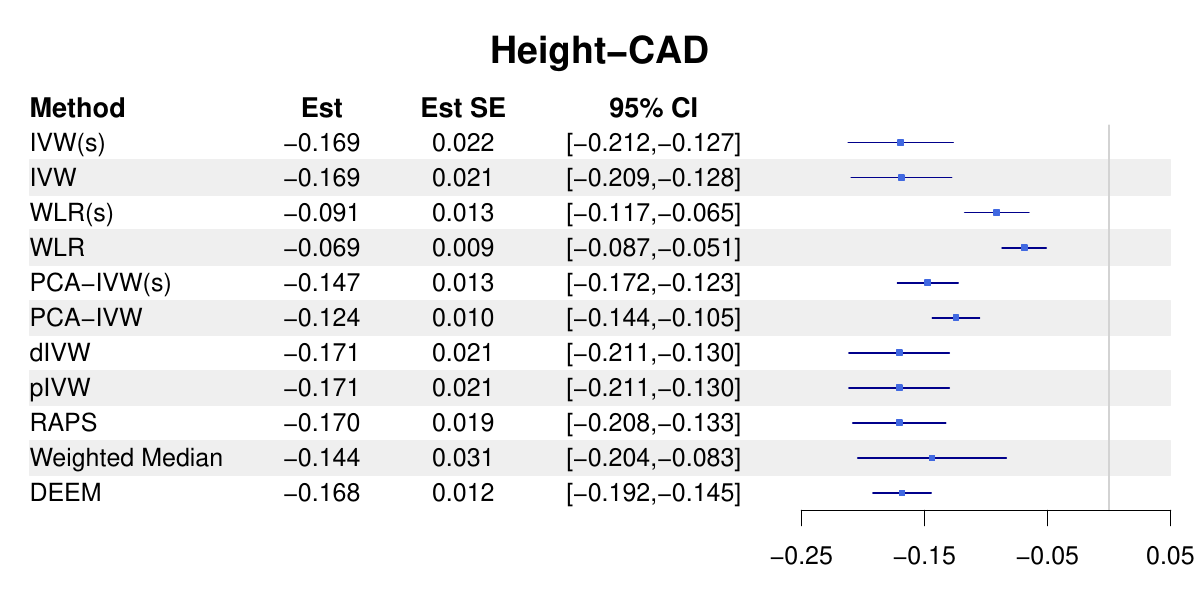}}
    	\caption{MR analysis of the effect of height on the CAD risk. Est: estimated causal effect; Est SE: estimated standard error; CI: confidence interval. The ``s" in the parentheses indicates the method is implemented with the stringent p-value threshold $10^{-4}$.}\label{fig: height-CAD}
    \end{figure}
    
    All methods indicate a negative causal effect of height on CAD risk. WLR(s), WLR, and PCA-IVW yield smaller estimates for the absolute effect size than others, potentially due to the weak IV bias. DEEM, on the other hand, yields a point estimate similar to those from IVW(s), IVW, dIVW, and RAPS, while exhibiting the smallest estimated standard error among all methods except the potentially biased one noted above. This demonstrates DEEM's advantage in efficiency through the inclusion of correlated weak IVs. 
    
    \section{Packages and Data Availability}\label{app: package and data}
    We developed an R package \texttt{MR.DEEM} that implements the proposed DEEM in both two-sample and one-sample settings, which can be installed from github (\url{https://github.com/ryw-stats/MR.DEEM}). The installation instructions and the help document can be found on the above github website. One can use data from 1000 Genomes Project as the reference panel to calculate the LD matrix based on a blockwise-diagonal structure. The LD block information as obtained by \cite{berisa2016approximately} can be downloaded from \url{https://bitbucket.org/nygcresearch/ldetect-data/src/master/}. Genotype data from the 1000 Genomes Project (Phase 3) are available from the plink website \url{https://www.cog-genomics.org/plink/2.0/resources}.
    
    The R packages to implement RAPS and pIVW are available at Github (\url{https://github.com/qingyuanzhao/mr.raps};\url{https://github.com/siqixu/mr.pivw}). The R packages \texttt{MendelianRandomization}, \texttt{MRPRESSO}, \texttt{https://github.com/gqi/MRMix}, and \texttt{MRAPSS} used in this paper to implement benchmark MR methods in comparison
    are available at \url{https://cran.r-project.org/web/packages/MendelianRandomization/index.html}, \url{https://github.com/rondolab/MR-PRESSO}, \url{https://github.com/gqi/MRMix}, and \url{https://github.com/YangLabHKUST/MR-APSS}, respectively.
    
    The genotype data in the simulation study are available at \url{https://dataverse.harvard.edu/dataset.xhtml?persistentId=doi:10.7910/DVN/COXHAP}. 

    In Section \ref{subsec: LDL-CAD}, the supplemental exposure dataset is a GWAS for LDL-C conducted by the Global Lipids Genetics Consortium \citep{global2013discovery} (sample size: $1.89 \times 10^{5}$, unit: mmol/L, available at \url{http://csg.sph.umich.edu/willer/public/lipids2013}); the exposure dataset is a GWAS for LDL-C from the UK biobank (sample size: $4.00\times 10^{5}$, unit: mmol/L, available at \url{https://pan.ukbb.broadinstitute.org/downloads}); the outcome data\-set is a GWAS for CAD risk conducted by the CARDIoGRAMplusC4D Consortium \citep{schunkert2011large} (sample size: $1.85\times 10^{5}$, available at \url{http://www.cardiogramplusc4d.org/data-downloads}).
    In Section \ref{subsec: BMI-SBP},
    the supplemental dataset is a GWAS for BMI conducted by the GIANT Consortium \citep{locke2015genetic} (sample size: $3.39 \times 10^{5}$, unit: kg / m$^2$, available at \url{https://portals.broadinstitute.org/collaboration/giant/index.php/GIANT_consortium_data_files#GIANT_Consortium_2012-2015_GWAS_Summary_Statistics}); the exposure dataset is a GWAS for BMI from the UK biobank (sample size: $3.36 \times 10^{5}$, unit: kg / m$^2$, available at \url{https://pan.ukbb.broadinstitute.org/downloads}); the outcome dataset is a GWAS for SBP from the UK biobank  (sample size: $3.18 \times 10^{5}$, unit: mmHg, available at \url{https://pan.ukbb.broadinstitute.org/downloads}). In Section \ref{sec: height-CAD}, the supplemental exposure dataset is a GWAS for height from the UK biobank (sample size: $4.00\times 10^{5}$, unit: cm, available at \url{https://pan.ukbb.broadinstitute.org/downloads}); the exposure dataset is a GWAS for height conducted by the GIANT Consortium (sample size: $10^{6}$, unit: cm, available at \url{https://giant-consortium.web.broadinstitute.org/index.php/GIANT_consortium_data_files#Height_GWAS_Summary_Statistics}); the outcome dataset is a GWAS for CAD risk conducted by the CARDIoGRAMplusC4D Consortium \citep{schunkert2011large} (sample size: $1.85\times 10^{5}$, available at \url{http://www.cardiogramplusc4d.org/data-downloads}).
    
    All the codes and data required to produce the results in this paper are available at
    \url{https://osf.io/j25rc/files/osfstorage}.
    
    \section*{Acknowledgement}
    We gratefully acknowledge UK biobank participants for their contributions, without whom this research would not have been possible. We also thank the Pan-UKB project and UK Biobank for making available the summary data examined in this study.

\bibliographystyle{asa}
\bibliography{DEEM-arXiv.bib}
\end{document}